\documentclass[11pt]{article}

\RequirePackage{graphicx}
\RequirePackage{amsthm}
\RequirePackage[cmex10]{amsmath}
\RequirePackage[colorlinks,citecolor=black,urlcolor=black]{hyperref}

\usepackage{natbib}
\setcitestyle{authoryear,open={(},close={)}}

\newtheorem{lemma}{Lemma}

\newtheorem{theorem}{Theorem}

\newtheorem{corollary}{Corollary}

\newtheorem{proposition}{Proposition}

\theoremstyle{definition}
\newtheorem{example}{Example}

\theoremstyle{remark}
\newtheorem{claim}{Claim}

\usepackage[top=1.25in, bottom=1.25in, left=1.25in, right=1.25in]{geometry}

\usepackage{txfonts}
\usepackage{amsfonts}
\usepackage{times} 
\usepackage{epsfig,setspace}
\usepackage{geometry, wasysym,enumerate}
\usepackage{sgame}
\usepackage{subcaption}
\usepackage{comment}
\usepackage[normalem]{ulem}
\usepackage{amsmath,amssymb,amsthm}
\usepackage{setspace,graphicx}
\usepackage{lmodern}
\usepackage{slantsc}%
\usepackage[cmtip,all]{xy}
\usepackage{accents}
\usepackage[cmyk]{xcolor}

\usepackage{pgfplots}
\pgfplotsset{compat=1.16}

\usepgfplotslibrary{fillbetween}
\usetikzlibrary{patterns}
\usetikzlibrary{arrows}
\usetikzlibrary{shapes,decorations} 

\usepackage{ifthen}
\usepackage{xfrac}
\usepackage{epstopdf}
\usepackage{times} 
\usepackage{titlesec}

\usetikzlibrary{calc,arrows,automata,shapes.misc,shapes.arrows,chains,matrix,positioning,scopes,decorations.pathmorphing,shadows}

\newcommand{\citeapos}[1]{\citeauthor{#1}'s (\citeyear{#1})}

\newcommand{\co}{\mathrm{co}}

\newcommand{\ext}{\mathrm{ext}}

\newcommand{\dd}{\mathrm{d}}
\newcommand{\supp}{\mathrm{supp}}
\newcommand{\bbE}{\mathbb E}

\newcommand{\real}{\mathbb{R}}

\newcommand{\longsquiggly}{\xymatrix{{}\ar@{~>}[r]&{}}}

\newcommand{\cl}{\mathrm{cl}}

\newcommand{\icost}{C}
\newcommand{\func}{f}

\newcommand{\wt}{t}
\newcommand{\wtb}{s}
\newcommand{\const}{\lambda}
\newcommand{\constb}{\gamma}
\newcommand{\constc}{\eta}
\newcommand{\belief}{\mu}

\newcommand{\Paystate}{\Theta}
\newcommand{\paystate}{\theta}
\newcommand{\paystateb}{\theta'}
\newcommand{\paystatec}{\hat\theta}
\newcommand{\payprior}{\pi}
\newcommand{\Player}{I}

\newcommand{\Playernum}{n}
\newcommand{\player}{i}
\newcommand{\playerb}{j}

\newcommand{\util}{u}
\newcommand{\utilb}{u^\prime}
\newcommand{\Act}{A}
\newcommand{\act}{a}
\newcommand{\actb}{b}
\newcommand{\actc}{c}
\newcommand{\mact}{\alpha}
\newcommand{\mactb}{\beta}
\newcommand{\aplan}{\sigma}

\newcommand{\Aplan}{\Sigma}
\newcommand{\BGame}{\mathcal{G}}

\newcommand{\out}{p}
\newcommand{\outb}{q}
\newcommand{\outc}{r}
\newcommand{\Out}{\Delta_\payprior(\Act\times\Paystate)}
\newcommand{\payvector}{v}

\newcommand{\grossval}{\bar{\payvector}}
\newcommand{\noinfoval}{\underline{\payvector}}

\newcommand{\actval}{U}

\newcommand{\Signal}{X}
\newcommand{\signal}{x}
\newcommand{\signalb}{x^\prime}
\newcommand{\Corstate}{Z}
\newcommand{\corstate}{z}
\newcommand{\corprior}{\zeta}
\newcommand{\itech}{\mathcal{T}}
\newcommand{\IT}{\mathcal{T}}
\newcommand{\Exper}{\mathcal{E}}
\newcommand{\exper}{\xi}
\newcommand{\experb}{\xi'}
\newcommand{\experc}{\xi''}
\newcommand{\garb}{g}

\newcommand{\cord}{\chi}
\newcommand{\PartA}{\mathcal{A}}

\newcommand{\BCE}{BCE}
\newcommand{\sBCE}{sBCE}

\newcommand{\BCEset}{P}

\newcommand{\NE}{\text{NE}}

\newcommand{\Mact}{\mathcal{M}}

\newcommand{\BR}{BR}
\newcommand{\Jeop}{J}

\newcommand{\welfex}{\bar{w}}
\newcommand{\Welfex}{\bar{w}}
\newcommand{\welfen}{\underline{w}}
\newcommand{\Welfen}{\underline{w}}

\newcommand{\perm}{\phi}

\newcommand{\Outsym}{\Delta^{sy}_{\payprior}(\Act\times\Paystate)}
\newcommand{\noinfoBCE}{\BCE_0}
\newcommand{\BCEsym}{\BCE^{sy}}

\newcommand{\ppcdf}{F}
\newcommand{\brsout}{Q}
\newcommand{\brcost}{k}
\newcommand{\brext}{x}

\newcommand{\maxpaystate}{\bar{\paystate}}
\newcommand{\minpaystate}{\underline{\paystate}}

\newcommand{\RIval}{V_{R}}
\newcommand{\Exval}{V_{I}}

\newcommand{\acomment}[1]{}

\makeatletter
  \renewcommand\@seccntformat[1]{\csname the#1\endcsname.{\hskip.7em\relax}} 
\makeatother

\DeclareMathOperator*{\argmax}{argmax}
\DeclareMathOperator*{\argmin}{argmin}

\title{{\bf Robust Predictions in Games with \\ Rational Inattention}\footnote{We thank Alex Bloedel, Benjamin Brooks, Modibo Camara, Marina Halac, Emir Kamenica, Ilia Krasikov, Elliot Lipnowski, Andrew McClellan, Stephen Morris, Luciano Pomatto,  Daniel Rappoport, Phil Reny, and Joe Root for helpful discussions.}

\author{
\begin{minipage}{0.3\textwidth}\centering  
Tommaso Denti\footnote{\texttt{tjd237@cornell.edu}} \\ \centering \it \small Cornell University
\end{minipage}                  
\begin{minipage}{0.3\textwidth}\centering 
Doron Ravid\footnote{\texttt{dravid@uchicago.edu}}  \\ \centering \it \small University of Chicago 
\end{minipage} 
}

\date{\vspace{0.8cm} \today}

} 

\begin{document}

\maketitle

\begin{spacing}{1}
\onehalfspacing

\begin{abstract}
We derive robust predictions in games involving flexible information acquisition, also known as \emph{rational inattention} \citep{Sims2003a}. These predictions remain accurate regardless of the specific methods players employ to gather information. Compared to scenarios where information is predetermined, rational inattention reduces welfare and introduces additional constraints on behavior. We show these constraints generically do not bind; the two knowledge regimes are behaviorally indistinguishable in most environments. Yet, we demonstrate the welfare difference they generate is substantial: optimal policy depends on whether one assumes information is given or acquired. We provide the necessary tools for policy analysis in this context.
\end{abstract}

\section{Introduction}

The economics discipline has long recognized that information matters for incentives.
Many studies have also noted the reverse, namely, that incentives shape information. The theory of rational inattention, initiated by \cite{Sims2003a}, is a case in point.
Motivated by people's limited cognition, this theory postulates that agents pay attention as if they flexibly acquire information in an optimal fashion, trading off costs and benefits.
This approach has been very successful, seeing a wide range of applications all across economics.\footnote{See, e.g., \citet*{mackowiak2023survey} for a recent review.} There is, however, a caveat: rational inattention models often require one to commit to difficult-to-test details of the information acquisition environment. This paper addresses this concern by developing a framework for making predictions under rational inattention that are robust to the exact structure of agents' learning technologies.

For concreteness, consider the following example.

\begin{example}\label{exa:intro}
Two investors, Ann and Bob, are choosing whether or not to invest in a new project, and if so, which. There are two new projects: project $A$ and project $B$. Each investor can either invest in one of these projects, or put their funds in the market. There are also two equally likely payoff states, $\paystate_A$ and $\paystate_B$, that describe which project is better. 
If both investors fund project $t \in \{A,B\}$, they both get a payoff of $2$ if $t$ is the better project (i.e., if the state is $\paystate_t$), and payoff of $1$ in the alternative state. If the two investors fund different projects, or one of the investors chooses to put their funds in the market, they both receive a payoff of $0$ regardless of the state. In this scenario, the only impediment for Ann and Bob to outperform the market is miscoordination. 
\end{example}

The above is an example of (what is commonly called) a \emph{base game}. A base game describes an economic environment of interest: the state of fundamentals and its prior distribution, agents in the role of players, actions, and preferences. Given a base game, we aim to make predictions about behavior and welfare.

Standard rational inattention models obtain predictions by combining a base game with what we term an \emph{information technology}, which specifies players' learning capabilities regarding the true payoff state.
In early studies, this technology assumes that information costs are proportional to entropy reduction \citep[e.g.,][]{Matejka2015a}, and that players can flexibly gather information, provided that it is independent of the other agents' information
\citep[e.g.,][]{mackowiak2009optimal, Yang2015coordination}. Recent papers consider more general information technologies, allowing for different cost functions and potentially correlated signals.\footnote{See, e.g., \citet*{morris2022coordination,ravid2022learning,denti2023unrestricted,hebert2023information}}

Together, a base game and an information technology  define \emph{an information acquisition game}. This game begins with the players covertly choosing what costly information to acquire. Subsequently, each player privately uses the acquired knowledge to take an action. Predictions are derived by solving for Nash equilibrium. 

We depart from the standard approach, and do not pair the base game with a fixed information technology. Instead, we simultaneously find all the predictions one can obtain with a technology that is consistent with rational inattention.
Specifically, we consider the
information technologies that satisfy three properties. 
First, no information comes at zero cost. Second, (strictly) more informative signals (in the sense of \citealt{Blackwell1951a,Blackwell1953}) cost (strictly) more. And third, information choice is flexible. These properties are satisfied by virtually all applications of rational inattention. We are agnostic whether signals are independent or correlated across players.

Theorem~\ref{thm:mon_tech} characterizes all behavioral and welfare implications of rational inattention in a given base game. Behavior is summarized by the joint distribution of the actions taken by the players and the payoff state drawn by nature; we call such distribution the \emph{outcome} of the game. Welfare is given by the \emph{value} each player obtains from the game, that is, each player's payoff, net any costs they pay for information. 

We show rational inattention can generate any outcome that satisfies two constraints. To describe these constraints, consider the classical metaphor that views outcomes as being generated by a mediator who provides players with private but potentially correlated action recommendations after viewing the state. 
The first constraint rational inattention generates is \emph{obedience},
which means no player can benefit by deviating from the mediator's recommendations. This constraint represents optimal behavior given a fixed signal structure. We follow the literature \citep{bergemann2013robust,bergemann2016bayes}, and use the term \emph{Bayes correlated equilibrium} (BCE) to refer to obedient outcomes.\footnote{\cite{forges1993five} refers to such outcomes as \emph{universal Bayesian solutions}.}

The second constraint is \emph{separation}: recommendations that result in different posterior beliefs about others' behavior and the state cannot share an optimal reply. In other words, distinct beliefs have separate best responses. The separation constraint originates from \cite{denti2021costly}, and stems from players' information acquisition incentives. Intuitively,   outcomes that violate separation involve the acquisition of non-instrumental information, a sub-optimal choice: by not learning this information, players can save on costs without impacting their gross payoffs from the base game.
Thus, Theorem~\ref{thm:mon_tech} says an outcome can be generated by rational inattention if and only if it is a separated BCE (sBCE).

Theorem~\ref{thm:mon_tech} also provides tight bounds on the net payoff each player can get from a given equilibrium outcome. One bound is the utility the player expects 
from the outcome; that is, the gross payoff the player obtains by always following the mediator's recommendation. We refer to this quantity as the outcome's \emph{gross value}. The second bound is the outcome's \emph{uninformed value}, which is given by the highest payoff the player can get by unilaterally deviating from all of the mediator's recommendation to a fixed, deterministic action. Note that for obedient outcomes, this deviation must be unprofitable, meaning the uninformed value is lower than the gross value. Theorem~\ref{thm:mon_tech} shows that, fixing a separated BCE, each player can attain any payoff that is above her uninformed value, and below her gross value.

Using Theorem~\ref{thm:mon_tech} we study the impact information costs have on welfare, and whether rationally inattentive agents naturally seek special kinds of information. For this purpose, we compare the predictions made by rational inattention to those generated by models where information is an exogenous variable. Under exogenous information, players do not choose what they know; instead, they are endowed with a fixed signal structure.  \cite{bergemann2013robust,bergemann2016bayes} show that an outcome can be induced by some signal structure if and only if the outcome is a BCE (i.e., it is obedient). Welfare is given by the outcome's gross payoff, since no learning costs are paid when information is predetermined. Thus, Theorem~\ref{thm:mon_tech} implies that, compared to exogenous information, rational inattention generates fewer outcomes, and lower payoffs per outcome.

In some cases, the separation constraint can dramatically shrink the BCE set. For a demonstration, consider Example~\ref{exa:intro}. \addtocounter{example}{-1}

\begin{example}[Continued]
Any outcome in which the two investors perfectly coordinate their decision is obedient: each investor is always happy to take the same action as the other investor. However, among these outcomes, the only ones that are consistent with rational inattention and satisfy the separation constraint are those that assign a probability of either zero or one to the event where both investors put their money in the market. 

For an explanation, note that for perfectly-coordinated outcomes, a ``market" recommendation and a ``project  $t$" recommendation yield different posterior beliefs regarding the other agent's choice: the former recommendation implies the other investor is not going to fund any project, whereas the latter recommendation means they will fund project $t$. 
Notice though, that funding $t$ is a best reply for both beliefs, because funding some project weakly dominates investing in the market. Hence, the two recommendations share a best response, but lead to different beliefs. It follows that any perfect-coordination outcome where sometimes neither investor funds any project and sometimes both investors choose the same project violates separation, and so is inconsistent with rational inattention.

In fact, one can show a stronger result: the only sBCE in which some investor forgoes funding any project with positive probability is the BCE in which both investors always keep their money in the market. Thus, the sBCE set is \emph{nowhere dense} in the BCE set. 
\end{example}

Turns out, however, that situations like Example~\ref{exa:intro}, where the separation constraint binds, are rare. Specifically, we show in Theorem~\ref{thm:genericity} that, for generic preferences, one can approximate every obedient outcome with outcomes that are both obedient and separated. In other words, perturbing the payoffs of any base game with a large set of non-separated BCEs turns it into a game where essentially all BCEs are separated. Thus, unless one commits to a highly structured economic setting or makes specialized assumptions on agents' information technologies, rational inattention is behaviorally indistinguishable from exogenous information.

Nevertheless, even when rational inattention and exogenous information yield the same behavioral predictions, they may have very different welfare implications. Intuitively, when information is exogenous, the value of information must go to the players. By contrast, in the worst case under rational inattention, players must pay this value to purchase their information. Theorem~\ref{thm:robust difference in welfare} shows information has positive value in the worst-case outcome under rational inattention for a generic set of base games. Thus, there is a robust class of environments where players' worst-case welfare under rational inattention is significantly lower than their minimal payoffs under exogenous information.

For a heuristic treatment of our results on genericity, consider again Example \ref{exa:intro}.\addtocounter{example}{-1}

\begin{example}[Continued] 
Suppose we perturb the economy so that the market's returns are reduced by $\epsilon>0$. Such a perturbation makes investing in the market strictly dominated. As a result, the game has no BCE in which the investors keep their money in the market.  Thus, consistent with Theorem~\ref{thm:genericity}, the perturbation pushes all of the non-separated BCEs we described above out of the BCE set.\footnote{The reader should not infer from Example~\ref{exa:intro} that the fragility of weakly dominated actions is the driving force behind Theorem~\ref{thm:genericity}; see Online Appendix \ref{sec:An Example Where the Separation constraint bind}.}

Next, we explain that for all $\epsilon>0$, the worst-case welfare for the investors is lower under rational inattention than under exogenous information. To find the worst case under rational inattention, one needs to find the sBCE that minimizes the uninformed value. The unique minimizer turns out to be the outcome where the two investors always coordinate on the inferior project. Whereas the gross payoff of this BCE is $1$, its uninformed value is $0.5$: the best an investor can do without seeing the mediator's recommendation is to blindly invest in one of the projects (e.g., project $A$), and so match the other investor's decision only when that project is the inferior one. By contrast, if the gross value were to determine players' payoffs, the worst BCE for the investors gives them a payoff of $0.6$.\footnote{The BCE that minimizes the gross value in the perturbed example is the one where, conditional on the state, both investors fund the inferior project with probability $0.6$. In the complementary event, the two investors miscoordinate in a symmetric fashion: both the event in which Ann invests in project $A$ and Bob invests in project $B$, and the event with the opposite investment pattern occur with probability $0.2$.} 
\end{example}

The above-mentioned differences in welfare can impact optimal policy. We give a proof of concept in Section~\ref{sec:welfare analysis}. We consider a utilitarian social planner who understands the economic environment (i.e., the base game), but does not know the source of players' information (i.e., the information technology). The planner takes a worst-cast approach, as is common in the robust mechanism and contract design literature.\footnote{See, e.g., \citet{carroll2019RobustnessSurvey} for a recent survey of this literature.}
We characterize the set of binary-action symmetric games in which this planner's welfare evaluations depend on whether she assumes players' information is given or acquired. We then apply this characterization to a regime change game, in which several investors choose whether or not to attack a distressed financial institution. The institution fails if the number of attackers surpasses a threshold, which represents the institution's fundamentals. We show that the distribution of the fundamentals impacts the worst-case welfare under exogenous information, but has no effect under rational inattention. Thus, whereas the planner may choose to bolster the institution's fundamentals if she thinks information is exogenous, she never does so if she believes players are rationally inattentive. 

We conclude the paper with two additional inquiries. The first is the study of the non-generic environments for which the separation constraint is binding. We obtain a tight characterization of these base games (Proposition~\ref{pro:BCE equals cl sBCE iff stuff}), and show that separation is an all-or-nothing refinement of BCE: the sBCE set is either dense or nowhere dense in the set of BCEs (Theorem~\ref{thm:Dense or nowhere dense}). The second line of inquiry characterizes the outcomes attainable under rational inattention as learning costs vanish. Such cheap-learning limits have been suggested in the literature as an equilibrium selection device for games with a commonly known state.\footnote{See, e.g., \citet{Yang2015coordination,Hoshino2018,denti2021costly,morris2022coordination,denti2023unrestricted}.} We provide a general answer that ranges across all cheap learning environments: a full-information Nash equilibrium is attainable as a limit equilibrium outcome for some sequence of information technologies if and only if it is in the closure of the sBCE set.  

In sum, our paper makes three main contributions. First, we show that in generic settings, the economic environment imposes no restrictions on the kinds of information that rationally inattentive agents may acquire. Consequently, studies that use rational inattention to investigate the nature of agents' information must either take a stand on agents' information technologies, or rely on non-generic aspects of the economic setting. Second, our work emphasizes the importance of recognizing the active acquisition of information when evaluating players' welfare. In particular, erroneously assuming information is given in situations where it is actually acquired may result in overstating the benefits players get from their information, and so can result in misleading welfare conclusions. Finally, we make a methodological contribution by developing a framework for obtaining robust predictions in games with rational inattention.

\paragraph{Related literature.}

By opening the door to robust analysis under rational inattention, our paper contributes to the literature on robustness in game-theoretic predictions, mechanism design, information design, and contracting. Our paper is especially pertinent to the work that uses  \citeauthor{bergemann2013robust}' \citeyearpar{bergemann2013robust,bergemann2016bayes} BCE solution concept to obtain reliable predictions in games with incomplete information \citep*[e.g.,][]{bergemann2017first,du2018robust,brooks2021optimal}.%
\footnote{Other related studies regarding robustness that lie outside the BCE framework are \cite{carroll2019InformationGames} and \cite{carroll2019RobustInfoAcquisition}. \cite{carroll2019InformationGames} studies the design of a socially efficient bilateral trade mechanism that is robust to agents' abilities to influence the information structure at a cost. His analysis deviates from the rational inattention framework by assuming agents already know their own valuation, and by allowing for actions that give information to other players. \cite{carroll2019RobustInfoAcquisition} considers a principal who contracts with an expert who can acquire information via a Blackwell experiment. The principal is uncertain about the expert's information technology, knowing only that it contains some fixed set of experiments. \citeapos{carroll2019RobustInfoAcquisition} analysis differs from ours in several ways, the most substantive of which being that he only derives the optimal contract under a worst-case objective, and does not span the set of attainable payoffs and outcomes from a fixed contract.  
} %
In particular, we show one can extend the robustness criterion of any such study to include endogenously determined information \`a la rational inattention, as long as the focus is only on the moments of the action-state distribution (e.g., the seller's revenue in an auction), and the separation constraint is not binding. Even when separation binds, our genericity result means such constraint should be taken seriously only if the analyst is confident about non-generic features of the economic environment. For studies that focus on welfare \citep*[e.g.,][]{bergemann2015limits}, our framework provides a road-map for understanding whether their conclusions continue to hold once one accounts for the cost of information.\footnote{The worst-case policy analysis of Section~\ref{sec:welfare analysis} is also related to the literature on robust Bayesian persuasion \citep*[e.g.,][]{dworczak2022preparing,kosterina2022persuasion,lipnowski2022perfect}. Particularly relevant are studies who take a worst-case approach to study optimal information provision in a binary action games \citep*{inostroza2021persuasion,halac2022addressing,morris2022implementation}. We do not consider information provision, whereas these studies do not consider information acquisition.
}  

We also contribute to the literature on rational inattention in games \citep*[e.g.,][]{Yang2015coordination,ravid2020ultimatum,denti2021costly, morris2022coordination,denti2023unrestricted,hebert2023information}. Within this literature, the closely related work of \cite{denti2021costly} is the first to consider the question of robustness. \cite{denti2021costly} studies a two-player signaling game where the receiver needs to pay a cost to monitor the sender's actions. He shows that, when costs are strictly monotone, off-path beliefs play no role in equilibrium. He then characterizes the set of behavioral predictions consistent with some strictly monotone cost function. As part of this characterization, he obtains what is essentially the single-player version of our Theorem~\ref{thm:mon_tech}: when the sender takes every action with positive probability, the receiver's behavior is consistent with some strictly monotone cost function if and only if it satisfies obedience and separation. Theorem~\ref{thm:mon_tech} expands on \citeapos{denti2021costly} result by allowing for multiple players, and by considering players' welfare from a given outcome \citep*[][does not discuss payoffs]{denti2021costly}. In addition, our analysis of generic games and the comparison to exogenous information have no parallel in \cite{denti2021costly}.

One can interpret some of our results as providing a test for rational inattention in a strategic, multi-agent setting: simply check whether the observed outcome satisfies obedience and separation. Moreover, our genericity result suggests that testing separation requires a highly structured and controlled environment. Thus, our paper can be seen as adding to the growing literature on the testable implications of rational inattention \citep*[e.g.,][]{caplin2015revealed, caplin2022rationally,denti2022posterior,lipnowski2022predicting}. Within this literature, the most closely related paper is \cite{caplin2015revealed}.\footnote{Another related paper is \cite{deClippel2020communication}. They  consider a two-staged game where the first stage player chooses how much to obfuscate the state, and the second stage player chooses what to learn about that state whenever it is obfuscated. In their setting, the second player essentially faces a single agent decision problem. Using this fact, \cite{deClippel2020communication} apply \citeapos{caplin2015revealed} results to obtain testable prediction that are valid across all cost functions. They then test these predictions in a lab experiment.} They develop a test for whether a single agent's choices in multiple menus are consistent with costly information acquisition, but do not require costs to be strictly monotone in information. Their
characterization includes obedience, as well as another condition called \emph{no improving attention cycles} (NIAC), which restricts the agent's behavior across decision problems. Since we consider a fixed base game, the NIAC restriction does not apply in our setting. By contrast, \citeapos{caplin2015revealed} characterization does not require separation, because they allow for costs that are not strictly monotone. Hence, their characterization differs from ours in that it only considers a single agent, does not require costs to be strictly monotone, and accounts for multiple decision problems. 

Our work also expands on the uses of correlated equilibrium and its cousins for spanning the set of predictions attainable across various ways of extending a base game \citep*[e.g.,][]{aumann1974subjectivity, aumann1987correlated, myerson1982optimal, forges1993five, bergemann2016bayes,doval2020sequential}. An early and closely related paper within this literature is \cite{lipman1990computation}. They consider base games without payoff uncertainty, and ask which obedient outcomes (i.e., correlated equilibria) can be attained by extending the game to allow players to acquire costly information about a common payoff irrelevant state space. They maintain two assumptions on players' information technology: costs are (ordinally) symmetric across players, and players must use partitional information. They show an obedient outcome can be generated from an information technology satisfying their assumptions if and only if it satisfies a cyclical monotonicity condition across players that is similar to \citeapos{caplin2015revealed} NIAC. Our work differs from theirs in that we allow for payoff uncertainty, asymmetric costs, and non-partitional information, and that we require costs to be strictly increasing in informativeness. In addition, \cite{lipman1990computation} have no analog of our payoff bounds, which play a central role in our paper. 

\section{Setup}

A finite number of players have to engage in a game with uncertain payoffs; we call such a game the \textit{base game}. Before playing the base game, the players can covertly acquire costly information to reduce the uncertainty they face. We call \textit{information technology} the description of learning resources. Together, a base game and an information technology define an \textit{information acquisition game}, our main object of study. We focus on two predictions: the \textit{equilibrium outcome}, which details the distribution of the players' actions in the base game, and the \textit{equilibrium value}, which is the players' expected payoffs net of information costs. We are interested in the outcomes and values that can be generated as we range over all information technologies that represent \textit{rational inattention}, i.e., when information choice is flexible and more information is more costly. In this section, we precisely define all these objects and provide the associated notation.

\paragraph*{Base Game.} 
Let $\Player$ be a finite set of players, with typical element $\player$. Each player $\player$ has to choose an action $\act_\player$ from a finite set $\Act_{\player}$. As usual, we define $\Act_{-\player}=\prod_{\playerb\neq \player} \Act_{\playerb}$ and $\Act=\Act_\player \times \Act_{-\player}$. Accordingly, we use  $\act_{-\player}=(\act_\playerb)_{\playerb\neq \player}$ to denote the action profile of all players other than $\player$, and $\act=(\act_\player,\act_{-\player})$ to denote the entire action profile. Throughout the paper we adopt the same notational conventions for all Cartesian products indexed by $I\setminus\{i\}$ and $I$. 

Players are expected utility maximizers who care about each other's actions  as well as an exogenous  variable $\paystate$, which is drawn from  a finite set $\Theta$ according to a full-support probability measure $\payprior \in \Delta (\Paystate)$. We denote by $\util_{\player}: \Act \times \Paystate \rightarrow \real$ player $\player$'s von Neumann-Morgenstern utility function. 

We call a tuple
$
\BGame=\left(\Player,\Paystate,\payprior,(\Act_{\player},\util_{\player})_{\player\in \Player}\right)
$
 a \textbf{base game}. 

\paragraph*{Information Technologies.}
Before taking an action in the base game, each player has the opportunity to acquire information  about $\theta$ as well as other exogenous quantities of potential interests (e.g., sunspots, noisy public information). We succinctly represent them by a single variable $\corstate$ taking values in a finite set $\Corstate$; a Markov kernel $\corprior:\Paystate \rightarrow \Delta(\Corstate)$ details the conditional distribution of $\corstate$ given $\paystate$.

Following \cite{Blackwell1951a}, we model the acquisition of information using experiments. An experiment for player $\player$ is a Markov kernel $\exper_{\player}: \Corstate\times \Paystate \rightarrow \Delta (\Signal_\player)$,
where $\Signal_\player$ is a finite space of signal realizations privately observable by player $\player$. The functions $\exper_{\player}$ details how the distribution of $\player$'s signal $\signal_\player$ depends on $\paystate$ and $\corstate$.
To simplify the exposition, we assume the signal space is rich; specifically, we assume $\Signal_\player$ has more elements than $\Act_\player$ and $\Corstate\times\Paystate$.
 
By construction, the players' signals are conditionally independent given $\paystate$ and $\corstate$. However, they may be correlated given $\paystate$ only. Thus, our framework incorporates a form of correlated information acquisition, as in, among others, \cite{Hellwig2009b}, \cite{Myatt2012a}, \cite{hebert2023information}, and \cite{denti2023unrestricted}. Indeed, one can simply view $\corstate$ as a modeling device for situations in which players have access to information sources with correlated noise (e.g., newspapers with similar slants, consultants from the same firm).

The acquisition of information faces two kinds of frictions. First, each player is constrained in the kind of experiments she can use: player $\player$ can only choose experiments that lie in a given set $\Exper_\player$. Second, experiments come at a cost, where $\icost_\player: \Exper_\player \rightarrow \mathbb{R}_{+}$ denotes $\player$'s cost function. As a normalization, we assume the existence of an experiment that costs zero; that is, $\icost_\player(\exper_\player) =0$ for some $\exper_\player\in\Exper_\player$.

We call a tuple 
$
\itech = \left(\Corstate, \corprior, (\Signal_\player,\Exper_\player,\icost_\player)_{\player\in \Player}\right)
$ 
an \textbf{information technology}. 

\paragraph*{Information Acquisition Games.}
Together, a base game $\BGame$ and an information technology $\itech$ define an \textbf{information acquisition game}. The game  begins with the realization of the state of the world $\omega=(\corstate,\paystate)$. Without observing the state, the players simultaneously and covertly choose experiments, and pay their costs. Then, each player privately observes the outcome of their own experiment and takes an action. We use $\aplan_\player: \Signal_{\player}\rightarrow \Delta(\Act_\player)$ to denote player $\player$'s action plan in this game, and let $\Aplan_\player$ be the set of  $\player$'s action plans.

The solution concept we adopt is Nash equilibrium. A strategy for player $\player$ consists of an experiment $\exper_\player\in \Exper_\player$ and an action plan $\aplan_\player\in\Aplan$. A strategy profile $(\exper^{*}_{\player},\aplan^{*}_{\player})_{\player \in \Player}$ is an equilibrium if for all players $\player$, $(\exper^*_\player,\aplan^*_\player)$ maximizes
\[
\left[\sum_{\act,\signal,\corstate,\paystate} \util_\player(\act,\paystate) \aplan_{\player}(\act_\player|\signal_\player)\exper_\player(\signal_{\player}|\corstate,\paystate)\prod_{\playerb\neq \player} \aplan^{*}_{\playerb}(\act_{\playerb}|\signal_\playerb)\exper_{\playerb}^*(\signal_{\playerb}|\corstate,\paystate)\corprior(\corstate|\paystate)\payprior(\paystate)\right] - \icost_{\player}(\exper_\player).
\]
over all $\exper_\player\in\Exper_\player$ and $\aplan_\player\in \Aplan_\player$. The objective function consists of two terms: the value of information (in square brackets) and the cost of information. As common in applications, value and cost are additively separable.

\paragraph{Equilibrium Predictions.}

We summarize the equilibria of information acquisition games using two statistics of the players' behavior: the outcome and the value. 

The \textbf{outcome} is the joint distribution $\out \in \Delta(\Act \times \Paystate)$ of the players' actions and the payoff-relevant state. Note that the marginal distribution of $\theta$ must be the prior $\payprior$; we denote by $\Out$ the set of probability measures over $\Act\times\Paystate$ whose marginal on $\Theta$ is $\payprior$.

The \textbf{value} is the vector $\payvector=(\payvector_\player)_{\player \in \Player} \in \mathbb{R}^{\Player}$ assigning each player the expected payoff in the game. Observe that this payoff includes players' information acquisition costs and so $\payvector_\player$ may differ from the expectation of $\util_\player$ under $\out$. 

\paragraph{Blackwell order.} 
To give a precise definition of ``more information,'' we build on the classic ranking of experiments due to  \cite{Blackwell1951a, Blackwell1953}. Given a pair of experiments $\exper_\player$ and $\experb_\player$,  we say $\exper_i$ \textbf{Blackwell dominates} $\experb_i$ (denoted $\exper_\player \succsim \experb_\player$) if there exists a Markov kernel $\garb:\Signal_\player \rightarrow \Delta(\Signal_\player)$ such that for every $\signal_\player \in \Signal_\player$, $\paystate\in\Paystate$, and $\corstate\in \Corstate$ with $\corprior(\corstate\vert \paystate)>0$,
\[
\experb_\player(\signal_\player|\corstate,\paystate) = \sum_{\signalb_\player \in \Signal_\player} \garb(\signal_\player|\signalb_\player)\exper(\signalb_\player|\corstate,\paystate).
\]
Intuitively, $\exper_\player$ Blackwell dominates $\experb_\player$ if one can generate $\experb_\player$ by ``garbling" the output of $\exper_\player$. As shown by Blackwell, $\exper_\player \succsim \experb_\player$ if and only if player $\player$ is better off observing the output of $\exper_\player$ rather than the output of $\experb_\player$ (holding fixed other players' behavior). In this sense, $\exper_\player$ is more informative than $\experb_\player$. We write $\exper_\player \succ \experb_\player$ whenever $\exper_\player \succsim \experb_\player$ and $\experb_\player \not\succsim \exper_\player$. 

\paragraph{Rational Inattention.}

Our aim is to study the outcomes and values can be generated in equilibrium as we fix the base game and range over all information technologies that represent rational inattention. Consistent with the literature, we interpret rational inattention as information technologies where the set of feasible experiments is flexible, and where it is costly to acquire more information.

We say a set of feasible experiments $\Exper_\player$ is \textbf{flexible} if, whenever a given amount of information is feasible, a lower amount of information is feasible as well: whenever $\exper_\player \in \Exper_\player$ and $\exper_\player \succsim \experb_\player$, then $\experb_\player \in \Exper_\player$. We say a cost function $\icost_\player$ is \textbf{monotone} if  less informative experiments are cheaper to acquire: for all $\exper_\player, \experb_\player \in \Exper_{\player}$ such that $\exper_\player \succsim \experb_\player$ (resp., $\exper_\player \succ \experb_\player$), we have $\icost_\player(\exper_\player) \geq \icost_\player(\experb_\player)$ (resp., $\icost_\player(\exper_\player) > \icost_\player(\experb_\player)$). We say a technology $\IT$ represents \textbf{rational inattention} if for every player, the set of feasible experiments is flexible and the cost function is monotone.
 
Many authors regard flexibility as the key difference between rational inattention and traditional information-acquisition models (see, e.g., \citealp*{mackowiak2023survey}, Section 2). Applications of rational inattention often assume all experiments are feasible: in the language of this paper, $\Exper_\player=\Delta(\Signal_\player)^{\Corstate\times\Paystate}$. Our results are unchanged if we adopt this more extreme definition of flexibility. The reason is that there are no observable differences between experiments that are unfeasible and experiments that are excessively costly.

By pairing flexibility with monotonicity, we postulate that players can save on costs by only acquiring the information they actually use in making decisions. \cite{Matejka2015a} provide a standard example of monotone cost function: $\icost_\player(\exper_\player)$ equals the expected reduction in uncertainty about $\paystate$ and $\corstate$, as measured by Shannon entropy, from observing the output of $\exper_\player$. More broadly, one could substitute Shannon's entropy with any other strictly concave measure of uncertainty \citep*{caplin2022rationally}. One could also consider increasing transformations of these costs (\citealp{denti2022posterior}; \citealp{zhong2022optimal}), or any differentiable cost function whose derivative is strictly convex \citep{lipnowski2022predicting}.\footnote{In each information acquisition game, prior beliefs are exogenously determined by $\payprior$ and $\corprior$. Thus, the issue of experiment-based vs.\ posterior-based information costs  that sometimes arise in applications of rational inattention (see, e.g., \citealp*{ravid2020ultimatum,denti2022experimental}) is irrelevant here.}

\section{A Characterization of Rational Inattention}\label{sec:rational_i}

In this section we characterize the outcomes and values that can be generated via rational inattention. To provide a benchmark, we first review the case of exogenous information.

The class of information acquisition games includes situations in which players' information is predetermined. One can obtain them by considering technologies in which each player has only \textit{one} feasible experiment (whose cost is zero by our normalization). Such technologies are identified by a tuple 
$
\mathcal{S} = \left(\Corstate, \corprior, (\Signal_\player,\exper_\player)_{\player\in \Player}\right)
$
that we call \textbf{information structure}. Of course, an information acquisition game $(\BGame,\mathcal{S})$ is simply a conventional game of incomplete information \`a la Harsanyi.

Among other results, \cite{bergemann2016bayes} characterize the equilibrium outcomes that can arise in a game of incomplete information as we fix the base game and range over all information structures: they call them Bayes correlated equilibria. A \textbf{Bayes correlated equilibrium} (BCE) of a base game $\BGame$ is an outcome $\out\in \Out$ that satisfies the following constraint: for all $\player\in\Player$ and $\act_\player,\actb_\player \in \Act_\player$,
\begin{equation}\label{eq:OC}
\sum_{\act_{-\player},\paystate}\left(\util_{\player}(\act_{\player},\act_{-\player},\paystate) - \util_{\player}(\actb_{\player},\act_{-\player},\paystate) \right)\out(\act_{\player},\act_{-\player},\paystate) \geq 0. 
\end{equation}
Following standard terminology, we name \eqref{eq:OC} the \textbf{obedience constraint}.
The standard way of viewing this constraint is through the lens of a mediator who generates $\out$ by observing the state and stochastically sending an action recommendation to each player. 
The players are willing to follow these recommendations if and only if the obedience constraint is satisfied.\footnote{Our definition of BCE corresponds to the specialization of \citeauthor{bergemann2016bayes}'s (\citeyear{bergemann2016bayes}) definition to the case in which players' original type spaces are degenerate.
} 

Calculating players' values under exogenous information is straightforward. 
If outcome $\out$ arises under exogenous information in base game $\BGame$, player $\player$'s expected payoff is 
\[
\grossval_\player (\out) := \sum_{\act,\paystate}\util_\player(\act,\paystate)\out(\act,\paystate).
\]
We call $\grossval_\player (\out)$ the \textbf{gross value} for player $\player$, since it ignores information costs. Let $\grossval(\out):=(\grossval_\player (\out))_{\player\in\Player}$ be the vector of gross values.

What changes when information is endogenous? Our analysis highlights two main differences between rational inattention and exogenous information. 

The first difference is that outcomes generated by costly information acquisition must satisfy an additional ``separation constraint.'' To present this constraint, we require a few definitions. Given an outcome $\out\in \Out$, a player  $\player\in \Player$, and an action $\act_{\player}\in \Act_\player$, let 
\(
\out(\act_{\player}):= \sum_{\act_{-\player},\paystate}\out(\act_{\player}, \act_{-\player},\paystate)
\)
be the probability of player $\player$ taking action $\act_{\player}$ under $\out$, and let
\[
\supp_{\player} (\out) := \{\act_{\player} \in \Act_{\player}: \out(\act_{\player}) >0 \}
\]
be the set of $\player$'s actions that have positive probability. For each $\act_\player \in \supp_{\player}(\out)$, let $\out_{\act_{\player}} \in \Delta(\Act_{-\player}\times \Paystate)$ be the conditional distribution of the actions of the players other than $\player$ and the payoff-relevant state: for all $\act_{-\player}\in \Act_{-\player}$ and $\paystate\in\Paystate$,
\[
\out_{\act_{\player}}(\act_{-\player},\paystate):=\frac{\out (\act_\player,\act_{-\player},\paystate)}{\out(\act_\player)}.
\]
Finally, denote the set of player $\player$'s best responses to $\out_{\act_\player}$ via 
\[
\BR(\out_{\act_\player}):=\argmax_{\actb_{\player} \in \Act_{\player}} \sum_{ \act_{-\player},\paystate} \util_\player(\actb_{\player},\act_{-\player},\paystate)\out_{\act_{\player}}(\act_{-\player},\paystate).
\]
Then an outcome $\out \in \Out$ satisfies the \textbf{separation constraint} if  for all $\player\in\Player$ and $\act_\player,\actb_\player\in\supp_\player(\out)$,
\begin{equation}\label{eq:SC}
\out_{\act_\player}\neq\out_{\actb_\player}\quad\text{implies}\quad\BR(\out_{\act_\player})\cap \BR(\out_{\actb_\player})= \varnothing.
\end{equation}
In other terms, an outcome satisfies the separation constraint if distinct beliefs have separate best responses. We refer to a BCE that satisfies the separation constraint as a \textbf{separated BCE} (sBCE).

The second difference between rational inattention and exogenous information is in the set of payoffs players obtain from a given outcome. When information is given, each player $\player$'s value is completely determined by the outcome $\out$. By contrast, under rational inattention the same outcome can arise under multiple cost functions, and so is consistent with a \emph{set} of values. As we demonstrate next, this set of values is convex, with an upper bound given by the gross value $\grossval_{\player}(\out)$. The lower bound is given by (what we call) the outcome's ``uninformed value,'' which is $\player$'s maximal value if she receives no information, and others' actions and $\paystate$ are distributed according to that outcome. Formally, the \textbf{uninformed value} of an outcome $\out$ in a base game $\BGame$ for player $\player$ is given by 
\[
\noinfoval_{\player}(\out) := \max_{\actb_{\player}\in \Act_{\player}} \sum_{\act,\paystate} \util_\player(\actb_\player,\act_{-\player},\paystate)\out(\act,\paystate).
\]
 Let $\noinfoval(\out):=(\noinfoval_\player (\out))_{\player\in\Player}$ be the vector of uninformed values.\footnote{The uninformed value is related to the notion of \emph{coarse correlated equilibrium} (see, e.g., \citealp{roughgarden2016twenty}, Definition 13.5). Coarse correlated equilibrium relaxes the obedience constraint, requiring instead that no player can benefit from deviating to a single action from all of the mediator's recommendation; that is, $\grossval_\player(\out) \geq \noinfoval_\player(\out)$ for every $\player$.}

The next result summarizes our characterization of rational inattention in games:

\begin{samepage}
\begin{theorem}\label{thm:mon_tech}
Fix a base game $\BGame$. A rational-inattention technology $\IT$ exists that induces the outcome-value pair $(\out,\payvector)$ in an equilibrium of $(\BGame,\IT)$ if and only if
\begin{enumerate}[(i)]
\item the outcome $\out$ is a separated BCE, and \label{thm:mon_tech_sBCE}
\item for every $\player\in \Player$, $\payvector_\player=\noinfoval_\player(\out)=\grossval_\player(\out)$ or $\payvector_\player\in [\noinfoval_\player(\out),\grossval_\player(\out))$.
\label{thm:mon_tech_value}
\end{enumerate}
\end{theorem}
\end{samepage}

Thus, an outcome-value pair is consistent with rational inattention if and only if two conditions hold. First, the outcome satisfies both the obedience constraint and the separation constraint. And second, each player's value is weakly above her uninformed value, but strictly below her gross value.

We now explain why the theorem's conditions are necessary. That every outcome generated by rational inattention must be obedient follows from the necessity of the obedience constraint under exogenous information. The reason is that any outcome one can attain when players choose their information must also be attainable if players where exogenously endowed with their chosen experiment. To understand why the separation constraint is necessary, consider a player $i$ who takes with positive probability a pair of actions $\act_\player$ and $\actb_\player$ such that $\out_{\act_\player}\neq \out_{\actb_\player}$. As in BCE, we can interpret $\act_\player$ and $\actb_\player$ as signals. With rational inattention, informative signals are costly. To save on information costs, the player could substitute $\act_\player$ and $\actb_\player$ with a \emph{single} action recommendation $\actc_\player$. For this not to be profitable, it must be that either $\actc_\player\notin \BR(\out_{\act_\player})$ or $\actc_\player\notin \BR(\out_{\actb_\player})$. Since the choice of $\actc_\player$ is arbitrary, it must be that $\BR(\out_{\act_\player})\cap \BR(\out_{\actb_\player}) =\varnothing$. 

To conclude our discussion of necessity, we now describe the origin of the theorem's payoff bounds. To see why $\noinfoval_{\player}(\out)$ is a lower bound on player $\player$'s payoff, suppose we have an information technology and an equilibrium that induces the outcome $\out$. By assumption, player $\player$ always has the option of remaining uninformed at no cost. Therefore, $\player$'s optimal payoff must be higher than $\noinfoval_{\player}(\out)$. That $\grossval_{\player}(\out)$ is the maximal payoff player $\player$ can attain in an equilibrium that induces the outcome $\out$ follows from information acquisition costs being non-negative. Moreover, these costs must be strictly positive whenever $\noinfoval_\player(\out)\neq \grossval_\player(\out)$: to generate $\out$ in this case player $\player$ must acquire \emph{some} information. Therefore, player $\player$ can actually attain her gross value from an outcome only if that value coincides with the outcome's uninformed value. 

We now briefly review our proof that conditions \eqref{thm:mon_tech_sBCE} and \eqref{thm:mon_tech_value} of Theorem~\ref{thm:mon_tech} are sufficient. The proof is constructive, and is based on a result by \citet{denti2021costly}. \citet{denti2021costly} studies a signaling game where the receiver has to pay a cost to monitor the sender's action. Among other results, \citet{denti2021costly} shows that, when the sender takes every action with positive probability, any receiver behavior satisfying (what we call in this paper) the obedience constraint and the separation constraint can be justified via rational inattention. 

To prove the ``if'' direction of Theorem~\ref{thm:mon_tech}, we first modify \citeauthor{denti2021costly}'s (\citeyear{denti2021costly}) single-agent construction so that it generates any payoff between the uninformed and the gross value. We emphasize \citet{denti2021costly} focuses on equilibrium outcomes and does not discuss achievable payoffs. Then, we use our modification of \citeauthor{denti2021costly}'s (\citeyear{denti2021costly}) single-agent result to construct a flexible and monotone technology that enables players to simultaneously acquire information about the same state space (in \citealp{denti2021costly}, only \emph{one} player endogenously acquires information).

Theorem~\ref{thm:mon_tech} shows that any BCE that arises from rational inattention must satisfy the separation constraint. Next, we record that the separation constraint never eliminates all of a game's Bayes correlated equilibria; that is, the set of separated BCEs is non-empty.
\begin{corollary}\label{cor:existence}
Every base game $\BGame$ admits a separated BCE.
\end{corollary}

The proof is straightforward (details omitted): A technology where all feasible experiments are free and uninformative is flexible and monotone. The corresponding game with information acquisition admits an equilibrium by standard arguments (information is de facto exogenous). It follows from Theorem \ref{thm:mon_tech} that the outcome of such an equilibrium is a separated BCE.\footnote{Another immediate consequence of this argument is that in the special case in which $\Theta$ is a singleton, all (pure or mixed) Nash equilibria of the base game are separated BCEs.} 

We conclude the section with a brief discussion of what happens if we allow for non-flexible and non-monotone technologies. To accommodate such technologies, we need to adjust Theorem~\ref{thm:mon_tech}'s statement in two ways. First, the separation constraint is no longer necessary. And second, players can attain their gross value even when it differs from their uninformed value. We refer the reader to Online Appendix \ref{sec:Arbitrary information technologies} for the precise details.

\section{Generic Environments}\label{sec:generic}

Theorem~\ref{thm:mon_tech} shows rational inattention differs from exogenous information in two ways. First, it  reduces the set of equilibrium outcomes: an additional separation constraint must be satisfied. Second, it expands the set of achievable payoffs for any given equilibrium outcome. In this section, we prove \emph{only} the second difference is meaningful in generic environments.

We adopt the following notion of genericity: We fix a finite set of players $\Player$, a finite set of payoff states $\Paystate$, a full-support prior $\payprior\in \Delta(\Paystate)$, and a finite set of actions $\Act_\player$ for each player $\player\in\Player$. To obtain a base game, it remains to specify a profile of utility function $\util=(\util_\player)_{\player\in\Player}$. We identify $\util$ with an element of the Euclidean space $\mathbb{R}^{\Player\times\Act\times\Paystate}$, and say a statement is true for \textbf{generic} $\util$ if the closure of the subset in $\mathbb{R}^{\Player\times\Act\times\Paystate}$ for which it is false has Lebesgue measure zero. We denote by $\BCE(\util)$ and $\sBCE(\util)$ the sets of BCEs and separated BCEs for the base game corresponding to $\util$.

\begin{theorem}\label{thm:genericity}
For generic $\util$, the set $\sBCE(\util)$ is dense in the set $\BCE(\util)$.
\end{theorem}

Thus, the environments in which rational inattention predict different outcomes than incomplete information are knife edge.\footnote{It is easy to construct generic examples where the set of separated BCEs is not close. In particular, the statement ``$\sBCE(\util)=\BCE(\util)$ for generic $\util$'' is false. It is also not true that ``$\cl(\sBCE(\util))=\BCE(\util)$ for all $\util$.''} An important caveat to this result is the notion of genericity we use: it is the most common for the static games we study in this paper, but also the most permissive. For example, according to this notion of genericity, many important economic applications---such as auctions or oligopolies---are non-generic. The notion of genericity is also not appropriate if the base game is not actually static but represents the strategic form of a primitive dynamic game. We discuss sBCE in non-generic environments in Section \ref{sec:non_generic}.

One might be tempted to think that  Theorem~\ref{thm:genericity} follows from indifferences being ``fragile.'' This intuition works, but only for the single agent case. When $I$ has one element, the set of \emph{strict} BCEs---i.e., the set of outcomes $\out$ such that $\BR(\out_{\act_\player})=\{\act_\player\}$ for every player $\player$ and $\act_\player\in \supp_\player(\out)$---is generically dense in the BCE set. Theorem~\ref{thm:genericity} then follows from noting that every strict BCE is separated.\footnote{See Online Appendix \ref{sec:generic_singla_agent} for a detailed argument.} 

The situation radically changes  when there are at least two players. The reason is that, with two or more players, indifferences emerge in equilibrium in generic fashion: when $I$ has more than one element, there exists an open set of games where \emph{no} BCE is strict. For example, consider the games in a neighborhood of Matching Pennies: such games have a unique correlated equilibrium (the fully-mixed Nash equilibrium) is which both players are indifferent between both actions.
Hence, to prove Theorem~\ref{thm:genericity} beyond the single-agent case, one cannot hope to show that indifferences are somewhat fragile. 

Our proof of Theorem~\ref{thm:genericity} combines two independent lemmas. The first lemma shows that for any BCE $\out$ and utility profile $\util$, there exists a perturbation of $\util$ that makes $\out$ separated. In other words, every BCE of a given game is the limit of separated BCEs of nearby games. 

The second lemma shows that for generic games, any BCE that is a limit of separated BCE in nearby games is also a separated BCE when the game is held fixed. Specifically, the second lemma shows that the correspondence 
$
\util\mapsto \cl(\sBCE(\util)),
$
which takes utilities to the closure of the sBCE set, 
generically is upper hemicontinuous (in fact, continuous). 

To prove this second lemma, we employ ideas from \cite{blume1994algebraic}, who study the algebraic geometry of Nash, perfect, and sequential equilibria. In particular, we use the Tarsky-Siedenberg Theorem to show that $\util\mapsto \cl(\sBCE(\util))$ has semi-algebraic graph, and so must be continuous for all utility profiles outside a closed low-dimensional manifold. Combined, these two lemmas imply that the set of games where the sBCE set is not dense in the BCE set is contained in a closed low-dimensional manifold. The theorem follows. 

Theorem~\ref{thm:genericity} implies that, in generic games, rational inattention and exogenous information are outcome equivalent. Next, we show this equivalence does not extend to players' welfare. In other words, even though the two knowledge regimes generically yield the same behavioral predictions, they can have very different welfare implications.

For $\util\in\mathbb{R}^{\Player\times\Act\times\Paystate}$, let $\RIval(\util)$ be the closure of set of attainable value vectors under rational inattention. By Theorem~\ref{thm:mon_tech}, these are given by the set of value vectors that lie between the uninformed value and the gross value of some limit of separated BCEs:\footnote{For two vector of real numbers $\const,\constb \in \mathbb{R}^{K}$ such that $\const \leq \constb$, we use the box, or ordered interval notation $[\const,\constb]=\{\constc \in \mathbb{R}^{K}: \const\leq \constc \leq \constb\}$ to denote the set of all vectors between $\const$ and $\constb$.}
\begin{equation}\label{eq:def_rival}
\RIval(\util) = \{v\in \left[\noinfoval(\out,\util),\grossval(\out,\util)\right]: \out\in\cl(\sBCE(\util))\},
\end{equation}
where $\noinfoval(\out,\util)$ and $\grossval(\out,\util)$ make explicit the dependence of uninformed and gross values on the players' utility functions.

We also denote by $\Exval(\util)$ the set of attainable value vectors under exogenous information. \citeapos{bergemann2016bayes} analysis implies this is given by the set of gross value vectors attainable in some BCE,
\begin{equation}\label{eq:def_exval}
\Exval(\util) =\left\{\grossval(\out,\util):\out\in\BCE(\util)\right\}.
\end{equation}

For an arbitrary $\util$, there is no simple relationship between $\RIval(\util)$ and $\Exval(\util)$: $\sBCE(\util)$ is a subset of $\BCE(\util)$, but $\grossval(\out,\util)$ is an element of $\left[\noinfoval(\out,\util),\grossval(\out,\util)\right]$, so one cannot easily conclude that $\RIval(\util)$ contains or is contained by $\Exval(\util)$, or neither.

In generic environments, the comparison is simpler: 

\begin{samepage}
\begin{theorem}\label{thm:robust difference in welfare}
    For generic $\util$,
    $
\Exval(\util) \subseteq \RIval(\util).
$
In addition, if $|\Player|\geq 2$, $|\Paystate| \geq 2$, and $|\Act_\player|\geq 2$ for all players $\player$, then the set of $\util$ for which 
    $
    \Exval(\util) \subset \RIval(\util)
    $
    has non-empty interior.
\end{theorem}
\end{samepage}
Thus, in generic environments, rational inattention expands the set of achievable payoffs, and it does so in a non-trivial way (i.e., with strict inclusion) for a class of environments of positive measure. The result suggests that rational inattention and exogenous information may have different welfare implications even when they are outcome equivalent. In the next section, we make this idea concrete in an application to robust policy analysis.

We conclude this section by sketching the proof of Theorem~\ref{thm:robust difference in welfare}. The first part of the theorem immediately follows from Theorem~\ref{thm:genericity}, together with (\ref{eq:def_rival}) and (\ref{eq:def_exval}): for generic $\util$, $\cl(\sBCE(\util))=\BCE(\util)$, and therefore 
\[
\Exval(\util) \subseteq \{v\in \left[\noinfoval(\out,\util),\grossval(\out,\util)\right]: \out\in\BCE(\util)\}=\RIval(\util).
\]

To show the strict inclusion, the key step involves finding a utility profile $\util^*$ and a BCE $\out^*$ of the corresponding base game with two properties. First, each player has exactly two positive probability actions, both of which are strict best responses; thus, in particular, $\out^*$ is separated. And second, for every player $\player$, $\out^*$ is the unique minimizer of $i$'s uniformed value across all BCEs of $\util^*$. The first property implies that $\noinfoval_\player(\out^*,\util^*)<\grossval_\player(\out^*,\util^*)$ for all players $\player$---the equality $\noinfoval_\player(\out^*,\util^*)=\grossval_\player(\out^*,\util^*)$ holds if and only if $\player$ has a single action that is a best response to all action recommendations. With the aid of the second property, then we show that
\[
\inf_{\out\in \sBCE(\util^*)}\sum_\player \noinfoval_\player(\out,\util^*)=\sum_\player\noinfoval_\player (\out^*, \util^*) < \min_{\out\in \BCE(\util^*)}\sum_\player \grossval_\player(\out,\util^*).
\]
This allows us to deduce that $\Exval(\util^*) \subset \RIval(\util^*)$. Finally, we use a continuity argument to show that $\Exval(\util) \subset \RIval(\util)$ for all games $\util$ in a neighborhood of $\util^*$.

\section{Application: Robust Welfare Analysis}\label{sec:welfare analysis}

In this section we demonstrate how one can use our results to conduct robust welfare analysis in an economy where agents exhibit rational inattention.

We consider an economy that consists of a fixed set of agents, $\Player$, who play an information acquisition game, $(\BGame,\IT)$. The structure of the game depends on the policy enacted by a utilitarian social planner. The planner has a good understanding of the policy's material implications, i.e., she knows a given policy leads to a given $\BGame$. However, the planner is unsure about the accompanying information technology $\IT$.

We focus on two cases regarding the source of the agents' information. The planner either postulates that information is exogenously given, or that it is generated by rational inattention. In both cases, the planner identifies a policy with the corresponding base game $\BGame$ and employs a robust criterion that evaluates it according to the worst-case utilitarian welfare across all relevant information technologies $\IT$ and ensuing equilibria. 

For the exogenous information case, \citeapos{bergemann2016bayes} results imply one can find the  social value of a policy $\BGame$ by minimizing the sum of the agents' gross payoffs across all BCEs. Specifically, let $\welfex(\out)$ be the utilitarian welfare of an outcome $\out$ assuming the players'
payoffs are given by their gross value,
\[
\welfex(\out):=\sum_{\player} \grossval_\player(\out). 
\]
Then the planner's value of a base game $\BGame$ under exogenous information is
\[
\Welfex := \min_{\out \in \BCE} \welfex (\out),
\]
where $\BCE$ is the set of Bayes correlated equilibria of $\BGame$. Since the BCE set is defined by linear inequalities, and $\welfex (\out)$ is linear in $\out$, one can compute $\Welfex$ via linear programming.

What about a planner who postulates that agents are rationally inattentive? Theorem~\ref{thm:mon_tech} provides the answer: such a planner evaluates each base game according to the lowest sum of \emph{uninformed} values that is attainable in some \emph{separated} BCE. More precisely, for an outcome $\out$, let $\welfen(\out)$ be the utilitarian welfare implied by $\out$ in the base game $\BGame$ if players' payoffs are given by their uninformed value,
\[
\welfen(\out) := \sum_{\player} \noinfoval_\player(\out).
\]
Then Theorem~\ref{thm:mon_tech} suggests that a planner who assumes information is endogenous would evaluate each game according to the minimum of $\welfen(\out)$ across all separated BCEs $\out$,
\[
\inf_{\out \in \sBCE}\welfen(\out).
\]

Recall, however, that Theorem~\ref{thm:genericity} says that the separation constraint does not bind for generic games. As such, imposing the separation constraint is only appropriate if the planner is absolutely certain of the structure of the base game. Whereas such certainty might be justifiable in certain cases, here we take the perspective of a cautious planner who, in an economy where agents are rationally inattentive, evaluates $\BGame$ according to the worst-case value of $\welfen(\out)$ across all BCE,
\[
\Welfen := \min_{\out \in \BCE} \welfen (\out).
\]
Since the BCE set is defined by linear inequalities, and $\welfen (\out)$ is convex in $\out$, one can compute $\Welfen$ using convex programming.

Clearly, the planner's decisions under rational inattention and exogenous information can differ only if the two generate different values, that is, if $\Welfen \neq \Welfex$. Because each players' uninformed value is always lower than her gross value, we have $\Welfen \leq \Welfex$; that is, welfare under rational inattention is always lower than the welfare under exogenous information. Below we characterize when the inequality is strict in symmetric binary-action games.

A base game $\BGame=(\Player,\Paystate,\payprior,(\Act_\player,\util_\player)_{\player\in\Player})$  has \textbf{binary actions} if for every player $\player$, $\Act_\player$ contains two elements. It is \textbf{symmetric} if $\Act_\player = \Act_\playerb$ for all $\player,\playerb \in \Player$, and if for every permutation $\perm: \Player \rightarrow \Player$, player $\player$, action profile $\act$, and payoff state $\paystate$,
\[
\util_\player(\act_\perm,\paystate) = \util_{\perm(\player)}(\act,\paystate),
\]
where $\act_\perm:=(\act_{\perm(\playerb)})_{\playerb\in\Player}$ is the action profile such that each player $\playerb$ takes action $\act_{\perm(\playerb)}$. Given a symmetric base game, an outcome $\out \in \Delta(\Act\times \Paystate)$ is \textbf{symmetric} if for every permutation $\perm$, action profile $\act$, and payoff state $\paystate$, 
\[\out(\act,\paystate) = \out(\act_\perm,\paystate).\] 
We denote the set of symmetric outcomes by $\Outsym$. As one might hope, focusing on symmetric outcomes is without loss for welfare analysis in symmetric games:

\begin{lemma}\label{lem:symmetric worst case}
    If $\BGame$ is symmetric, then both $\min_{\out \in \BCE}\welfen(\out)$ and $\min_{\out \in \BCE}\welfex(\out)$ admit symmetric optimal solutions.
\end{lemma} 

The next result characterizes the binary-action symmetric games for which rational inattention yields strictly lower worst-case welfare than exogenous information. 

\begin{proposition}\label{pro:welfare_binary_symmetric}
    Let $\BGame$ be a symmetric, binary-action base game. Then, $\Welfen<\Welfex$ if and only if all $\out^* \in \argmin_{\out \in \Outsym} \welfen(\out)$ satisfy the following condition:
    \begin{equation}\label{eq:welfare_binary_symmetric}
    \act_\player\in\supp_\player(\out^*)\text{ and }\BR(\out_{\act_\player}^*) = \{\act_\player\}\quad\quad \text{for all }\player\in \Player \text{ and } \act_\player \in \Act_\player.
    \end{equation}
    Moreover, in this case, 
    \[
    \argmin_{\out \in \Outsym} \welfen(\out) \subseteq \argmin_{\out \in \BCE}\welfen(\out).
    \]
\end{proposition}
Thus, one can check whether $\Welfen<\Welfex$ by examining the minimizers of $\welfen(\out)$ among \emph{all} symmetric outcomes $\out$, ignoring the players' obedience constraints. In particular, one needs to check whether all these minimizers recommend both actions to every player, and only send recommendations that induce unique best responses. An immediate implication of the above proposition is that, in a symmetric binary-action game, if $\Welfen<\Welfex$, then
    \[
    \Welfen = \min_{\out \in \Outsym} \welfen(\out).
    \]

We now briefly explain the proposition's proof. The key observation is that in a binary action game, a BCE $\out$ satisfies \eqref{eq:welfare_binary_symmetric} if and only if $\grossval_\player(\out) > \noinfoval_\player(\out)$ holds for every player $\player$. To get intuition for the ``if" direction, note that having \emph{two} recommendations that lead to strict best responses means players get a strictly positive benefit from following them. This benefit creates a wedge between the gross value $\grossval_\player(\out)$, which accounts for the value of information, and the uninformed value $\noinfoval_\player(\out)$, which does not. For the converse direction, note that a violation of \eqref{eq:welfare_binary_symmetric} means some player $\player$ has an action that is optimal across all of the mediator's recommendations. As such, player $\player$ loses nothing by ignoring the mediator's recommendations and taking that action. In other words, player $\player$'s gross value equals her uninformed value. 

Armed with the above observation, we prove the proposition in two steps. The first step shows $\Welfen<\Welfex$ holds if and only if all optimal solutions of $\min_{\out \in \BCEsym} \welfen(\out)$ satisfy (\ref{eq:welfare_binary_symmetric}). This step follows from applying the above-mentioned observation to symmetric outcomes. The second step shows that all optimal solutions of $\min_{\out \in \BCEsym} \welfen(\out)$ satisfy (\ref{eq:welfare_binary_symmetric}) if and only if all optimal solutions of $ \min_{\out \in \Outsym} \welfen(\out)$ satisfy (\ref{eq:welfare_binary_symmetric}). Loosely speaking, the reason is as follows: if the obedience constraint does not bind at the optimum---as (\ref{eq:welfare_binary_symmetric}) dictates---then it can be relaxed, and therefore minimizing over $\out\in \BCEsym$ is the same as minimizing over $\out\in \Outsym$.

Proposition~\ref{pro:welfare_binary_symmetric} enables us to find circumstances where the planner's optimal policy depends on whether she believes players' information is exogenously given, or generated by rational inattention. We demonstrate this fact below in a regime change game. 

\begin{example}\label{exa:regime change}
We consider a regime change game whereby a status quo is abandoned if a sufficiently large number of players take an action against it. Such games are well-studied and have been used to model a variety of social phenomena, including currency crises, bank runs, debt crises, and political revolts.\footnote{See, e.g., \cite*{obstfeld1996models,Morris1998unique,goldstein2005demand,morris2023inspiring}.} 

In our application, there are $\Playernum$ identical investors, $\player \in \Player =\{1,\ldots,\Playernum\}$, each of which decides whether to speculate against (i.e., attack) a distressed financial institution ($\act_\player =1$) or not ($\act_\player = 0$). Speculating costs $\brcost\in (0,1)$. If enough investors speculate, the institution fails (i.e., the attack succeeds), generating a profit of $1$ to the speculators, and an externality of $-\brext$ to all passive investors, where $\brext\in (0,\infty)$. The payoff state $\paystate \in \Paystate \subseteq \{1,\ldots,\Playernum\}$, determines the number of speculators required for the attack to succeed. We assume that $\min \Paystate > 1$, and $\max \Paystate < \Playernum - 1$, meaning no single investor can go against the will of all the others.\footnote{A fortiori, $n>3$ because $2\leq \min \Paystate\leq\max \Paystate < \Playernum - 1$.} We summarize these payoffs below:

\begin{center}
\begin{game}{2}{2}
        &     $\sum_\playerb \act_\playerb \geq \paystate $   &    $\sum_\playerb \act_\playerb < \paystate $\\
$\util_\player(1,\act_{-\player},\paystate)$  &    $1-\brcost$  &  $-\brcost$\\
$\util_\player(0,\act_{-\player},\paystate)$  &    $-\brext$ &   $0$
\end{game}
\end{center}

\medskip

We focus on finding conditions under which there is a difference between the worst-case welfare under rational inattention and exogenous information.\footnote{It is easy to verify that, in this example, rational inattention and exogenous information are outcome equivalent, that is, the sBCE set is dense in the BCE set. Indeed, the regime change game admits a strict BCE $\out$ where all players take both actions with positive probability (e.g., a convex combinations of the pure Nash equilibria in which everyone speculates or no-one speculates). Thus, any BCE $\outb$ can be approximated by a separated BCE $(1-\epsilon)\outb+\epsilon \out$ as $\epsilon\rightarrow 0$.} 
By Proposition~\ref{pro:welfare_binary_symmetric}, answering this question requires us to minimize the sum of the uninformed values across all symmetric outcomes, ignoring obedience constraints. The following claim characterizes the solutions to this problem.
\begin{claim}\label{claim:br unconstrained min}
    In the regime change game, a symmetric $\out^*$ minimizes $\welfen(\out)$ across all outcomes $\out \in \Outsym$ if and only if the following conditions hold: for all payoff state $\paystate$,
    \begin{align}
    \out^* \left( \left\{(\act,\paystate) : \paystate-1\leq\sum_\player\act_\player\leq\paystate \right\} \right) & = 0, \label{eq:br not pivotal} \\ 
    \out^* \left( \left\{ (\act,\paystate) : \sum_{\player}\act_\player > \paystate \right\} \right)& =  \frac{\brcost}{1+\brext}.\label{eq:br ex-ante fail prob}
    \end{align}    
\end{claim}
The two optimality conditions have simple interpretations. First, no investor is ever pivotal. And second, the attack succeeds with probability $\brcost/(1+\brext)$. To obtain these conditions, we show that $\out^*$ minimizes $\welfen(\out)$ across all $\out \in \Outsym$ only if each individual investor $\player$ is \emph{indifferent} between never speculating and always speculating when other investors play according to $\out^*$. Minimizing across all symmetric outcomes satisfying this condition then delivers the result.

Finding a symmetric $\out^*$ that satisfies the obedience constraints in addition to the claim's conditions is easy: have all investors attack together with probability $\brcost/(1+\brext)$ regardless of the state, and no one speculates otherwise. Calculating the sum of the uninformed values from  $\out^*$ immediately give the worst-case welfare under rational inattention,
\begin{equation}\label{eq:invariance_ex}
\Welfen = \welfen(\out^*) = -\frac{\Playernum\brext\brcost}{1+\brext}.
\end{equation}
Hence, under rational inattention, the planner's value decreases in the size of the externality $\brext$ and the cost of betting on the financial institution's demise $\brcost$, but does not depend on $\payprior$; that is, the planner's value does not depend on the institution's fundamentals. 

Next we argue a planner who views information as exogenously given may adopt different policies than a planner who thinks investors are rationally inattentive. Towards this goal, we first specialize Proposition~\ref{pro:welfare_binary_symmetric} to the current setting. 

\begin{claim}\label{claim:br different welfare char}
In a regime change game $\BGame$ one has $\Welfen(\BGame) < \Welfex(\BGame)$ if and only if every symmetric $\out$ that satisfies \eqref{eq:br not pivotal} and \eqref{eq:br ex-ante fail prob} must also satisfy
\begin{equation}\label{eq:br large conditional failure prob}
    \out_{\act_\player = 1}\bigg\{\sum_\playerb \act_\playerb \geq \paystate \bigg\} > \frac{\brcost}{1+\brext}.
\end{equation}    
\end{claim}

For intuition, we first explain why the claim's conditions are sufficient; that is, why \eqref{eq:br large conditional failure prob} holding for all symmetric outcomes satisfying \eqref{eq:br not pivotal} and \eqref{eq:br ex-ante fail prob} implies that $\Welfen < \Welfex$. 
Note that an investor that is never pivotal is indifferent between not speculating and speculating if and only if she believes the institution is going to fail with a probability of exactly $\brcost/(1+\brext)$. Therefore, whenever \eqref{eq:br not pivotal}, \eqref{eq:br ex-ante fail prob}, and \eqref{eq:br large conditional failure prob} all hold, investors are ex-ante indifferent between speculating and not, but strictly prefer to speculate conditional on getting a "speculate" recommendation. Moreover, by Bayes rule, the probability the institution fails conditional on the mediator telling an investor not to speculate is strictly below $\brcost/(1+\brext)$. So, \eqref{eq:br not pivotal}, \eqref{eq:br ex-ante fail prob} and \eqref{eq:br large conditional failure prob} together imply that both the "speculate" and the "do not speculate" recommendations lead investors to have a strict best response. Since \eqref{eq:br not pivotal} and \eqref{eq:br ex-ante fail prob} are equivalent to minimizing $\welfen$ over $\Outsym$, we obtain that the claim's condition implies Proposition~\ref{pro:welfare_binary_symmetric}-\eqref{pro:welfare_binary_symmetric:symmetric max condition}, which in turn delivers $\Welfen < \Welfex$. For the converse direction, suppose one can find a symmetric outcome $\out$ that satisfies \eqref{eq:br not pivotal} and \eqref{eq:br ex-ante fail prob}, but violates \eqref{eq:br large conditional failure prob}. Since \eqref{eq:br large conditional failure prob} fails for $\out$, not speculating must be a best response following a "speculate" recommendation; that is, $\out$ violates condition \eqref{eq:welfare_binary_symmetric:unique best response}. Proposition~\ref{pro:welfare_binary_symmetric} then implies $\Welfen(\BGame) = \Welfex(\BGame)$.

Next we argue a planner who views information as exogenous may adopt different policies than a planner who thinks investors are rationally inattentive. To do so, we characterize when these two knowledge environments attain the same worst-case welfare. The appendix provides a general characterization (see Claim \ref{claim:br general welfare difference inequality}); the claim below specializes this characterization to the case where $\Paystate$ is binary. 

\begin{claim}\label{claim:br two states char}
    Suppose $\Paystate=\{\minpaystate,\maxpaystate\}$, where $\maxpaystate\geq\minpaystate$. Then $\Welfen = \Welfex$ if and only if $\maxpaystate-\minpaystate \geq 3$ and 
    \begin{equation}\label{eq:br binary welfare diff condition}
         1 - \frac{1}{3}\left(\maxpaystate-\minpaystate\right)\payprior(\maxpaystate) \leq \frac{\brcost}{1+\brext} \leq \frac{1}{3}\left(\maxpaystate-\minpaystate\right)\left(1-\payprior(\maxpaystate)\right).
    \end{equation}
\end{claim}

In other terms, with binary $\paystate$, worst-case welfare under rational inattention is the same as under exogenous information if and only if the state is sufficiently uncertain in the sense that $\maxpaystate-\minpaystate \geq 3$, and the probability of $\maxpaystate$ is not too extreme compared to $\brcost/(1+\brext)$---e.g., (\ref{eq:br binary welfare diff condition}) fails if $\payprior(\maxpaystate)$ goes to zero or one. In particular, the worst-case welfare under rational inattention is always strictly below the welfare under exogenous information when there's certainty about the institution's fundamentals (i.e., when $\paystate$ is deterministic). 

The proof of Claim~\ref{claim:br two states char} is rather detailed; here we provide only a rough intuition. Combining  Proposition~\ref{pro:welfare_binary_symmetric} and Claim~\ref{claim:br unconstrained min}, one can show the worst-case welfare under exogenous information coincides with the  worst-case welfare under rational inattention only if a symmetric outcome exists that satisfies two conditions. First, an investor who does not see the mediator's recommendation is indifferent between attacking and not attacking. And second, for an investor who does see the mediator's recommendation, either not speculating is a best response to a ``speculate" recommendation, or vice versa. Appealing to Bayes rule, one can show that an outcome satisfies these conditions if and only if there is limited overlap between the event where many investors are attacking and the event in which the attack succeeds. Claim~\ref{claim:br two states char} follows from showing a sufficient disconnect between these two events is attainable if and only if there is enough uncertainty about $\paystate$. Intuitively, disconnecting the two events is easy when $\paystate$ obtains both high and low values with large probability: one can have the number of speculators just come short of a successful attack when $\paystate$ is high, and come just above the threshold when $\paystate$ is low. The same cannot be done when $\paystate$ is deterministic. In that case, successful attacks necessarily involve more speculating investors than failed ones. 

A takeaway is that, unlike under rational inattention, a planner who views information as exogenous may adopt policies that change the institution's fundamentals. For a concrete illustration, consider two policies $\BGame$ and $\BGame^\prime$ that differ only in the institution's fundamentals, that is, in the set of states $\Paystate$ and their distribution $\payprior\in\Delta(\Paystate)$. Suppose $\BGame$ satisfies the conditions of Claim~\ref{claim:br two states char}, but $\BGame^\prime$ does not. Worst-case welfare under rational inattention is the same for $\BGame$ and $\BGame^\prime$:  
$
\Welfen =\welfen^\prime
$
by (\ref{eq:invariance_ex}).
By contrast, $\BGame$ generates lower welfare under exogenous information:
$
\Welfex<\welfex^\prime
$
by (\ref{eq:invariance_ex}) and Claim \ref{claim:br two states char}.
Consequently, a planner who views information as being exogenous would pay some amount to change the institution's fundamentals. 
\end{example}

\section{Further Results}\label{sec:further_results}

In this section, we discuss additional results regarding non-generic environments and equilibria with almost-free information. All related proofs are in the Online Appendix. 

\subsection{Non-generic Environments}\label{sec:non_generic}
Theorem \ref{thm:genericity} shows the separation constraint has no bite for generic games. However, many economic environments, such as auctions, are non-generic. In this section we characterize the (non-generic) environments where the separation constraint has substantive impact. We also show that the impact of separation has an all-or-nothing flavor: if the sBCE set is not dense in the BCE set, it is in fact \emph{nowhere dense}. 

Fix a base game $\BGame$. For the separation constraint to bite, players must have weak incentives to follow the mediator's action recommendations. A special case is one where, whenever the mediator recommends player $\player$ action $\actb_\player$ in a BCE of the game, the player would be equally happy to take action $\act_\player$. \cite{myerson1997dual} calls this scenario ``jeopardization'': action $\act_\player$ \textbf{jeopardizes} action $\actb_\player$ if, for every BCE $\out$ such that $\actb_\player\in\supp_\player(\out)$, $\act_\player\in BR(\out_{\actb_\player})$.\footnote{Myerson defines jeopardization for games without payoff uncertainty; here we give the obvious extension to games where $\Theta$ is not a singleton.}  We denote by $J(\actb_\player)$ the set of actions that jeopardizes $\actb_\player$. 

Every action jeopardizes itself by the obedience constraint; hence, $J(\actb_\player)$ is not empty. A sufficient condition for jeopardization is weak domination: if $\util_\player(\act_\player,\act_{-\player},\paystate)\geq \util_\player(\actb_\player,\act_{-\player},\paystate)$ for all $\act_{-\player}\in\Act_{-\player}$ and $\paystate\in\Paystate$, then $\act_\player$ jeopardizes $\actb_\player$. But the concept of jeopardization is broader than weak domination. For example, in Matching Pennies, Heads and Tails jeopardize each other, even if neither action is weakly dominant.

Weak incentives are necessary but not sufficient for the separation constraint to be binding: it must also be that different action recommendations induce distinct posterior beliefs. Next we introduce a class of BCE in which such separation is most pronounced. Say a BCE $\out$ is \textbf{maximally supported} if the support of every other BCE is contained by the support of $\out$. A maximally-supported BCE $\out$ is \textbf{minimally mixed} if 
\[
\outb_{\act_\player}\neq\outb_{\actb_\player}\quad\text{implies}\quad\out_{\act_\player}\neq\out_{\actb_\player}
\]
for every BCE $\outb$, $\player\in \Player$, and $\act_\player,\actb_\player\in\supp_\player(\outb)$.

For an interpretation of minimal mixing, take the perspective of a mediator who wants to implement a BCE $\out$. When $\out_{\act_\player}=\out_{\actb_\player}$, the mediator can replace the distinct recommendations of playing $\act_\player$ and $\actb_\player$ with a single recommendation of mixing between the two actions with probabilities $\out(\act_\player)/(\out(\act_\player)+\out(\actb_\player))$ and $\out(\actb_\player)/(\out(\act_\player)+\out(\actb_\player))$. A BCE $\out$ is minimally mixed if a mediator has the least amount of opportunities to implement $\out$ recommending mixed actions. Whereas minimally mixed BCEs seem esoteric at first, they are in fact, ubiquitous: the set of minimally mixed BCEs is open and dense in the BCE set (see Lemma \ref{lem:B-BCEs are open and dense in B} in the Online Appendix).

Our next result uses the concepts of jeopardization and minimally mixed BCEs to characterize when the BCE set and the closure of the sBCE set coincide. 
\begin{samepage}
\begin{proposition}\label{pro:BCE equals cl sBCE iff stuff}
The following statements are equivalent:
\begin{enumerate}[(i)]
\item \label{pro:BCE equals cl sBCE iff stuff: closure} The sBCE set is dense in the BCE set. 
\item \label{pro:BCE equals cl sBCE iff stuff: min mixed} A minimally mixed sBCE exists. 
\item \label{pro:BCE equals cl sBCE iff stuff: jeopardization} For every BCE $p$, $\player\in \Player$, $\act_\player,\actb_\player\in \supp_\player(\out)$,
\[
\out_{\act_\player}\neq\out_{\actb_\player}\quad\text{implies}\quad \Jeop(\act_\player)\cap \Jeop(\actb_\player)=\varnothing.
\]
\end{enumerate}
\end{proposition}
\end{samepage}

The result shows how jeopardization and minimal mixing can be used in applications to study sBCE. To verify that the sBCE set is dense in the BCE set, it is enough to produce a minimally mixed sBCE. To verify that the sBCE set is not dense in the BCE set, it is enough to produce a BCE in which two actions induce distinct beliefs and share a common jeopardizing action. As shown by \cite{myerson1997dual}, the jeopardizing actions can be easily computed from the dual of the system of linear inequalities that defines BCE. In Online Appendix \ref{sec:Checking for Equal Beliefs}, we simplify the analysis of minimally mixed BCE by describing how to find the actions that induce different beliefs for some BCE.

Next, we build on Proposition \ref{pro:BCE equals cl sBCE iff stuff} and obtain that sBCE is an all-or-nothing refinement of BCE.
\begin{theorem}\label{thm:Dense or nowhere dense}
    The sBCE set is either dense or nowhere dense in the BCE set.
\end{theorem}

For a rough explanation, consider first the case in which a minimally mixed sBCE exists. Then, by Proposition \ref{pro:BCE equals cl sBCE iff stuff}, the sBCE set is dense in the BCE set. Consider now the case in which a minimally sBCE does not exist. By Proposition \ref{pro:BCE equals cl sBCE iff stuff}, the sBCE set is not dense in the BCE set. To reach the stronger conclusion that the sBCE set is nowhere dense in the BCE set, we use the fact that the set of minimally mixed BCE is open and dense in the BCE set.

Thus, whereas the separation constraint does not bind in most circumstances, whenever it does bind, it significantly restricts the set of attainable outcomes. 

\subsection{Vanishing Cost Equilibrium}\label{sec:vanishing_costs}

The case in which the cost of information is very small can be interpreted as a perturbation of complete information. As such, it can serve as a device for selecting equilibria of games where the state of the environment is commonly known.\footnote{See, e.g., \cite*{Yang2015coordination, Hoshino2018, denti2021costly, denti2023unrestricted,  morris2022coordination}. \cite{ravid2022learning} conduct a similar exercise in which only one player can acquire information.} Existing works study this limit scenario under various restrictions on the information technology; next we obtain a robust characterization by leveraging the connection between rational inattention and separated BCE.

Fix a base game $\BGame$. An outcome $\out\in\Delta_\payprior(\Act\times\Paystate)$ is a \textbf{complete-information Nash equilibrium} if for every $\paystate\in\Paystate$ there is a mixed action profile $\mact_{\paystate}=(\mact_{\paystate,\player})_{\player\in\Player}$ such that
\begin{align*}
\frac{\out(\act,\paystate)}{\payprior(\paystate)}&=\prod_{\player\in\Player}\mact_{\paystate,\player}(\act_\player),&&\text{for all }\act\in\Act,\\
\mact_{\paystate,\player} &\in\arg\max_{\mactb_\player\in\Delta(\Act_\player)}\sum_{\act}\util_\player(\act,\paystate)\mactb_\player(\act_\player)\prod_{\playerb\neq\player}\mact_{\paystate,\playerb}(\act_\playerb),&&\text{for all }\player\in\Player. 
\end{align*}
Notice that every complete-information Nash equilibrium is a BCE. Such BCE is generated by an information structure in which the payoff state $\paystate$ is commonly known and the correlation state $\corstate$ is degenerate. 

Under what conditions a complete-information Nash equilibrium is the result of almost-frictionless information acquisition? To answer this question, we first make precise what we mean by ``almost frictionless'' information acquisition. 

A technology $\IT$ represents \textbf{unconstrained rational inattention} if for every player $\player$, all experiments are feasible (i.e., $\Exper_\player=\Delta(\Signal_\player)^{\Corstate\times\Paystate}$), and the cost function $\icost_\player$ is monotone. An outcome $\out\in\Delta_\payprior(\Act\times\Paystate)$ is a \textbf{vanishing cost equilibrium} if for every $\epsilon>0$, there exist an unconstrained rational-inattention technology $\itech$ and an equilibrium $(\exper,\aplan)$ of the information-acquisition game $(\BGame,\itech)$ such that, denoting by $\outb$ the outcome of $(\exper,\aplan)$, 
\begin{equation*}
\max \icost_\player(\Exper_\player) \leq \epsilon
\quad \text{and} \quad 
\left\vert \out(\act,\paystate) - \outb(\act,\paystate) \right\vert  
\leq \epsilon
\end{equation*}
for all $\player \in \Player$, $\act \in \Act$, and $\paystate \in \Paystate$.

The next theorem characterizes what complete-information Nash equilibria are vanishing cost equilibria.

\begin{theorem}\label{thm:complete_info_limit}
A complete-information Nash equilibrium $\out$ is a vanishing cost equilibrium if and only if it belongs to the closure of the sBCE set.
\end{theorem}

The ``only if'' side of Theorem \ref{thm:complete_info_limit} is an immediate consequence of Theorem \ref{thm:mon_tech}: since outcomes of rational inattention games are separated BCEs, any vanishing cost equilibrium $\out$ must be the limit of a sequence of separated BCEs. The ``if'' side of Theorem \ref{thm:complete_info_limit} requires more work. Theorem \ref{thm:mon_tech} only guarantees the existence of a sequence of rational-inattention games whose outcomes converge to $\out$: it does not give that information costs go to zero along the sequence. To prove this additional property, we use the fact that $\out$ is a complete-information Nash equilibrium.

Recall that, by Theorem \ref{thm:genericity}, the sBCE set generically is dense in the BCE set. Thus, a corollary of Theorem \ref{thm:complete_info_limit} is that, generically, all complete-information Nash equilibria are vanishing cost equilibria. 
 
\begin{corollary}\label{cor:generic_limit}
For generic $\util$, every complete-information Nash equilibrium is a vanishing cost equilibrium.
\end{corollary}

As a consequence, if one intends to use rational inattention to select equilibria of complete information games, one needs either to focus on specialized economic environments or take a stance on the nature of the information acquisition technology.  

Finally, we remark that not all vanishing cost equilibria are complete-information Nash equilibria. We provide a complete characterization of vanishing cost equilibria in Online Appendix \ref{sec:proof_vanishing_cost_gen}. In particular, we show that \emph{convex combinations} of complete-information Nash equilibria can also be the result of almost-frictionless information acquisition.

\section{Conclusion}
Since its introduction by \cite{Sims2003a}, rational inattention has emerged as a powerful theory, with a wide range of applications throughout economics. This theory serves both as a portable model of limited cognition, and as a way of anchoring information to incentives. But for all its merits, the theory's current formalization has a significant drawback: researchers must make hard-to-verify assumptions on players' information acquisition capabilities. The current paper addresses this concern by developing a framework for spanning all predictions consistent with rational inattention in a given economic setting. 

We also use our framework to make two conceptual contributions. First, we show that in generic settings, the economic environment imposes no restrictions on the types of information that may emerge under rational inattention. Hence, studies employing rational inattention to understand the shape of agents' information must either make specialized assumptions about the underlying information technology, or rely on highly-structured characteristics of the economic environment. Second, our work underscores the significance of considering the costs of information acquisition when assessing players' welfare. Specifically, mistakenly assuming that information is given when it is actually acquired can lead one to overestimate the benefits players derive from their information, and so may result in misleading welfare conclusions.

\newpage
\appendix

\begin{center}
\huge{Appendix}
\end{center}

\section{Preliminary results}

The next sections introduce lemmas we will use to prove the results from the main text. The lemmas may be of independent interest, so we collect them here, in individual sections.

\subsection{Single-agent lemmas on rational inattention}\label{subsec:singla_agent}

This section presents single-agent results on costly information acquisition. We take the perspective of an individual $\player$ who has to choose an action $\act_\player\in \Act_\player$ whose utility $w_\player(\act_\player,\omega)$ depends on an uncertain state of nature $\omega\in \Omega$. Both $\Act_\player$ and $\Omega$ are finite. Let $\rho\in \Delta(\Omega)$ be the prior distribution of the state; $\rho$ may not have full support.

Before choosing an action, the decision maker can run an experiment $\exper_\player:\Omega\rightarrow\Delta(\Signal_i)$ at a cost $C_i(\exper_\player)\in\mathbb{R}_+$. Let $\Exper_\player\subseteq \Delta(\Signal_\player)^\Omega$ be the set of feasible experiments. The signal space $\Signal_i$ is finite and contains more elements than the sets $\Omega$ and $\Act_\player$. 

Overall, the decision maker faces the following information-acquisition problem:
\begin{equation}\label{eq:single_agent}
\max_{\exper_\player\in\Exper_\player,\aplan_\player\in\Aplan_\player} \left[\sum_{\omega,\signal_\player,\act_\player} w_\player(\act_\player,\omega) \aplan_{\player}(\act_\player|\signal_\player)\exper_\player(\signal_{\player}|\omega)\rho(\omega)\right] - \icost_{\player}(\exper_\player)
\end{equation}
where $\Aplan_\player$ is the set of all action plans $\aplan_\player:\Signal_\player\rightarrow\Delta(\Act_\player)$. 

In accordance with the terminology used in the main text, we write $\exper_\player\succsim \experb_\player$ if there is a Markov kernel $\garb:\Signal_\player \rightarrow \Delta(\Signal_\player)$ such that for every $\signal_\player \in \Signal_\player$ and $\omega\in\Omega$ with $\rho(\omega)>0$,
\begin{equation}\label{eq:BOrder_as}
\experb_\player(\signal_\player|\omega) = \sum_{\signalb_\player \in \Signal_\player} \garb(\signal_\player|\signalb_\player)\exper(\signalb_\player|\omega).
\end{equation}
We say that $\Exper_\player$ is \textbf{flexible} if, whenever $\exper_\player \succsim \experb_\player$ and $\exper_\player \in \Exper_\player$, then $\experb_\player \in \Exper_\player$. We also say that $\icost_\player:\Exper_\player\rightarrow\mathbb{R}_+$ is \textbf{monotone} if, whenever $\exper_\player, \experb_\player \in \Exper_{\player}$ are such that $\exper_\player \succsim \experb_\player$ (resp., $\exper_\player \succ \experb_\player$), then $\icost_\player(\exper_\player) \geq \icost_\player(\experb_\player)$ (resp., $\icost_\player(\exper_\player) > \icost_\player(\experb_\player)$).

Next we characterize the pairs $(\exper_\player,\aplan_\player)$ that are optimal solutions of (\ref{eq:single_agent}) for some flexible $\Exper_\player$ and monotone $\icost_\player$. We will use the following notation. Given a pair $(\exper_\player,\aplan_\player)$, we denote by $\mu_\player\in\Delta(\Act_\player\times\Signal_\player\times\Omega)$ the induced probability measure over actions, signals, and states:
\[
\mu_\player(\act_\player,\signal_\player,\omega)=\aplan_{\player}(\act_\player|\signal_\player)\exper_\player(\signal_{\player}|\omega)\rho(\omega).
\]
We also denote by $\mu_\player(\signal_\player)$ the unconditional probability of signal $\signal_\player$:
\[
\mu_\player(\signal_\player)=\sum_{\act_\player,\omega}\mu(\act_\player,\omega).
\]
For all $\signal_\player$ such that $\mu_\player(\signal_\player)>0$, let $\mu_{\signal_\player}\in \Delta (\Omega)$ be the conditional distribution of the state:
\[
\mu_{\signal_\player}(\omega)=
\frac{\sum_{\act_{\player}}\mu_\player(\act_\player,\signal_\player,\omega)}
{\mu_\player(\signal_\player)}.
\]
Let $\BR\left(\mu_{\signal_\player}\right)$ be the corresponding set of best responses:
\[
\BR\left(\mu_{\signal_\player}\right) = \argmax_{\act_\player \in \Act_\player} \left[ 
\sum_{\omega}\util_{\player}(\act_\player,\omega)\mu_{\signal_\player}(\omega)
\right].
\] 
\begin{lemma}\label{lemma:single_agent_plus}
A flexible $\Exper_\player\subseteq \Delta(\Signal_\player)^\Omega$ and a monotone $\icost_\player:\Exper_\player\rightarrow\mathbb{R}_+$ exist such that $(\exper_\player,\aplan_\player)$ is an optimal solution of (\ref{eq:single_agent}) if and only if the following conditions hold:
\begin{itemize}
\item[(i)] For all signals $\signal_\player$ such that $\mu_\player(\signal_\player)>0$, 
    \[
    \aplan_\player \left(\BR\left(\mu_{\signal_\player}\right)\vert \signal_\player\right)=1.
    \]
\item[(ii)] For all signals $\signal_\player$ and $\signal_\player^\prime$ such that $\mu_\player(\signal_\player)>0$ and $\mu_\player(\signalb_\player)>0$,
    \[
\mu_{\signal_i}\neq \mu_{\signal^\prime_i}\quad\text{implies}\quad \BR(\mu_{\signal_\player})\cap \BR(\mu_{\signalb_\player})=\varnothing.
    \]
\end{itemize}
In addition, for every $\lambda_\player\in(0,1]$, one can choose $\icost_\player$ so that
\[
\icost_\player(\exper_\player)=\lambda_\player \left(\sum_{\act_\player, \signal_\player,\omega} w_\player(\act_\player,\omega) \mu_\player(\act_\player,\signal_\player,\omega)-
\max_{\act_\player\in\Act_\player}\sum_{\omega} w_\player(\act_\player,\omega) \rho(\omega)\right).
\]
\end{lemma}

Lemma \ref{lemma:single_agent_plus} extends a result in \cite{denti2021costly}. \cite{denti2021costly} uses a stronger version of Blackwell's informativeness:  $\exper_\player\succsim^* \experb_\player$ if there is a Markov kernel $\garb:\Signal_\player \rightarrow \Delta(\Signal_\player)$ such that for every $\signal_\player \in \Signal_\player$ and $\omega\in\Omega$,
\begin{equation}\label{eq:black_there}
\experb_\player(\signal_\player|\omega) = \sum_{\signalb_\player \in \Signal_\player} \garb(\signal_\player|\signalb_\player)\exper(\signalb_\player|\omega).
\end{equation}
The difference between $\succsim$ and $\succsim^*$ is that (\ref{eq:BOrder_as}) holds on the support of $\rho$ while (\ref{eq:black_there}) on $\Omega$. 

We say that $\Exper_\player$ is \textbf{flexible*} if, whenever $\exper_\player \succsim^* \experb_\player$ and $\exper_\player \in \Exper_\player$, then $\experb_\player \in \Exper_\player$. We also say that $\icost_\player:\Exper_\player\rightarrow\mathbb{R}_+$ is \textbf{monotone*} if, whenever $\exper_\player, \experb_\player \in \Exper_{\player}$ are such that $\exper_\player \succsim^* \experb_\player$ (resp., $\exper_\player \succ^* \experb_\player$), then $\icost_\player(\exper_\player) \geq \icost_\player(\experb_\player)$ (resp., $\icost_\player(\exper_\player) > \icost_\player(\experb_\player)$).

\begin{lemma}[\citealp{denti2021costly}]\label{lemma:single_agent}
A flexible* $\Exper_\player\subseteq \Delta(\Signal_\player)^\Omega$ and a monotone* $\icost_\player:\Exper_\player\rightarrow\mathbb{R}_+$ exist such that $(\exper_\player,\aplan_\player)$ is an optimal solution of (\ref{eq:single_agent}) if and only if the following conditions hold:
\begin{itemize}
\item[(i)] For all signals $\signal_\player$ such that $\mu_\player(\signal_\player)>0$, 
    \[
    \aplan_\player \left(\BR\left(\mu_{\signal_\player}\right)\vert \signal_\player\right)=1.
    \]
\item[(ii)] For all signals $\signal_\player$ and $\signal_\player^\prime$ such that $\mu_\player(\signal_\player)>0$ and $\mu_\player(\signalb_\player)>0$,
    \[
\mu_{\signal_i}\neq \mu_{\signal^\prime_i}\quad\text{implies}\quad \BR(\mu_{\signal_\player})\cap \BR(\mu_{\signalb_\player})=\varnothing.
    \]
\item[(iii)] For all signals $\signal_\player$ and states $\omega$, if $\exper_\player(\signal_\player\vert\omega)>0$ then $\mu_\player(\signal_\player)>0$.
\end{itemize}
In addition, one can choose $\icost_\player$ so that
\[
\icost_\player(\exper_\player)=\sum_{\omega,\signal_\player,\act_\player} w_\player(\act_\player,\omega) \mu_\player(\act_\player,\signal_\player,\omega)-
\max_{\act_\player\in\Act_\player}\sum_{\omega} w_\player(\act_\player,\omega) \rho(\omega).
\]
\end{lemma}

\begin{proof}[Proof of Lemma \ref{lemma:single_agent_plus}]To use Lemma \ref{lemma:single_agent}, we define $\Omega^* = \supp (\rho)$. Let $\rho^*$ be the restriction of $\rho$ to $\Omega$, and let $w_\player^*$ be the restriction of $w_\player$ to $\Act_\player\times\Signal_\player\times\Omega^*$. By construction, $\rho^*$ has full support on $\Omega^*$; thus, for experiments defined on $\Omega^*$, the rankings $\succeq$ and $\succeq^*$ coincide.

``If.'' Suppose $(\exper_\player,\aplan_\player)$ satisfy Lemma \ref{lemma:single_agent_plus}-(i) and Lemma \ref{lemma:single_agent_plus}-(ii). Let $\exper_\player^*$ be the restriction of $\exper_\player$ to $\Omega^*$, and let $\mu^*_\player$ be the restriction of $\mu_\player$ to $\Act_\player\times\Signal_\player\times\Omega^*$. 

Lemma \ref{lemma:single_agent_plus}-(i) implies that for all signals $\signal_\player$ such that $\mu^*_\player(\signal_\player)>0$, 
    \[
    \aplan_\player \left(\BR\left(\mu^*_{\signal_\player}\right)\vert \signal_\player\right)=1.
    \]
Lemma \ref{lemma:single_agent_plus}-(ii) implies that for all signals $\signal_\player$ and $\signal_\player^\prime$ such that $\mu^*_\player(\signal_\player)>0$ and $\mu^*_\player(\signalb_\player)>0$,
    \[
\mu^*_{\signal_i}\neq \mu^*_{\signal^\prime_i}\quad\text{implies}\quad \BR(\mu^*_{\signal_\player})\cap \BR(\mu^*_{\signalb_\player})=\varnothing.
    \]
Since $\rho^*$ has full support, if $\exper^*_\player(\signal_\player\vert\omega^*)>0$ for $\signal_\player\in\Signal_\player$ and $\omega^*\in\Omega^*$, then $\mu_\player^*(\signal_\player)>0$. 

Thus, by Lemma \ref{lemma:single_agent}, there exist a flexible* $\Exper^*_\player\subseteq\Delta(\Signal_\player)^{\Omega^*}$ and monotone* $\icost^*:\Exper^*\rightarrow \mathbb{R}_+$ such that $(\exper^*_\player,\aplan_\player)$ is an optimal solution of 
\begin{equation}\label{eq:apr_12_morning}
\max_{\exper^\star_\player\in\Exper^*_\player,\aplan_\player\in\Aplan_\player} \left[\sum_{\omega^*,\signal_\player,\act_\player} w^*_\player(\act_\player,\omega^*) \aplan_{\player}(\act_\player|\signal_\player)\exper^\star_\player(\signal_{\player}|\omega^*)\rho^*(\omega^*)\right] - \icost^*_{\player}(\exper^\star_\player).
\end{equation}
In addition, we can choose the cost function so that
\[
\icost^*_\player(\exper^*_\player)=\sum_{\omega^*,\signal_\player,\act_\player} w^*_\player(\act_\player,\omega^*) \mu_\player^*(\act_\player,\signal_\player,\omega)-
\max_{\act_\player\in\Act_\player}\sum_{\omega^*} w_\player^*(\act_\player,\omega^*) \rho^*(\omega^*).
\]

Fix $\lambda_\player\in (0,1]$. Let $\Exper_\player$ be the set of all $\exper^\prime_\player:\Omega\rightarrow\Delta(\Signal_\player)$ such that $\exper\succeq\exper^\prime$. For $\exper^\prime\in\Exper_\player$, define $\icost_\player(\exper^\prime_\player)=\lambda_\player \icost^*_\player(\exper^\star_\player)$ where $\exper^\star_\player$ is the restriction of $\exper^\prime_\player$ to $\Omega^*$. Note that $\icost_\player$ is well defined because, if $\exper_\player\succeq\exper^\prime_\player$ and $\exper^\star_\player$ is the restriction of $\exper^\prime$ to $\Omega^*$, then $\exper^*_\player\succeq^*\exper^\star_\player$, which in turn implies $\exper^\star_\player\in \Exper_\player^*$ (given that $\Exper_\player^*$ is flexible*). Note that 
\[
\icost_\player(\exper_\player)=\lambda_\player \left(\sum_{\act_\player, \signal_\player,\omega} w_\player(\act_\player,\omega) \mu_\player(\act_\player,\signal_\player,\omega)-
\max_{\act_\player\in\Act_\player}\sum_{\omega} w_\player(\act_\player,\omega) \rho(\omega)\right).
\]

Clearly, $\Exper_\player$ is flexible. Since $\icost^*_\player$ is monotone*, $\icost_\player$ is monotone. Moreover, since $(\exper^*_\player,\aplan_\player)$ is an optimal solution of (\ref{eq:apr_12_morning}), then $(\exper_\player,\aplan_\player)$ is an optimal solution of 
\[
\max_{\exper^\prime_\player\in\Exper_\player,\aplan_\player\in\Aplan_\player} \left[\sum_{\omega,\signal_\player,\act_\player} w_\player(\act_\player,\omega) \aplan_{\player}(\act_\player|\signal_\player)\exper^\prime_\player(\signal_{\player}|\omega)\rho(\omega)\right] - \frac{1}{\lambda_\player}\icost_{\player}(\exper^\prime_\player).
\]
By construction of $\Exper_\player$, $(\exper_\player,\aplan_\player)$ is also optimal solution of
\[
\max_{\exper^\prime_\player\in\Exper_\player,\aplan_\player\in\Aplan_\player} \left[\sum_{\omega,\signal_\player,\act_\player} w_\player(\act_\player,\omega) \aplan_{\player}(\act_\player|\signal_\player)\exper^\prime_\player(\signal_{\player}|\omega)\rho(\omega)\right].
\]
Combining these two facts, we obtain that $(\exper_\player,\aplan_\player)$ is an optimal solution of 
\[
\max_{\exper^\prime_\player\in\Exper_\player,\aplan_\player\in\Aplan_\player} \left[\sum_{\omega,\signal_\player,\act_\player} w_\player(\act_\player,\omega) \aplan_{\player}(\act_\player|\signal_\player)\exper^\prime_\player(\signal_{\player}|\omega)\rho(\omega)\right] - \icost_{\player}(\exper^\prime_\player).
\]
This concludes the proof of the ``if'' part of Lemma \ref{lemma:single_agent_plus}.

``Only if.'' Let $(\exper_\player,\aplan_\player)$ be an optimal solution of (\ref{eq:single_agent}) for some flexible $\Exper_\player$ and some monotone $\icost_\player$. Let $\exper^*$ be the restrictions of $\exper$ to $\Omega^*$, and let  $\mu^*_\player$ be the restriction of $\mu_\player$ to $\Act_\player\times\Signal_\player\times\Omega^*$.

We denote by $\Exper_\player^*$ the set of all $\exper_\player^\star:\Omega^*\rightarrow\Delta(\Signal_\player)$ for which there exists $\exper_\player^\prime\in\Exper_\player$ such that $\exper_\player^\star(\cdot\vert \omega^*)=\exper_\player^\prime(\cdot\vert \omega^*)$ for all $\omega^*\in\Omega^*$. For $\exper_\player^\star\in \Exper_\player^*$, let $\icost_\player^*(\exper_\player^\star)$ be the infimum of $\icost_\player(\exper_\player^\prime)$ over all $\exper_\player^\prime\in\Exper_\player$ such that $\exper_\player^\star$ is the restriction of $\exper_\player^\prime$ to $\Omega^*$. Notice that, since $\icost$ is monotone, $\icost_\player^*(\exper_\player^\star)=\icost_\player(\exper_\player^\prime)$ for all $\exper_\player^\star\in \Exper_\player^*$ and $\exper_\player^\prime\in \Exper_\player$ such that $\exper_\player^\star$ is the restriction of $\exper_\player^\prime$ to $\Omega^*$.

Since $\Exper_\player$ is flexible, $\Exper_\player^*$ is flexible*. Since $\icost$ is monotone, $\icost^*$ is monotone*. Moreover, since  $(\exper_\player,\aplan_\player)$ is a optimal solution of (\ref{eq:single_agent}) for $\Exper_\player$ and $\icost_\player$, then $(\exper_\player^*,\aplan_\player)$ is an optimal solution of 
\[
\max_{\exper^\star_\player\in\Exper^*_\player,\aplan_\player\in\Aplan_\player} \left[\sum_{\omega^*,\signal_\player,\act_\player} w^*_\player(\act_\player,\omega^*) \aplan_{\player}(\act_\player|\signal_\player)\exper^\star_\player(\signal_{\player}|\omega^*)\rho^*(\omega^*)\right] - \icost^*_{\player}(\exper^\star_\player).
\]
Then, Lemma \ref{lemma:single_agent}-(i) implies  Lemma \ref{lemma:single_agent_plus}-(i), and Lemma \ref{lemma:single_agent}-(ii) implies  Lemma \ref{lemma:single_agent_plus}-(ii).
\end{proof}

Next we refine Lemma \ref{lemma:single_agent_plus} by showing that one can put a bound on the cost of all experiments. For short, we define
\begin{align*}
\underline{w}_\player(\mu_\player) & = \max_{\actb_\player\in\Act_\player}\sum_{\act_\player,\signal_\player,\omega} w_\player(\actb_\player,\omega) \mu_\player(\act_\player,\signal_\player,\omega)=\max_{\actb_\player\in\Act_\player}\sum_{\omega} w_\player(\actb_\player,\omega) \rho(\omega),\\
\overline{w}_\player(\mu_\player) & = \sum_{\act_\player,\signal_\player,\omega} w_\player(\act_\player,\omega) \mu_\player(\act_\player,\signal_\player,\omega), \\
\hat{w}_\player(\mu_\player) & =\sum_{\omega} \max_{\actb_\player\in\Act_\player} \sum_{\act_\player,\signal_\player} w_\player(\actb_\player,\omega) \mu_\player(\act_\player,\signal_\player,\omega)=\sum_{\omega} \max_{\actb_\player\in\Act_\player} w_\player(\actb_\player,\omega) \rho(\omega).
\end{align*}
\begin{lemma}\label{lemma:single_agent_bound}
If $(\exper_\player,\aplan_\player)$ satisfies Lemma \ref{lemma:single_agent_plus}-(i) and Lemma \ref{lemma:single_agent_plus}-(ii), then there exists a monotone $\icost_\player:\Delta(\Signal_\player)^\Omega\rightarrow\mathbb{R}_+$ such that $(\exper_\player,\aplan_\player)$ is an optimal solution of (\ref{eq:single_agent}), with $\Exper_\player=\Delta(\Signal_\player)^\Omega$. Moreover, for every $\lambda_\player\in(0,1]$, one can choose  $\icost_\player$ so that 
\begin{align*}
\icost_\player(\exper_\player) & =\lambda_\player\left(\overline{w}_\player(\mu_\player)- \underline{w}_\player(\mu_\player)\right),\\
\max_{\experb_\player\in\Delta(\Signal_\player)^\Omega} \icost_\player(\experb_\player)& \leq \lambda_\player + \hat{w}_\player(\mu_\player) - \left[(1-\lambda_\player)\overline{w}_\player(\mu_\player)+\lambda_\player \underline{w}_\player(\mu_\player)\right].
\end{align*}
\end{lemma}
\begin{proof}
By Lemma \ref{lemma:single_agent_plus}, there are flexible  $\Exper_\player^\prime\subseteq \Delta(\Signal_\player)^\Omega$ and monotone $\icost_\player^\prime:\Exper_\player^\prime\rightarrow\mathbb{R}_+$ such that $(\exper_\player,\aplan_\player)$ is an optimal solution of (\ref{eq:single_agent}). In addition, we can assume that 
\[
\icost_\player^\prime(\exper_\player)= \lambda_\player\left(\overline{w}_\player(\mu_\player)- \underline{w}_\player(\mu_\player)\right).
\]
Note that the upper bound on information costs we wish to obtain can be re-written as
\[
\max_{\experb_\player\in\Delta(\Signal_\player)^\Omega} \icost_\player(\experb_\player) \leq \lambda_\player + \icost_\player^\prime(\exper_\player) +\hat{w}_\player(\mu_\player) - \overline{w}_\player(\mu_\player).
\]

Let $\icost_\player^\star:\Delta(\Signal_\player)^\Omega\rightarrow \mathbb{R}_+$ be any finite monotone cost function defined over the set of all experiments (e.g., the entropy cost of \citealp{Matejka2015a}). Let $K_\player>0$ be a positive constant such that for all $\exper_\player^\prime\in \Delta(\Signal_\player)^\Omega$, 
\[
K_\player\icost_\player^\star(\exper_\player^\prime)\leq \lambda_\player.
\]

We define $\icost_\player:\Delta(\Signal_\player)^\Omega\rightarrow\mathbb{R}_+$ as follows. For $\exper_\player\succeq \exper_\player^\prime$, we set 
$
\icost_\player(\exper_\player^\prime)=\icost_\player^\prime(\exper_\player^\prime),
$
and for $\exper_\player\not\succeq \exper_\player^\prime$, we set 
\[
\icost_\player(\exper_\player^\prime)= K_\player\icost_\player^\star(\exper_\player^\prime)+ \icost_\player^\prime(\exper_\player)+\hat{w}_\player(\mu_\player) - \overline{w}_\player(\mu_\player).
\]
Notice that for all $\exper_\player^\prime\in \Delta(\Signal_\player)^\Omega$,
\[
\icost_\player(\exper_\player^\prime) \leq \lambda_\player + \icost_\player^\prime(\exper_\player)+\hat{w}_\player(\mu_\player) - \overline{w}_\player(\mu_\player).
\]
Thus, $\icost_\player$ has the desired bound. It remains to check that $\icost_\player$ is monotone, and that $(\exper_\player,\aplan_\player)$ is an optimal solution of (\ref{eq:single_agent}) when all experiments are feasible and costs are given by $\icost_\player$. 

To verify that $\icost_\player$ is monotone, let $\experb_\player$ and  $\experc_\player$ a pair of experiments. If  $\experb_\player\sim \experc_\player$, then either both experiments are dominated by $\exper_\player$, in which case 
\[
\icost_\player(\exper_\player^\prime)=\icost_\player^\prime(\exper_\player^\prime)=\icost_\player^\prime(\exper_\player^{\prime\prime})=\icost_\player(\exper_\player^{\prime\prime})
\]
because $\icost_\player^\prime$ is monotone, or neither experiment are dominated by $\exper_\player$, in which case 
\[
\icost_\player(\exper_\player^\prime)-\icost_\player(\exper_\player^{\prime\prime})=K_\player(\icost_\player^\star(\exper_\player^\prime)-\icost_\player^\star(\exper_\player^{\prime\prime}))=0
\]
because $\icost_\player^\star$ is monotone. In any case, if $\experb_\player\sim \experc_\player$, then $\icost_\player(\exper_\player^\prime)=\icost_\player(\exper_\player^{\prime\prime})$.

Suppose now that $\experb_\player\succ \experc_\player$. If $\exper_\player \succeq \experb_\player$, then $\icost_\player(\exper_\player^\prime)>\icost_\player(\exper_\player^{\prime\prime})$ because $\icost_\player^\prime$ is monotone. If $\exper_\player \not\succeq \experc_\player$, then $\icost_\player(\exper_\player^\prime)>\icost_\player(\exper_\player^{\prime\prime})$ because $\icost_\player^\star$ is monotone. Finally, consider the case in which $\exper_\player \not\succeq \experb_\player$ and $\exper_\player \succeq \experc_\player$. We have that
\begin{align*}
\icost_\player(\exper_\player^\prime) 
& = K_\player\icost_\player^\star(\exper_\player^\prime)+ \icost_\player^\prime(\exper_\player)+\hat{w}_\player(\mu_\player) - \overline{w}_\player(\mu_\player)\\
& > \icost_\player^\prime(\exper_\player)+\hat{w}_\player(\mu_\player) - \overline{w}_\player(\mu_\player)\\
& \geq\icost_\player^\prime(\exper_\player^{\prime\prime})=\icost_\player(\exper_\player^{\prime\prime}).
\end{align*}
where the strict inequality follows from $\icost_\player^\star$ being monotone and $\exper_\player \not\succeq \experb_\player$, and the weak inequality from $\icost_\player^\prime$ being monotone and $\exper_\player \succeq \experc_\player$. Overall, we conclude that $\icost_\player$ is monotone.

Next we check that $(\exper_\player,\aplan_\player)$ is an optimal solution of (\ref{eq:single_agent}) when all experiments are feasible and the cost of information is $\icost_\player$. Let $\exper_\player^\prime\in \Delta(\Signal_\player)^\Omega$ and $\aplan_\player^\prime\in\Aplan_\player$ be an alternative solution of (\ref{eq:single_agent}). If $\exper_\player\succeq \exper_\player^\prime$, then $\exper_\player^\prime\in \Exper_\player^\prime$ because $\exper_\player\in \Exper_\player^\prime$ and $\Exper_\player^\prime$ is flexible. Thus, since $(\exper_\player,\aplan_\player)$ is an optimal solution of (\ref{eq:single_agent}) when the set of feasible of experiments is $\Exper_\player^\prime$ and the cost of information is $\icost_\player^\prime$, then $(\exper_\player^\prime,\aplan_\player^\prime)$ cannot be strictly better than $(\exper_\player,\aplan_\player)$ when the costs of $\exper_\player$ and $\exper_\player^\prime$ are given by $\icost_\player$, given that $\icost_\player$ coincide with $\icost_\player^\prime$ on $\Exper_\player^\prime$.

Consider now the case in which $\exper_\player\not\succeq \exper_\player^\prime$. The net payoff from $(\exper_\player^\prime,\aplan_\player^\prime)$ is at most 
\begin{align*}
 \sum_\omega \max_{\act_\player} w_\player(\act_\player,\omega)\rho(\omega) - \icost_\player(\exper_\player^\prime) 
& = - \left( K_\player\icost_\player^\star(\exper_\player^\prime)+ \icost_\player^\prime(\exper_\player)-\sum_{\omega,\signal_\player,\act_\player} w_\player(\act_\player,\omega) \mu_\player(\act_\player,\signal_\player,\omega)\right)\\
& \leq \sum_{\omega,\signal_\player,\act_\player} w_\player(\act_\player,\omega) \mu_\player(\act_\player,\signal_\player,\omega)-\icost_\player^\prime(\exper_\player)\\
& = \sum_{\omega,\signal_\player,\act_\player} w_\player(\act_\player,\omega) \mu_\player(\act_\player,\signal_\player,\omega)-\icost_\player(\exper_\player).
\end{align*}
Overall, we conclude that $(\exper_\player,\aplan_\player)$ is an optimal solution of (\ref{eq:single_agent}) when all experiments are feasible and the cost of information is $\icost_\player$.
\end{proof}

\subsection{Equilibrium information structures}\label{subsec:embedding}

In this section we characterize the information structures that can arise in an equilibrium of a rational inattention game for a fixed base game $\BGame$.

We adopt the following notation. Given an information structure $\mathcal{S}=(\Corstate,\corprior,(\Signal_\player,\exper_\player)_{\player\in\Player})$ and a profile of action plans $\aplan=(\aplan_\player)_{\player\in\Player}$, we denote by $\nu\in \Delta (\Act\times\Signal\times\Corstate\times\Paystate)$ the induced probability measure over actions, signals, and states:
\begin{equation}
\nu(\act,\signal,\corstate,\paystate) =
\left[\prod_{\player\in\Player} \aplan_\player(\act_\player|\signal_\player) \exper_\player(\signal_\player|\corstate,\paystate)\right]
\corprior(\corstate|\paystate)\payprior(\paystate).
\end{equation}
Let $\nu(\signal_\player)$ be the probability that player $\player$ observes signal $\signal_\player$:
\[
\nu(\signal_\player)=\sum_{\act,\signal_{-\player},\corstate,\paystate}\nu(\act,\signal_\player,\signal_{-\player},\corstate,\paystate).
\]
For all $\signal_\player$ such that $\nu(\signal_\player)>0$, we denote by $\nu_{\signal_\player}\in \Delta (\Act_{-\player}\times\Signal_{-\player}\times\Corstate\times\Paystate)$ the conditional distribution of the others' actions, the others' signals, and the state:
\[
\nu_{\signal_\player}(\act_{-\player},\signal_{-\player},\corstate,\paystate)=
\frac{\sum_{\act_{\player}}\nu(\act_\player,\act_{-\player},\signal_\player,\signal_{-\player},\corstate,\paystate)}
{\nu(\signal_\player)}.
\]
Let $\BR\left(\nu_{\signal_\player}\right)$ be the corresponding set of best responses:
\[
\BR\left(\nu_{\signal_\player}\right) = \argmax_{\act_\player \in \Act_\player} \left[ 
\sum_{\act_{-\player},\signal_{-\player},\corstate,\paystate}\util_{\player}(\act_\player,\act_{-\player},\paystate)\nu_{\signal_\player}(\act_{-\player},\signal_{-\player},\corstate,\paystate)
\right].
\] 
With a slight abuse of of notation, we also define
\begin{align*}
\noinfoval_\player(\nu) &=\max_{\actb_\player\in\Act_\player}\sum_{\act,\signal,\corstate,\paystate} \util_\player(\actb_{\player},\act_{-\player},\paystate) \nu(\act,\signal,\corstate,\paystate),\\
\grossval_\player(\nu) &=\sum_{\act,\signal,\corstate,\paystate} \util_{\player}(\act,\paystate) \nu(\act,\signal,\corstate,\paystate),\\
\hat{v}_\player(\nu) & =\sum_{\corstate,\paystate}\max_{\actb_\player} \sum_{\act,\signal} \util_{\player}(\actb_\player,\act_{-\player},\paystate) \nu(\act,\signal,\corstate,\paystate).
\end{align*}
\begin{lemma}\label{lem:multi_agent}
Let $\BGame$ be a base game, $\mathcal{S}=(\Corstate,\corprior,(\Signal_\player,\exper_\player)_{\player\in\Player})$ an information structure, and $\aplan=(\aplan_\player)_{\player\in\Player}$ a profile of action plans. For every player $\player$, there exist a flexible $\Exper_\player$ and a monotone $\icost_\player$ such that $(\exper,\aplan)$ is an equilibrium of $(\BGame,\itech)$, with $\itech=(\Corstate,\corprior,(\Signal_\player,\Exper_\player,\icost_\player)_{\player\in\Player})$, if and only if for every player $\player$, the following conditions hold:
\begin{itemize}
\item[(i)] For all signals $\signal_\player$ such that $\nu(\signal_\player)>0$, 
    \[
    \aplan_\player \left(\BR\left(\nu_{\signal_\player}\right)\vert \signal_\player\right)=1.
    \]
\item[(ii)] For all signals $\signal_\player$ and $\signal_\player^\prime$ such that $\nu(\signal_\player)>0$ and $\nu(\signalb_\player)>0$,
    \[
\nu_{\signal_\player}\neq \nu_{\signalb_\player}\quad\text{implies}\quad \BR\left(\nu_{\signal_\player}\right)\cap \BR\left(\nu_{\signalb_\player}\right)=\varnothing.
    \]
\end{itemize}
In addition, for every player $\player$ and scalar $\lambda_\player\in (0,1]$, one can choose $\Exper_\player=\Delta(\Signal_\player)^{\Corstate\times\Paystate}$ and $\icost_\player:\Delta(\Signal_\player)^{\Corstate\times\Paystate}\rightarrow \mathbb{R}_+$ such that
\begin{align*}
\icost_\player(\exper_\player) &= \lambda_\player\left( \grossval_\player(\nu) -\noinfoval_\player(\nu)\right),\\
\max_{\experb_\player\in \Delta(\Signal_\player)^{\Corstate\times\Paystate}}\icost_\player(\experb_\player) & \leq  \lambda_\player+\hat{v}_\player(\nu)-\left[(1-\lambda_\player)\grossval_\player(\nu)+\lambda_\player\noinfoval_\player(\nu)\right].
\end{align*}
\end{lemma}

\begin{proof} The proof of Lemma \ref{lem:multi_agent} builds on Lemmas \ref{lemma:single_agent_plus} and \ref{lemma:single_agent_bound}. Paralleling the notation in Section \ref{subsec:singla_agent}, for every player $\player$ we define an auxiliary decision problem. Given $\Omega : = \Corstate \times \Paystate$, we take $\rho \in \Delta(\Omega)$ and  $w_\player:\Act_\player\times\Omega \rightarrow \mathbb{R}$ as follows: 
\begin{align}
\label{eq:multi-to-single-agent - prior}
\rho(\corstate,\paystate)& = \corprior(\corstate\vert\paystate)\payprior(\paystate),\\
\label{eq:multi-to-single-agent - utility}
w_\player(\act_\player,\corstate,\paystate)&=\sum_{\act_{-\player},\signal_{-\player}} \util_\player(\act,\paystate) \left[\prod_{\playerb\neq \player} \aplan_{\playerb}(\act_{\playerb}|\signal_\playerb)\exper_\playerb(\signal_\playerb\vert \corstate,\paystate)\right].
\end{align}
In the single-agent problem, $\Omega$ is the set of possible states, $\rho$ is the prior distribution of the state, and $w_\player$ is the utility function ($\Act_\player$ remains the set of feasible actions). Like in Section \ref{subsec:singla_agent}, given a signal $\signal_\player$ generated by $\exper_\player$ with positive probability, we denote by $\mu_{\signal_\player}\in\Delta(\Omega)$ the conditional distribution of $\omega$, and we write $\BR(\belief_{\signal_\player})$ for the best response set.

We will use the following claims, which relate the single-agent problem to the primitive many-player environment.
\begin{claim}\label{claim:equal_br}
For all signals $\signal_{\player}$ such that $\nu(\signal_\player)>0$,
\[
\BR\left(\belief_{\signal_\player}\right)=\BR\left(\nu_{\signal_\player}\right).
\]
\end{claim}
\begin{proof}[Proof of the claim]
For all actions $\act_\player$, we have 
\begin{equation*}
\begin{split}
\sum_{\omega}w_\player(\act_\player,\omega)\belief_{\signal_\player}(\omega) 
& = 
\sum_{\act_{-\player},\signal_{-\player},\corstate,\paystate} \util_\player(\act,\paystate) 
\left[\prod_{\playerb\neq \player} \aplan_{\playerb}(\act_{\playerb}|\signal_\playerb)\exper_\playerb(\signal_\playerb\vert \corstate,\paystate)
\right]
\frac{\exper_\player(\signal_\player|\corstate,\paystate)\corprior(\corstate|\paystate)\payprior(\paystate)
}
{
\nu(\signal_\player)
}
\\ & =
\sum_{\act_{-\player},\signal_{-\player},\corstate,\paystate} \util_{\player}(\act,\paystate) \nu_{\signal_\player}(\act_{-\player},\signal_{-\player},\corstate,\paystate).
\end{split}
\end{equation*}
We conclude that $\BR(\belief_{\signal_\player})=\BR(\nu_{\signal_\player})$.
\end{proof}

\begin{claim}\label{claim:equal_beliefs}
For all signals $\signal_{\player}$ and $\signal_{\player}^\prime$ such that $\nu(\signal_\player)>0$ and $\nu(\signalb_\player)>0$,
\[
    \nu_{\signal_\player}=\nu_{\signalb_\player} \text{ if and only if } \mu_{\signal_\player} =  \mu_{\signalb_\player}.
\]
\end{claim}
\begin{proof}[Proof of the claim]
Notice first that for all $\omega$,
\begin{align*}
\mu_{\signal_\player}(\omega) & =\sum_{\act_{-\player},\signal_{-\player}}\nu_{\signal_\player}(\act_{-\player},\signal_{-\player},\omega),\\
\mu_{\signal^\prime_\player}(\omega) & =\sum_{\act_{-\player},\signal_{-\player}}\nu_{\signalb_\player}(\act_{-\player},\signal_{-\player},\omega).
\end{align*}
Thus, $\nu_{\signal_\player} = \nu_{\signalb_\player}$ implies 
$\mu_{\signal_\player} = \mu_{\signalb_\player}$. For the converse direction, suppose $\mu_{\signal_\player}=\mu_{\signalb_\player}$. Then one obtains the following equality chain for all $\act_{-\player}$, $\signal_{-\player}$, $\corstate$, and $\paystate$:
\[
\begin{split}
\nu_{\signal_{\player}}(\act_{-\player},\signal_{-\player},\corstate,\paystate)
& = 
\frac{
\exper_{\player}(\signal_\player|\corstate,\paystate)\corprior(\corstate|\paystate)\payprior(\paystate)\prod_{\playerb\neq \player}\aplan_{\playerb}(\act_{\playerb}|\signal_{\playerb})\exper_\playerb(\signal_{\playerb}|\corstate,\paystate)
}
{
\phi(\signal_\player)
}
\\ & = 
\frac{
\exper_{\player}(\signal_\player|\corstate,\paystate)\rho(\corstate,\paystate)
}
{
\nu(\signal_\player)
}
\prod_{\playerb\neq \player}\aplan_{\playerb}(\act_{\playerb}|\signal_{\playerb})\exper_\playerb(\signal_{\playerb}|\corstate,\paystate)
\\ & =
\mu_{\signal_\player}(\corstate,\paystate)
\prod_{\playerb\neq \player}\aplan_{\playerb}(\act_{\playerb}|\signal_{\playerb})\exper_\playerb(\signal_{\playerb}|\corstate,\paystate)
\\ & = 
\mu_{\signal'_\player}(\corstate,\paystate)
\prod_{\playerb\neq \player}\aplan_{\playerb}(\act_{\playerb}|\signal_{\playerb})\exper_\playerb(\signal_{\playerb}|\corstate,\paystate)
= \nu_{\signalb_{\player}}(\act_{-\player},\signal_{-\player},\corstate,\paystate).
\end{split}
\]
We deduce that $\nu_{\signal_\player}=\nu_{\signalb_\player}$.
\end{proof}

We are now ready to complete the proof of Lemma \ref{lem:multi_agent}. For the ``only if'' part, notice that if $(\exper,\aplan)$ is an equilibrium of $(\BGame,\IT)$, then for every player $\player$, $(\exper_\player,\aplan_\player)$ solve \eqref{eq:single_agent} for the setting defined above. Lemma~\ref{lemma:single_agent_plus} then delivers two facts. First, for all signals $\signal_{\player}$ such that $\nu(\signal_\player)>0$,
\[
\aplan_{\player}(\BR(\nu_{\signal_\player})|\signal_\player) =
\aplan_{\player}(\BR(\mu_{\signal_\player})|\signal_\player) = 1,
\]
where the first equality follows from Claim \ref{claim:equal_br}, and the second from Lemma~\ref{lemma:single_agent_plus}-(i). We conclude Lemma~\ref{lem:multi_agent}-(i) holds. Second, for all signals $\signal_{\player}$ and $\signal_{\player}^\prime$ such that $\nu(\signal_\player)>0$ and $\nu(\signalb_\player)>0$, if $\nu_{\signal_\player}\neq \nu_{\signalb_\player}$ then $\mu_{\signal_\player}\neq \mu_{\signalb_\player}$ (by Claim \ref{claim:equal_beliefs}), which implies 
\[
\varnothing = \BR(\mu_{\signal_\player})\cap  \BR(\mu_{\signalb_\player}) = \BR(\nu_{\signal_\player})\cap  \BR(\nu_{\signalb_\player}),
\]
where the first equality follows from Lemma~\ref{lemma:single_agent_plus}-(ii), and the second from Claim \ref{claim:equal_br}. Thus, Lemma~\ref{lem:multi_agent}-(ii) also holds. 

For the ``if'' part, suppose that Lemma~\ref{lem:multi_agent}-(i) and Lemma~\ref{lem:multi_agent}-(ii) hold. Reasoning as above, one can use Claims \ref{claim:equal_br} and \ref{claim:equal_beliefs} to obtain that Lemma~\ref{lemma:single_agent_plus}-(i) and Lemma~\ref{lemma:single_agent_plus}-(ii) are satisfied. It follows that for every player $\player$, we can find a flexible $\Exper_\player$ and a monotone $\icost_{\player}$ such that $(\exper_\player,\aplan_\player)$ solves \eqref{eq:single_agent} for the setting defined above. In addition, using Lemma \ref{lemma:single_agent_bound}, we can assume that $\Exper_\player=\Delta(\Signal_\player)^{\Omega}$ and, given $\lambda_i\in (0,1]$, that $\icost_{\player}:\Delta(\Signal_\player)^{\Omega}\rightarrow\mathbb{R}_+$ is such that 
\[
\icost_\player(\exper_\player)  =\lambda_\player\left(\overline{w}_\player(\mu_\player)- \underline{w}_\player(\mu_\player)\right)=\lambda_\player\left(\overline{v}_\player(\nu)- \underline{v}_\player(\nu_\player)\right),
\]
and for all $\experb_\player\in\Delta(\Signal_\player)^\Omega$,
\begin{align*}
 \icost_\player(\experb_\player)& \leq \lambda_\player + \hat{w}_\player(\mu_\player) - \left[(1-\lambda_\player)\overline{w}_\player(\mu_\player)+\lambda_\player \underline{w}_\player(\mu_\player)\right] \\
& = \lambda_\player + \hat{v}_\player(\nu) - \left[(1-\lambda_\player)\overline{v}_\player(\nu)+\lambda_\player \underline{v}_\player(\nu)\right].
\end{align*}
Going back to the many-player environment, we construct the desired technology $\IT$ by combining the information structure $\mathcal{S}$ with $(\Exper_\player,\icost_\player)_{\player\in\Player}$: $\IT:=(\Corstate,\corprior,(\Signal_\player,\Exper_\player,\icost_\player)_{\player\in\Player})$. \end{proof}

\subsection{Representing outcomes with information structures}

Fix a base game $\BGame$. An information structure $\mathcal{S}=(\Corstate,\corprior,(\Signal_\player,\exper_\player)_{\player\in\Player})$ and a profile of action plans $\aplan=(\aplan_\player)_{\player\in\Player}$ \textbf{represent} an outcome $\out \in \Delta(\Act\times \Paystate)$ if for all $\act\in\Act$ and $\paystate\in\Paystate$,
\[
\out(\act,\paystate)=\sum_{\corstate,\signal}\left[\prod_{\player\in\Player} \aplan_\player(\act_\player|\signal_\player) \exper_\player(\signal_\player|\corstate,\paystate)\right]
\corprior(\corstate|\paystate)\payprior(\paystate).
\]
Each outcome admits multiple representations. For example, as is well known, one can represent $\out$ as follows. Set $\Corstate=\Act$, and for every $\player \in \Player$, take $\Signal_\player$ so that $\Signal_\player \supseteq \Act_\player$. For every $\corstate\in\Act$ and $\paystate\in\Paystate$, set $\corprior(\corstate\vert\paystate)=\out(\corstate,\paystate)/\payprior(\theta)$. For every $\player\in\Player$, $\signal_\player\in\Signal_\player$, $\corstate\in\Act$, and $\paystate\in\Paystate$, take
\[
\exper_\player(\signal_\player\vert \corstate,\paystate)=
\begin{cases}
1 &\text{if }\signal_\player=\corstate_\player,\\
0 &\text{otherwise.}
\end{cases}
\]
And finally, for every $\player\in\Player$ and $\signal_\player\in\Act_\player$, 
\[
\aplan_\player(\act_\player\vert \signal_\player) = \begin{cases}
1 &\text{if }\act_\player=\signal_\player,\\
0 &\text{otherwise.}
\end{cases}
\]

This outcome representation has been used, among others, by \cite{aumann1987correlated} and \cite{bergemann2016bayes} to study games of incomplete information. Next, we introduce outcome representations that are useful to analyze information acquisition games.

Let $\PartA_\player^\out$ be the partition of $\Act_\player$ with the following properties: for $\act_\player\in\supp_\player(\out)$,
\[
\PartA^\out_\player(\act_\player)=\{\actb_\player\in\supp_\player(\out):\out_{\act_\player}=\out_{\actb_\player}\},
\]
and for $\act_\player\notin\supp_\player(\out)$,
\[
\PartA^\out_\player(\act_\player)=\Act_\player\setminus\supp_\player(\out).
\]
In the formulas above, $\PartA^\out_\player(\act_\player)$ is the cell of the partition that contains $\act_\player$.

A representation $(\mathcal{S},\sigma)$ of an outcome $\out$ is \textbf{canonical} if the following conditions hold:

\begin{itemize}
\item $\Corstate\subseteq \prod_{\player\in\Player}\Delta(\PartA^{\out}_\player)$;
\item for every $\player\in\Player$, 
\begin{equation}\label{eq:canonical_sig}
 \PartA^{\out}_\player\subseteq \Signal_\player;
\end{equation}
\item for every $\player\in \Player$, $\signal_\player\in \Signal_\player$, $\corstate\in \Corstate$, and $\theta\in\Theta$,
\begin{equation}\label{eq:canonical_exp}
\exper_\player(\signal_\player|\corstate,\paystate)=
\begin{cases}
\corstate_\player(\signal_\player) &\text{if } \signal_\player\in \PartA^{\out}_\player,\\
0 &\text{otherwise};
\end{cases}
\end{equation}
\item for every $\player\in\Player$, $\signal_\player\in\PartA^{\out}_\player$, and $\act_\player\in\Act_\player$,
\begin{equation}\label{eq:canonical_plan}
\aplan_\player(\act_\player|\signal_\player) = 
\begin{cases}
    \frac{\out(\act_\player)}{\sum_{\actb_\player \in \signal_\player}\out(\actb_\player)} & \text{if }\act_\player \in \signal_\player\text{ and }\act_\player\in \supp_\player(\out),\\
    \frac{1}{\vert\Act_\player\vert-\vert\supp_\player(\out)\vert} & \text{if }\act_\player \in \signal_\player\text{ and }\act_\player\notin \supp_\player(\out),\\
    0 & \text{otherwise.}
\end{cases}
\end{equation}
\end{itemize}
Observe that canonical representations of the same outcome (essentially) differ only in the specification of $\Corstate$ and $\corprior$, the correlation state.
\begin{lemma}\label{lem:p-canonical-existence}
Every outcome $\out \in \Delta_{\payprior}(\Act\times \Paystate)$ admits a canonical representation.
\end{lemma}
To prove the existence of a canonical representation, we use the following decomposition of $\out$ in terms of the partitions $\PartA^\out_\player$, $\player\in\Player$:
\begin{lemma}\label{lem:p-decomposition}
Fix an outcome $p\in\Delta_{\payprior}(\Act\times \Paystate)$. For all $\act\in\Act$ and $\paystate\in\Paystate$, 
\begin{equation}\label{eq:02-02-2023c}
    \out(\act,\paystate)  =
\left[\prod_{\player\in\Player}\aplan_\player(\act_\player|\PartA^\out_\player(\act_\player))\right]\sum_{\actb\in \PartA^\out(\act)}\out(\actb,\paystate),
\end{equation}
where $\aplan_\player(\act_\player|\PartA^\out_\player(\act_\player))$ defined by (\ref{eq:canonical_plan}), and $\PartA^\out(\act)=\prod_{\player\in\Player}\PartA^\out_\player(\act_\player)$.
\end{lemma}
\begin{proof}
It suffices to show that for all $\act\in\Act$, $\paystate\in\Paystate$, and $\player\in\Player$,
\begin{equation}\label{eq:04_20_2023}
\out(\act,\paystate)=\aplan_\player\left(\act_\player|\PartA^\out_\player(\act_\player)\right)\sum_{\actb_{\player}\in \PartA^\out_\player(\act_\player)}\out(\actb_{\player},\act_{-\player},\paystate).
\end{equation}
The desired result then follows from reasoning by induction on the number of players.

If $\out(\act_\player)=0$, then 
\[
\out(\act,\paystate)=0=\sum_{\actb_{\player}\in \PartA^\out_\player(\act_\player)}\out(\actb_{\player},\act_{-\player},\paystate)=\aplan_\player\left(\act_\player|\PartA^\out_\player(\act_\player)\right)\sum_{\actb_{\player}\in \PartA^\out_\player(\act_\player)}\out(\actb_{\player},\act_{-\player},\paystate).
\]
Thus, if $\out(\act_\player)=0$, then (\ref{eq:04_20_2023}) holds. 

Suppose now $\out(\act_\player)>0$. By Bayes rule, 
\[
\out(\act,\paystate) = \out(\act_\player) \out_{\act_\player}(\act_{-\player},\paystate).
\]
Recall that $\actb_\player\in \PartA^\out_\player(\act_\player)$ if and only if $\out_{\actb_\player}=\out_{\act_\player}$; in addition, $\sum_{\actb_{\player}\in \PartA^\out_\player(\act_\player)}\aplan_\player\left(\actb_\player\vert \PartA^\out_\player(\act_\player)\right)=1$. Therefore, 
\[
\begin{split}
\out(\act_{\player},\actb_{-\player},\paystate) 
& =
\out(\act_\player)\out_{\act_{\player}}(\actb_{-\player},\paystate) 
\\ & =
\out(\act_\player)\sum_{\actb_{\player}\in \PartA^\out_\player(\act_\player)}\aplan_\player\left(\actb_\player\vert \PartA^\out_\player(\act_\player)\right)\out_{\act_{\player}}(\act_{-\player},\paystate)
\\ & =
\out(\act_\player)\sum_{\actb_{\player}\in \PartA^\out_\player(\act_\player)}\aplan_\player(\actb_\player\vert \PartA^\out_\player(\act_\player))\out_{\actb_{\player}}(\act_{-\player},\paystate).
\end{split}
\]
Substituting in the definition of $\aplan_\player$, one obtains (\ref{eq:04_20_2023}).\end{proof}

\begin{proof}[Proof of Lemma \ref{lem:p-canonical-existence}]
Take 
\begin{equation}\label{eq:spec_can_z}
\Corstate=\prod_{\player\in\Player}\left\{\delta_{B_\player}:B_\player\in\PartA^\out_\player\right\},
\end{equation}
where $\delta_{B_\player}\in \Delta(\PartA^\out_\player)$ is the Dirac measure concentrated on $B_\player$. For $\corstate=(\delta_{B_\player})_{\player\in\Player}$ and $B=\prod_{\player\in\Player} B_\player$, define
\begin{equation}\label{eq:spec_can_ze}
\corprior(\corstate\vert\paystate)=\sum_{\act\in  B}\frac{\out(\act,\paystate)}{\payprior(\paystate)}.
\end{equation}
For $\player\in\Player$, define $\Signal_\player$, $\exper_\player$, and $\aplan_\player$ as in (\ref{eq:canonical_sig})-(\ref{eq:canonical_plan}). Let $\mathcal{S}:=(\Corstate,\corprior,(\Signal_\player,\exper_\player)_{\player\in\Player})$ and $\aplan:=(\aplan_\player)_{\player\in\Player}$.

We need to show that for all $\act\in\Act$ and $\paystate\in\Paystate$,
\begin{equation}\label{eq:correct_outcome}
\out(\act,\paystate)=\sum_{\corstate,\signal}\left[\prod_{\player\in\Player} \aplan_\player(\act_\player|\signal_\player) \exper_\player(\signal_\player|\corstate,\paystate)\right]
\corprior(\corstate|\paystate)\payprior(\paystate).
\end{equation}
Using Lemma \ref{lem:p-decomposition} we get that
\begin{align*}
 \out(\act,\paystate) & = \left[\prod_{\player\in\Player}\aplan_\player(\act_\player\vert\PartA^\out_\player(\act_\player))\right]\sum_{\corstate}\left[\prod_{\player\in\Player}\corstate_{\player}(\PartA^\out_\player(\act_\player))\right]\corprior(\corstate\vert\paystate)\payprior(\paystate)\\
 & = 
\sum_{\corstate}\left[\prod_{\player\in\Player}\aplan_\player(\act_\player\vert\PartA^\out_\player(\act_\player))\exper_\player(\PartA^\out_\player(\act_\player)|\corstate,\paystate)\right]\corprior(\corstate\vert\paystate)\payprior(\paystate)\\
& = \sum_{\corstate,\signal}\left[\prod_{\player\in\Player}\aplan_\player(\act_\player\vert\signal_\player)\exper_\player(\signal_\player|\corstate,\paystate)\right]\corprior(\corstate\vert\paystate)\payprior(\paystate),
\end{align*}
where the first equality combines (\ref{eq:02-02-2023c}) with (\ref{eq:spec_can_z}) and (\ref{eq:spec_can_ze}), the second equality follows from (\ref{eq:canonical_exp}), and the last equality from (\ref{eq:canonical_plan}). We conclude \eqref{eq:correct_outcome} holds. \end{proof}

The next lemma concerns canonical representations of separated BCE. With a slight abuse of notation, we define
\[
\hat{v}_\player(\mathcal{S},\aplan)= \sum_{\corstate,\paystate}\corprior(\corstate|\paystate)\payprior(\paystate)\left[\max_{\act_\player}\sum_{\signal_{-\player},\act_{-\player}}\util_\player(\act_\player,\act_{-\player},\paystate)\prod_{\playerb\neq\player} \aplan_\playerb(\act_\playerb|\signal_\playerb) \exper_\playerb(\signal_\playerb|\corstate,\paystate)\right].
\]

\begin{lemma}\label{lem:p-canonical}
Let $\mathcal{S}=(\Corstate,\corprior,(\Signal_\player,\exper_\player)_{\player\in\Player})$ and $\aplan$ be a canonical representation of an outcome $\out$. If $\out$ is a separated BCE, then for every player $\player$ there exist a flexible $\Exper_\player$ and a monotone $\icost_\player$ such that $(\exper,\aplan)$ is an equilibrium of $(\BGame,\IT)$ with $\IT=(\Corstate,\corprior,(\Signal_\player,\Exper_\player,\icost)_{\player\in\Player})$. Moreover, for every player $\player$, one can set $\Exper_\player=\Delta(\Signal_\player)^{\Corstate\times\Paystate}$ and, given $\lambda_\player\in (0,1]$, choose $\icost_\player:\Delta(\Signal_\player)^{\Corstate\times\Paystate}\rightarrow\mathbb{R}_+$ such that 
\begin{align*}
\icost_{\player}(\exper_\player) & = 
\lambda_\player(\grossval_\player(\out)
-
\noinfoval_\player(\out)),\\
\max_{\experb_\player\in \Delta(\Signal_\player)^{\Corstate\times\Paystate}} \icost_{\player}(\experb_\player) & \leq \lambda_\player + \hat{v}_\player(\mathcal{S},\aplan) -[(1-\lambda_\player)\grossval_\player(\out)+\lambda_\player \noinfoval_\player(\out)].
\end{align*}
\end{lemma}
\newcommand{\gpp}{\Gamma}
\newcommand{\pgpp}{\phi}

\begin{proof}  We need to show that the conditions of Lemma~\ref{lem:multi_agent} hold; like in Section \ref{subsec:embedding}, we denote by $\nu\in\Delta(\Act\times\Signal\times\Corstate\times\Paystate)$ the probability measure over actions, signals, and states induced by $\mathcal{S}=(\Corstate,\corprior,(\Signal_\player,\exper_\player)_{\player\in\Player})$ and $\aplan$. 

We will use the following claim:
\begin{claim}\label{claim:br_br}
For all $\player\in \Player$, $\signal_\player\in \Signal_\player$ such that $\nu(\signal_\player)>0$, and $\act_\player \in \signal_\player$, $\BR(\nu_{\signal_\player}) = \BR(\out_{\act_\player})$.
\end{claim}
\begin{proof}[Proof of the claim]
Fix $\player\in \Player$, $\signal_\player\in \Signal_\player$ such that $\nu(\signal_\player)>0$, and $\act_\player \in \signal_\player$. Observe first that for all $\actb_\player\in \signal_\player$,
\[
\out(\actb_\player)=\sum_{\signalb_\player,\corstate,\paystate}\aplan_\player(\actb_\player\vert\signalb_\player)\exper_\player(\signalb_\player\vert\corstate,\paystate)\corprior(\corstate\vert\paystate)\payprior(\paystate)=\sum_{\corstate,\paystate}\aplan_\player(\actb_\player\vert\signal_\player)\exper_\player(\signal_\player\vert\corstate,\paystate)\corprior(\corstate\vert\paystate)\payprior(\paystate),
\]
where the first equality holds because $(\mathcal{S},\sigma)$ represents $\out$, and the second equality because $\aplan_\player(\actb_\player\vert\signalb_\player)>0$ if and only if $\actb_\player\in\signalb_\player$. Thus,
\begin{equation}\label{eq:monday_morning}
\sum_{\actb_\player\in\signal_\player}\out(\actb_\player)=\nu(\signal_\player).
\end{equation}
Reasoning as above, we have that for all $\act_{-\player}\in\Act_{-\player}$ and $\paystate\in\Paystate$,
\[
\out(\act,\paystate)=\sum_{\signalb_\player,\signal_{-\player},\corstate}\nu(\act,\signalb_\player,\signal_{-\player},\corstate,\paystate)=\sum_{\signal_{-\player},\corstate}\nu(\act,\signal_\player,\signal_{-\player},\corstate,\paystate).
\]
Then, we obtain that
\begin{align*}
\out(\act,\paystate)
                    & = \aplan_\player(\act_\player\vert \signal_\player)\sum_{\actb_\player,\signal_{-\player},\corstate}\nu(\actb_\player,\act_{-\player},\signal,\corstate,\paystate) \\
                    & = \frac{\out(\act_\player)}{\sum_{\actb_\player\in\signal_\player}\out(\actb_\player)}\sum_{\actb_\player,\signal_{-\player},\corstate}\nu(\actb_\player,\act_{-\player},\signal,\corstate,\paystate)=\out(\act_\player) \sum_{\signal_{-\player},\corstate}\nu_{\signal_\player}(\act_{-\player},\signal_{-\player},\corstate,\paystate),
\end{align*}
where the first equality follows from $\act_\player$ being conditionally independent of the other variable given $\signal_\player$, and the second equality from (\ref{eq:monday_morning}). In sum, we deduce that 
\[
\out_{\act_{\player}}(\act_{-\player},\paystate)=\sum_{\signal_{-\player},\corstate}\nu_{\signal_\player}(\act_{-\player},\signal_{-\player},\corstate,\paystate).
\]
It follows that $\BR(\nu_{\signal_\player}) = \BR(\out_{\act_\player})$.
\end{proof}

We now complete the proof by verifying the conditions of Lemma~\ref{lem:multi_agent}. For Lemma \ref{lem:multi_agent}-(i), take $\signal_\player \in \Signal_\player$ and $\act_\player\in\Act_\player$ such that $\nu(\signal_\player)>0$ and $\aplan_\player(\act_\player\vert\signal_\player)>0$. By the construction of the action plan, $\act_\player\in\signal_\player$. By the obedience constraint for $\out$, $\act_\player\in BR(\out_{\act_\player})$. From Claim \ref{claim:br_br}, $BR(\out_{\act_\player})=BR(\nu_{\signal_\player})$. Therefore, $\act_\player\in BR(\nu_{\signal_\player})$. We deduce that Lemma \ref{lem:multi_agent}-(i) holds.

To verify Lemma \ref{lem:multi_agent}-(ii), take $\signal_\player,\signal_\player^\prime\in \Signal_\player$ such that $\nu(\signal_\player)>0$, $\nu(\signalb_\player)>0$, and $\nu_{\signal_\player}\neq\nu_{\signalb_\player}$. Let $\act_\player\in \signal_\player$ and $\actb_\player\in \signal_\player^\prime$. Since $\nu_{\signal_\player}\neq\nu_{\signalb_\player}$, we must have $\signal_\player\neq \signal_\player^\prime$ and, therefore, $\out_{\act_\player}\neq \out_{\actb_\player}$. By the separation constraint for $\out$, $BR(\out_{\act_\player})\cap BR(\out_{\actb_\player})=\varnothing$. From Claim \ref{claim:br_br}, $BR(\out_{\act_\player})=BR(\nu_{\signal_\player})$ and $BR(\out_{\actb_\player})=BR(\nu_{\signal_\player^\prime})$. We deduce $BR(\nu_{\signal_\player})=BR(\nu_{\signal_\player^\prime})=\varnothing$: Lemma \ref{lem:multi_agent}-(ii) holds.\end{proof}

\section{Proof of Theorem \ref{thm:mon_tech}}

The ``if'' statement of Theorem \ref{thm:mon_tech} follows immediately from Lemma \ref{lem:p-canonical}---to apply the lemma, one needs a canonical representation for $\out$, whose existence is guaranteed by Lemma \ref{lem:p-canonical-existence}. 

Next we prove the ``only if'' statement of Theorem \ref{thm:mon_tech}. Let $(\out,\payvector)$ be the outcome-value pair induced by an equilibrium $(\exper,\aplan)$ of a rational inattention game $(\BGame,\itech)$, with $\itech=(\Corstate,\corprior,(\Signal_\player,\Exper_\player,\icost_\player)_{\player\in\Player})$. It follows the conditions of Lemma~\ref{lem:multi_agent} must hold for the information structure $\mathcal{S}=(\Corstate,\corprior,(\Signal_\player,\exper_\player)_{\player\in\Player})$ and the profile of action plans $\aplan$. We now use this fact to prove a connection between the best responses in the information acquisition game and optimal behavior given the mediator's recommendation. Like in Section \ref{subsec:embedding}, we denote by $\nu\in\Delta(\Act\times\Signal\times\Corstate\times\Paystate)$ the measure over actions, signals, and states induced by $(\mathcal{S},\aplan)$. For $\player\in\Player$ and $\act_\player\in \supp_\player(\out)$, we denote by $X_{\act_\player}$ the set of positive-probability signals that makes player $\player$ take action $\act_\player$:
\[
\Signal_{\act_\player}:=\{\signal_\player:\nu(\signal_\player)>0\text{ and }\aplan_\player(\act_\player|\signal_\player)>0\}.
\]

\begin{lemma}\label{lem:BR in info-game contains BR from outcome}
For every $\player\in\Player$ and $\act_\player\in \supp_\player(\out)$,

\[
\BR(\out_{\act_\player}) = \bigcap_{\signal_\player\in X_{\act_\player}} \BR(\nu_{\signal_\player}).
\]
\end{lemma}
\begin{proof}
First we show that 
\[
\BR(\out_{\act_\player}) \subseteq \bigcap_{\signal_\player\in X_{\act_\player}} \BR(\nu_{\signal_\player}).
\]
Take $\actb_\player\in BR(\out_{\act_\player})$. By Lemma~\ref{lem:multi_agent}-(i), $\act_\player\in \BR(\nu_{\signal_\player})$ for all $\signal_\player\in\Signal_{\act_\player}$, which implies
\begin{equation}\label{eq:pointwise}
\sum_{\act_{-\player},\signal_{-\player},\corstate,\paystate}
\util_\player(\act_{\player},\act_{-\player},\paystate)
\nu_{\signal_\player}(\act_{-\player},\signal_{-\player},\corstate,\paystate)
\geq 
\sum_{\act_{-\player},\signal_{-\player},\corstate,\paystate}
\util_\player(\actb_{\player},\act_{-\player},\paystate)
\nu_{\signal_\player}(\act_{-\player},\signal_{-\player},\corstate,\paystate).
\end{equation}
Letting 
\[
\nu(\signal_\player|\act_\player) = \frac{\aplan_\player(\act_\player\vert \signal_\player)\nu(\signal_\player)}{\out(\act_\player)},
\]
simple algebra shows that 
\begin{equation}\label{eq:bayes_rule}
\out_{\act_\player}(\act_{-\player},\paystate)=\sum_{\signal_\player:\nu(\signal_\player)>0}\nu(\signal_\player|\act_\player) 
\left[ 
\sum_{\signal_{-\player},\corstate}\nu_{\signal_\player}(\act_{-\player},\signal_{-\player},\corstate,\paystate)
\right].
\end{equation}
Therefore,
\begin{align*}
\sum_{\signal_\player:\nu(\signal_\player)>0}\nu(\signal_\player|\act_\player) 
\left(\sum_{\act_{-\player},\signal_{-\player},\corstate,\paystate}
\util_\player(\act_{\player},\act_{-\player},\paystate)
\nu_{\signal_\player}(\act_{-\player},\signal_{-\player},\corstate,\paystate) \right) 
& = \sum_{\act_{-\player},\paystate} \util_\player(\act_\player,\act_{-\player},\theta)\out_{\act_\player}(\act_{-\player},\paystate),
\\ 
\sum_{\signal_\player:\nu(\signal_\player)>0}\nu(\signal_\player|\act_\player) 
\left(\sum_{\act_{-\player},\signal_{-\player},\corstate,\paystate}
\util_\player(\actb_{\player},\act_{-\player},\paystate)
\nu_{\signal_\player}(\act_{-\player},\signal_{-\player},\corstate,\paystate) \right) 
& = \sum_{\act_{-\player},\paystate} \util_\player(\actb_\player,\act_{-\player},\theta)\out_{\act_\player}(\act_{-\player},\paystate).
\end{align*}
As a consequence, since $\actb_\player\in BR(\out_{\act_\player})$, we have that
\begin{align*}
   & \sum_{\signal_\player:\nu(\signal_\player)>0}\nu(\signal_\player|\act_\player) 
\left(\sum_{\act_{-\player},\signal_{-\player},\corstate,\paystate}
\util_\player(\act_{\player},\act_{-\player},\paystate)
\nu_{\signal_\player}(\act_{-\player},\signal_{-\player},\corstate,\paystate) \right) \\ 
\leq & \sum_{\signal_\player:\nu(\signal_\player)>0}\nu(\signal_\player|\act_\player) 
\left(\sum_{\act_{-\player},\signal_{-\player},\corstate,\paystate}
\util_\player(\actb_{\player},\act_{-\player},\paystate)
\nu_{\signal_\player}(\act_{-\player},\signal_{-\player},\corstate,\paystate) \right).
\end{align*}
It follows (\ref{eq:pointwise}) holds with equality for all $\signal_\player\in\Signal_{\act_\player}$. Since $\act_\player\in BR(\nu_{\signal_\player})$,  $\actb_\player\in BR(\nu_{\signal_\player})$.

Now we show that 
\[
\BR(\out_{\act_\player}) \supseteq \bigcap_{\signal_\player\in \Signal_{\act_\player}} \BR(\nu_{\signal_\player}).
\]
Let $\actb_\player\in \Act_\player$ such that $\actb_\player\in\BR(\nu_{\signal_\player})$ for all $\signal_\player\in\Signal_{\act_\player}$, that is, 
\[
\sum_{\act_{-\player},\signal_{-\player},\corstate,\paystate}
\util_\player(\actc_{\player},\act_{-\player},\paystate)
\nu_{\signal_\player}(\act_{-\player},\signal_{-\player},\corstate,\paystate)
\leq 
\sum_{\act_{-\player},\signal_{-\player},\corstate,\paystate}
\util_\player(\actb_{\player},\act_{-\player},\paystate)
\nu_{\signal_\player}(\act_{-\player},\signal_{-\player},\corstate,\paystate).
\]
for all $\actc_\player\in\Act_\player$. Averaging across inequalities, we obtain that
\begin{align*}
& \sum_{\signal_\player\in\Signal_{\act_\player}}\nu(\signal_\player\vert\act_\player)\sum_{\act_{-\player},\signal_{-\player},\corstate,\paystate}
\util_\player(\actc_{\player},\act_{-\player},\paystate)
\nu_{\signal_\player}(\act_{-\player},\signal_{-\player},\corstate,\paystate)\\
\leq &
\sum_{\signal_\player\in\Signal_{\act_\player}}\nu(\signal_\player\vert\act_\player)\sum_{\act_{-\player},\signal_{-\player},\corstate,\paystate}
\util_\player(\actb_{\player},\act_{-\player},\paystate)
\nu_{\signal_\player}(\act_{-\player},\signal_{-\player},\corstate,\paystate).
\end{align*}
Using (\ref{eq:bayes_rule}), we obtain that 
\[
\sum_{\act_{-\player},\paystate}\util_\player(\actc_\player,\act_{-\player},\paystate)\out_{\act_\player}(\act_{-\player},\paystate)
\leq 
\sum_{\act_{-\player},\paystate}\util_\player(\actb_\player,\act_{-\player},\paystate)\out_{\act_\player}(\act_{-\player},\paystate).
\]
We deduce that $\actb_\player\in \BR(\out_{\act_\player})$.
\end{proof}

Next, we use the above to show the equilibrium outcome must satisfy the obedience constraint and the separation constraint.

\begin{lemma}\label{lem:isasBCE}
The equilibrium outcome $\out$ is a sBCE.
\end{lemma}
\begin{proof}
To check the obedience constraint, take $\player\in\Player$ and $\act_\player\in\supp_\player(\out)$.
By Lemma~\ref{lem:multi_agent}-(i), $\act_\player\in \BR(\nu_{\signal_\player})$ for all $\signal_\player\in\Signal_{\act_\player}$. By Lemma~\ref{lem:BR in info-game contains BR from outcome}, $\act_\player\in \BR(\out_{\act_\player})$. 

To verify the separation constraint, let $\player\in\Player$ and $\act_\player,\actb_\player\in\supp_\player (\out)$ such that $\out_{\act_\player}\neq \out_{\actb_\player}$. Since $\out_{\act_\player}\neq \out_{\actb_\player}$, there must be a pair of signals $\signal_{\act_\player}\in\Signal_{\act_\player}$ and $\signal_{\actb_\player}\in \Signal_{\actb_\player}$ such that $\mu_{\signal_{\act_\player}}\neq \mu_{\signal_{\actb_\player}}$. By Lemma~\ref{lem:multi_agent}-(ii), 
\[
BR(\mu_{\signal_{\act_\player}})\cap BR(\mu_{\signal_{\actb_\player}})=\varnothing.
\]
By Lemma~\ref{lem:BR in info-game contains BR from outcome}, $BR(\out_{\act_\player})\cap BR(\out_{\actb_\player})=\varnothing$.
\end{proof}

The next lemma concludes the proof of the ``only if'' statement of Theorem~\ref{thm:mon_tech}.

\begin{lemma}\label{lem:valuesss}
For every player $\player$,
\[
\payvector_\player=\noinfoval_\player(\out)=\grossval_\player(\out)\quad\text{or}\quad\payvector_\player\in [\noinfoval_\player(\out),\grossval_\player(\out)).
\]
\end{lemma}
\begin{proof}
For every player $i$, $\icost_\player(\exper_\player)\geq 0$, which implies that $\payvector_\player\leq \grossval_\player(\out)$. In addition, uninformative experiments have zero cost by hypothesis.  Thus, since $(\exper_\player,\aplan_\player)$ is a best response to $(\exper_{-\player},\aplan_{-\player})$, we have that for any uninformative experiment $\exper_\player^\prime$,
\[
\payvector_\player  \geq \max_{\aplan^\prime_\player}\sum_{\act,\signal,\corstate,\paystate} \util_\player(\act,\paystate) \left[\aplan^\prime_{\player}(\act_\player|\signal_\player)\exper_\player(\signal_{\player}|\corstate,\paystate)\prod_{\playerb\neq \player} \aplan_{\playerb}(\act_{\playerb}|\signal_\playerb)\exper_{\playerb}(\signal_{\playerb}|\corstate,\paystate)\right]\corprior(\corstate|\paystate)\payprior(\paystate)
 = \noinfoval_\player(\out).
\]
Overall, we conclude that $\payvector_\player\in[\noinfoval_\player(\out),\grossval_\player(\out)]$

If $\noinfoval_\player(\out)=\grossval_\player(\out)$, then $\payvector_\player\in[\noinfoval_\player(\out),\grossval_\player(\out)]$ implies 
$
\payvector_\player=\noinfoval_\player(\out)=\grossval_\player(\out).
$
Suppose instead that $\noinfoval_\player(\out)<\grossval_\player(\out)$. If 
\[
\sum_{\act,\signal,\corstate,\paystate} \util_\player(\act,\paystate) \left[\prod_{\playerb\in\Player} \aplan_{\playerb}(\act_{\playerb}|\signal_\playerb)\exper_{\playerb}(\signal_{\playerb}|\corstate,\paystate)\right]\corprior(\corstate|\paystate)\payprior(\paystate)<\grossval_\player(\out),
\]
then obviously $\payvector_i<\grossval_\player(\out)$ because $\icost_{\player}(\exper_\player)\geq 0$. If, on the other hand,
\[
\sum_{\act,\signal,\corstate,\paystate} \util_\player(\act,\paystate) \left[\prod_{\playerb\in\Player} \aplan_{\playerb}(\act_{\playerb}|\signal_\playerb)\exper_{\playerb}(\signal_{\playerb}|\corstate,\paystate)\right]\corprior(\corstate|\paystate)\payprior(\paystate)=\grossval_\player(\out),
\]
then $\exper_\player$ cannot be uninformative because $\noinfoval_\player(\out)<\grossval_\player(\out)$. By monotonicity, $\icost_{\player}(\exper_\player)>0$, which implies that $\payvector_\player<\grossval_\player(\out)$. In sum, if $\noinfoval_\player(\out)<\grossval_\player(\out)$, then 
$
\payvector_\player\in [\noinfoval_\player(\out),\grossval_\player(\out)).
$
\end{proof}

\section{Proof of Theorem \ref{thm:genericity}}

Throughout this section, we fix the set of players $\Player$, the set of payoff states $\Paystate$, the prior $\payprior$ (with full support), and an action set $\Act_\player$ for every player $\player$. Given these, specifying the players' utilities is all that is left for defining a base game. As such, we use the profile of utility functions $\util=(\util_\player)_{\player\in\Player}$ as a shorthand for the base game it defines, writing for example $\BCE(\util)$ for the set of BCEs. We also denote by $\sBCE(\util)$ the set of sBCEs, by $\cl(\sBCE(\util))$ the closure of $\sBCE(\util)$, and by  $\Vert \util\Vert$ the Euclidean norm. 

\begin{lemma}\label{lem:poking}
For every $\util\in\mathbb{R}^{\Player\times\Act\times\Paystate}$, $\out\in \BCE(\util)$, and $\epsilon>0$, there exists $\utilb\in\mathbb{R}^{\Player\times\Act\times\Paystate}$ such
that $\Vert \util-\utilb\Vert\leq\epsilon$ and $\out\in \sBCE(\utilb)$.
\end{lemma}

\begin{proof}
For each player $\player\in \Player$, we consider a set $P_\player\subseteq\Delta(\Act_{-\player}\times\Paystate)$
given by
\[
P_\player=\left\{\out_{\act_\player}:\act_\player\in\supp_\player(\out)\right\}.
\]
Let $n_\player$ be the cardinality of $P_\player$ (of course, it could be that $n_\player$ is smaller than the cardinality of $\supp_\player(\out)$). Reasoning inductively,
we can find an enumeration $\out_{1},\ldots,\out_{n_\player}$ of the elements
of $P_{i}$ such that, for every $m_\player\in\{1,\ldots,n_{i}\}$, $\out_{m_\player}$
is an extreme point of the convex hull of $\left\{\out_{1},\ldots,\out_{m_\player}\right\}$. 

By an hyperplane separation theorem (e.g., \citealp{rockafellar1970convex}, Corollary 11.4.2) for every $m_\player\in\{2,\ldots,n_\player\}$
we can find a function $f_{m_\player}:A_{-\player}\times\Paystate\rightarrow\mathbb{R}$ such that
\begin{equation}\label{eq:separation_thm}
\sum_{\act_{-\player},\paystate}f_{m_{\player}}(\act_{-\player},\paystate)\out_{m_{\player}}(\act_{-\player},\paystate)>0\geq\max_{l_\player\in\{1,\ldots,m_\player-1\}}\sum_{\act_{-\player},\paystate}f_{m_\player}(\act_{-\player},\paystate)\out_{l_{\player}}(\act_{-\player},\paystate).
\end{equation}
For $m_{\player}=1$, we define $f_{1}(\act_{-\player},\paystate)=1$ for all $\act_{-\player}\in\Act_{-\player}$ and $\paystate\in\Paystate$. 

For every $l_\player\in\{1,\ldots,n_\player-1\}$, we choose $t_{l_{\player}}\in(0,1]$
such that for every $m_{\player}\in\{l_{\player}+1,\ldots,n_{\player}\}$,
\begin{equation}
\sum_{\act_{-\player},\paystate}f_{m_{\player}}(\act_{-\player},\paystate)\out_{m_{\player}}(\act_{-\player},\paystate)>t_{l_{\player}}\sum_{\act_{-\player}}f_{l_{\player}}(\act_{-\player})\out_{m_{\player}}(\act_{-\player}).\label{eq:scalar_way}
\end{equation}
We can choose such a $t_{l_{\player}}$ because the left-hand side of (\ref{eq:scalar_way})
is positive---see (\ref{eq:separation_thm}). For $l_{\player}=n_{\player}$,
we simply define $t_{n_{\player}}=1$.

For every $l_{\player}\in\{1,\ldots,n_{\player}\}$, we define
\[
s_{l_{\player}}=\prod_{m_{\player}=l_{\player}}^{n_{\player}}t_{m_{\player}}.
\]
Using (\ref{eq:scalar_way}), simple algebra shows that for every
$l_{\player}\in\{1,\ldots,n_{\player}-1\}$ and $m_{\player}\in\{l_{\player}+1,\ldots,n_{\player}\}$,
\begin{equation}
s_{m_{\player}}\sum_{\act_{-\player},\paystate}f_{m_{\player}}(\act_{-\player},\paystate)\out_{m_{\player}}(\act_{-\player},\paystate)>s_{l_{\player}}\sum_{\act_{-\player},\paystate}f_{l_{\player}}(\act_{-\player},\paystate)\out_{m_{\player}}(\act_{-\player},\paystate).\label{eq:combined_in}
\end{equation}

For $\act_{\player}\in \supp_\player(\out)$, we define the function $g_{\act_{\player}}:\Act_{-\player}\times\Paystate\rightarrow\mathbb{R}$ by
\[
g_{\act_{\player}}(\act_{-\player},\paystate)=s_{m_{\player}}\cdot f_{m_{\player}}(\act_{-\player},\paystate)
\]
where $m_\player$ is such that $\out_{\act_{\player}}=\out_{m_{\player}}$.
For $\act_{\player}\notin \supp_\player(\out)$, we define $g_{\act_{\player}}=0$. 

We claim that for all $\act_{\player}\in \supp_\player(\out)$ and $\actb_{\player}\notin\{\actc_\player\in \supp_\player(\out): \out_{\actc_{\player}}=\out_{\act_{\player}}\}$,
\begin{equation}\label{eq:claim_strict}
\sum_{\act_{-\player},\paystate}g_{\act_{\player}}(\act_{-\player},\paystate)\out_{\act_{\player}}(\act_{-\player},\paystate)>\sum_{\act_{-\player},\paystate}g_{\actb_{\player}}(\act_{-\player},\paystate)\out_{\act_{\player}}(\act_{-\player},\paystate).
\end{equation}

To verify the claim, pick $m_{\player}\in\{1,\ldots,n_{\player}\}$ such that $\out_{\act_{\player}}=\out_{m_{\player}}$.
From the left-hand side of (\ref{eq:separation_thm}) and the fact that  $s_{m_{i}}>0$, we obtain that
\begin{equation}
\sum_{\act_{-\player},\paystate}g_{\act_{\player}}(\act_{-\player},\paystate)\out_{\act_{\player}}(\act_{-\player},\paystate)=s_{m_{\player}}\sum_{\act_{-\player},\paystate}f_{m_{\player}}(\act_{-\player},\paystate)\out_{m_{\player}}(\act_{-\player},\paystate)>0.\label{eq:strict_posit}
\end{equation}
Hence, for $\actb_{\player}\notin\supp_\player(\out)$, we have
\[
\sum_{\act_{-\player},\paystate}g_{\act_{\player}}(\act_{-\player},\paystate)\out_{\act_{\player}}(\act_{-\player},\paystate)>0=\sum_{\act_{-\player},\paystate}g_{\actb_{\player}}(\act_{-\player},\paystate)\out_{\act_{\player}}(\act_{-\player},\paystate)
\]
where the equality follows from $g_{\actb_{\player}}=0$. 

Assume now that $\actb_{\player}\in\supp_\player(\out)$. Choose $l_{\player}$ such that $\out_{\actb_{\player}}=\out_{l_{\player}}$.
Since $\out_{\act_{\player}}\neq \out_{\actb_{\player}}$, $m_{\player}\neq l_{\player}$. Suppose that $l_{\player}>m_{\player}$. It follows from the right-hand side of
(\ref{eq:separation_thm})---in (\ref{eq:separation_thm}) the roles
of $l_{\player}$ and $m_{\player}$ are inverted---that
\[
0\geq \sum_{\act_{-\player},\paystate}f_{l_{\player}}(\act_{-\player},\paystate)\out_{m_{\player}}(\act_{-\player},\paystate).
\]
Thus, given that $s_{l_{\player}}>0$, we deduce that
\[
0\geq \sum_{\act_{-\player},\paystate}g_{\actb_{\player}}(\act_{-\player},\paystate)\out_{\act_{\player}}(\act_{-\player},\paystate)=s_{l_{\player}}\sum_{\act_{-\player},\paystate}f_{l_{\player}}(\act_{-\player},\paystate)\out_{m_{\player}}(\act_{-\player},\paystate).
\]
We obtain that
\[
\sum_{\act_{-\player},\paystate}g_{\act_{\player}}(\act_{-\player},\paystate)\out_{\act_{\player}}(\act_{-\player},\paystate)>0\geq\sum_{\act_{-\player},\paystate}g_{\actb_{\player}}(\act_{-\player},\paystate)\out_{\act_{\player}}(\act_{-\player},\paystate)
\]
where we use again (\ref{eq:strict_posit}). For the case $l_{\player}<m_{\player}$, the condition
\[
\sum_{\act_{-\player},\paystate}g_{\act_{\player}}(\act_{-\player},\paystate)\out_{\act_{\player}}(\act_{-\player},\paystate)>\sum_{\act_{-\player},\paystate}g_{\actb_{\player}}(\act_{-\player},\paystate)\out_{\act_{\player}}(\act_{-\player},\paystate)
\]
is equivalent to (\ref{eq:combined_in}). We conclude that (\ref{eq:claim_strict}) holds.

To complete the proof of the lemma, for every $\delta>0$ we define $\utilb=(\utilb_\player)_{\player\in\Player}$ by
\[
\utilb_{\player}(\act,\paystate)=\util_{\player}(\act,\paystate)+\delta g_{\act_{\player}}(\act_{-\player},\paystate).
\]
By choosing $\delta$ sufficiently small, we can be make sure that
$\Vert \util-\utilb\Vert\leq\epsilon$. Since $\out\in \BCE(\util)$, it follows
from (\ref{eq:claim_strict}) that for all $\player\in\Player$ and $\act_\player\in\supp_\player(\out)$,
\[
\act_\player \in BR^\prime(\out_{\act_\player})\subseteq\{\actb_\player\in\supp_\player(\out):\out_{\act_\player}=\out_{\actb_\player}\}
\]
where $BR^\prime(\out_{\act_\player})$ is the set of $i$'s best response to a belief $\out_{\act_\player}$ given utility function $\utilb_\player$. Thus, $\out\in \sBCE(\utilb)$.
\end{proof}

A subset of a Euclidean space is \textbf{semi-algebraic} if it is defined by finite systems of polynomial inequalities. A correspondence between Euclidean spaces is semi-algebraic if its graph is semi-algebraic. The background knowledge on semi-algebraic sets that we use in this proof can be gathered from \citet[Section 2]{blume1994algebraic}. 

\begin{lemma}\label{lem:semi_algebraic}
The correspondences $\util\mapsto \BCE(\util)$, $\util\mapsto \sBCE(\util)$, and $\util\mapsto \cl(\sBCE(\util))$ are semi-algebraic.
\end{lemma}
\begin{proof}
The BCE correspondence is semi-algebraic: for all $\util\in\mathbb{R}^{\Player\times\Act\times\Paystate}$ and $\out\in \mathbb{R}^{\Act\times\Paystate}$, $\out\in \BCE(\util)$ if and only if the pair $(\util,\out)$ is a solution to the following finite system of polynomial inequalities:
\begin{align*}
\out(\act,\paystate) & \geq 0 &&\text{for all $a\in A$ and $\paystate\in\Paystate$},\\
\sum_{\act}\out(\act,\paystate) & = \payprior(\paystate) && \text{for all $\paystate\in \Paystate$},\\
\sum_{\act_{-\player},\paystate}(\util (\act,\paystate)-\util (\actb_\player,\act_{-\player},\paystate))\out(\act,\paystate) & \geq 0 &&\text{for all $\player\in\Player$ and $\act_\player,\actb_\player\in\Act_\player$}.
\end{align*}

The sBCE correspondence is also semi-algebraic. To prove it, for every $\util\in \mathbb{R}^{\Player\times\Act\times\Paystate}$, $\out\in \mathbb{R}^{\Act\times\Paystate}$, $\player\in\Player$, and $\act_\player,\actb_\player,\actc_\player\in\Act_\player$, we denote by $F(\util,\out,\act_\player,\actb_\player,\actc_\player)$ the quantity
\[
\sum_{\act_{-\player},\paystate}(\util (\act_\player,\act_{-\player},\paystate)-\util (\actc_\player,\act_{-\player},\paystate))\out(\act_\player,\act_{-\player},\paystate)+(\util (\actb_\player,\act_{-\player},\paystate)-\util (\actc_\player,\act_{-\player},\paystate))\out(\actb_\player,\act_{-\player},\paystate).
\]
We observe that $\out\in \sBCE(\util)$ if and only if $\out\in \BCE(\util)$ and for every $\player\in\Player$ there is $T_\player\subseteq \Act_\player\times \Act_\player$ such that the pair $(\util,\out)$ is a solution of the following finite system of polynomial inequalities:
\begin{align*}
\sum_{\act_{-\player},\paystate}(\out(\act_\player,\act_{-\player},\paystate)\out(\actb_\player)- \out(\actb_\player,\act_{-\player},\paystate)\out(\act_\player))^2 & =0 && \text{for all $(\act_\player,\actb_\player)\in T_\player$},\\
F(\util,\out,\act_\player,\actb_\player,\actc_\player) & > 0 &&\text{for all $(\act_\player,\actb_\player)\notin T_\player$ and $\actc_\player\in\Act_\player$}.
\end{align*}
Thus, $\out\in \sBCE(\util)$ if and only if it the solution of one of finitely many systems of polynomial inequalities; we conclude that the sBCE correspondence is semi-algebraic. 

The correspondence $\util\mapsto \cl(\sBCE(\util))$ is also semi-algebraic. Indeed, $\out\in \cl(\sBCE(\util))$ if and only if for every $\epsilon>0$ there exists $\outb\in\mathbb{R}^{\Act\times\Paystate}$ such that $\Vert\out-\outb\Vert \leq \epsilon$ and $\outb\in\sBCE(\util)$. Thus, since the sBCE correspondence is semi-algebraic, the graph of the correspondence $\util\mapsto \cl(\sBCE(\util))$ is defined by a first-order formula and therefore semi-algebraic by the Tarski-Seidenberg theorem (\citealp{blume1994algebraic}, page 787).
\end{proof}

We are ready to complete the proof of the theorem. By Lemma \ref{lem:semi_algebraic}, the correspondence $\util\mapsto \cl(\sBCE(\util))$ is semi-algebraic. Hence, there is an open subsets $U$ of $\mathbb{R}^{\Player\times\Act\times\Paystate}$ such that the complement of $U$ has Lebesgue measure zero, and $\util\mapsto \cl(\sBCE(\util))$ is continuous on $U$ (\citealp{blume1994algebraic}, page 786). 

We claim that for all $\util\in U$, $\BCE(\util)=\cl(\sBCE(\util))$. To prove the claim, take any $\util\in U$. Since $\BCE(\util)$ is closed and $\sBCE(\util)\subseteq \BCE(\util)$, $\cl(\sBCE(\util))\subseteq \BCE(\util)$. To verify the other inclusion, we use Lemma \ref{lem:poking} to find a sequence of games $(\util^n)_{n=1}^\infty$ such that $\util^n\rightarrow \util$ and, for every $n$, $\out\in\sBCE(\util^n)\subseteq \cl(\sBCE(\util^n))$. Since $ \cl(\sBCE(\util))$ is continuous at $\util$, we have $\out\in \cl(\sBCE(\util))$; see \citet[Theorem 17.16]{Aliprantis2006}. Hence, $\BCE(\util) \subseteq \cl(\sBCE(\util))$. We deduce that $\BCE(\util)=\cl(\sBCE(\util))$, as desired.

We conclude that for generic $\util$, the sBCE set is dense in the BCE set.

\section{Proof of Theorem~\ref{thm:robust difference in welfare}}

Throughout this section, we fix the set of players $\Player$, the set of payoff states $\Paystate$, the prior $\payprior$ (with full support), and an action set $\Act_\player$ for every player $\player$. Given these, specifying the players' utilities is all that is left for defining a base game. As such, we use the profile of utility functions $\util=(\util_\player)_{\player\in\Player}$ as a shorthand for the base game it defines, writing for example $\BCE(\util)$ for the set of BCEs.

The first part of the theorem immediately follows from Theorem \ref{thm:genericity}, together with (\ref{eq:def_rival}) and (\ref{eq:def_exval}). Next, we assume that $\vert\Player\vert\geq 2$, $\vert\Paystate\vert\geq 2$, and $\vert A_\player\vert\geq 2$ for all players $\player$, and show there is an open set $U\subseteq\mathbb{R}^{\Player\times\Act\times\Paystate}$ such that 
$\Exval(\util)\subset\RIval(\util)$ for all $\util\in U$.

To prove $\Exval(\util)\subset\RIval(\util)$, we will determine that 
\begin{equation}\label{eq:stric_ineq_v}
\min_{v\in \RIval(\util)}\sum_\player v_\player <  \min_{v\in \Exval(\util)}\sum_\player v_\player.
\end{equation}
Consistently with the notation of Section \ref{sec:welfare analysis}, we write $\welfen(\out,\util)=\sum_\player \noinfoval_\player(\out,\util)$ and $\welfex(\out,\util)=\sum_\player \grossval_\player(\out,\util)$. We also denote by $\welfen(\util)$ the minimum of $\welfen(\out,\util)$ over all $\out\in\BCE(\util)$, and by $\welfex(\util)$ the minimum of $\welfex(\out,\util)$ over all $\out\in\BCE(\util)$. Note that 
\[
\min_{v\in \RIval(\util)}\sum_\player v_\player  = \min_{\out\in\cl(\sBCE(\util))}\welfen(\out,\util)\geq \welfen(\util),\quad\text{and}\quad
\min_{v\in \Exval(\util)}\sum_\player v_\player  = \welfex(\util).\\
\]

The next lemma gives a sufficient condition for the existence of an open set $U\subseteq\mathbb{R}^{\Player\times\Act\times\Paystate}$ such that (\ref{eq:stric_ineq_v}) holds for all $\util\in U$.

 \begin{lemma}\label{lem:robust worst case difference sufficiency}
    Suppose $\util^*\in\mathbb{R}^{\Player\times\Act\times\Paystate}$ and $\out^* \in \Delta_\payprior(\Act\times\Paystate)$ satisfy the following properties:
    \begin{enumerate}[(i)]
    
    \item Each player $\player$ takes at least two actions at $\out^*$: $|\supp_\player(\out^*)|\geq 2$. 

    \item $\out^*$ is a strict BCE: $BR(\out_{\act_\player})=\{\act_\player\}$ for all $\player\in\Player$ and $\act_\player\in\supp_\player(\out)$. 
    
    \item $\out^*$ is the unique minimizer of $\welfen(\out,\util^*)$ over all $\out\in \BCE(\util^*)$.
      
    \end{enumerate}
    Then there is a neighborhood $U$ of $\util^*$ such that, for all $\util\in U$,
    \[
    \min_{\out\in\cl(\sBCE(\util))}\welfen(\out,\util)<\welfex(\util).
    \]
    
\end{lemma}
\begin{proof}
Take $\util^*$ and $\out^*$ that satisfy (i)-(iii). First, we verify that
\begin{equation}\label{eq:nabi}
    \welfen(\out^*,\util^*)<\welfex(\util^*).
\end{equation}
Take $\out\in \BCE(\util^*)$ such that $\welfex(\util^*)=\welfex(\out,\util^*)$. If $\out\neq\out^*$, then $\welfex(\out,\util^*)\geq \welfen(\out,\util^*)> \welfen(\out^*,\util^*)$, where the strict inequality holds by (iii); thus, $\welfex(\util^*)>\welfen(\out^*,\util^*)$. If instead $\out=\out^*$, then $\welfex(\out^*,\util^*)>\welfen(\out^*,\util^*)$ by (i) and (ii); thus $\welfex(\util^*)>\welfen(\out^*,\util^*)$. Overall, we conclude that (\ref{eq:nabi}) holds, as desired.

The rest of the proof proceed by contradiction. To attain this contradiction, suppose there is a sequence $(\util^n)_{n=1}^\infty$ converging to $\util^*$ such that
\begin{equation}\label{eq:contradiction_ass}
\min_{\out\in\cl(\sBCE(\util^n))}\welfen(\out,\util^n)=\welfex(\util^n)\quad\text{for all }n.
\end{equation}
By (ii), $\out^*\in \sBCE(\util^n)$ for all $n$ sufficiently large. Thus,
\begin{equation}\label{eq:n_comparison}
\Welfen(\out^*,\util^n) \geq \min_{\out\in\cl(\sBCE(\util^n))}\welfen(\out,\util^n)\quad\text{for all $n$ large enough}.
\end{equation}
Combining (\ref{eq:contradiction_ass}) and (\ref{eq:n_comparison}), we obtain that 
\begin{equation}\label{eq:saturday}
\Welfen(\out^*,\util^n) \geq \welfex(\util^n)\quad\text{for all $n$ large enough}.
\end{equation}
 By standard arguments, $\Welfen(\out^*,\util)$ is continuous in $\util$. In addition, since the correspondence $\util\mapsto \BCE(\util)$ is upper hemicontinuous, $\Welfex(\util)=\min_{\out\in \BCE(\util)} \Welfex(\out,\util)$ is lower semicontinuous in $\util$ \citep[e.g.,][Lemma 17.3]{Aliprantis2006}. It follows from (\ref{eq:saturday}) that
\begin{align*}
\Welfen(\out^*,\util^*)=\liminf_{n\rightarrow\infty}\Welfen(\out^*,\util^n)  \geq
 \liminf_{n\rightarrow\infty}\Welfex(\util^n)\geq \Welfex(\util^*).
\end{align*}
Hence, $\Welfen(\out^*,\util^*)\geq \Welfex(\util^*)$, which contradicts (\ref{eq:nabi}).
\end{proof}

To complete the proof of the theorem, we construct a utility profile $\util^*$ and and outcome $\out^*$ that satisfy the conditions of Lemma \ref{lem:robust worst case difference sufficiency}. This lemma then delivers a neighborhood $U$ of $\util^*$ such that (\ref{eq:stric_ineq_v}) holds for all $\util\in U$, which in turn means that the set of $\util$ such that $\Exval(\util)\subset\RIval(\util)$ has non-empty interior.

We now construct $\util^*$ and $\out^*$. Let $n$ be cardinality of $\Player$; by hypothesis, $n\geq 2$. For every player $\player$, we order the set of actions from $0$ to $m_i$ (where $m_i+1$ is the cardinality of $\Act_\player$): $\Act_\player=\{0,\ldots,m_\player\}$. By hypothesis, $\Act_\player$ contains at least two distinct elements, thus $m_i\geq 1$. We also consider a partition $\Paystate=\Paystate_l\cup\Paystate_h$ of the set of payoff states such that both $\Paystate_l$ and $\Paystate_h$ are nonempty; this is feasible because, by hypothesis, $\Paystate$ contains at least two elements.

For player $\player$, we define $\util^*_\player$ as follows:
\[
\util_\player^*(\act,\paystate)=\begin{cases}
0 &\text{if }\act_\player=0,\\
\frac{1}{\pi(\Theta_l)}\left(-1 + \frac{1}{n-1}\sum_{\playerb\neq \player}\act_\playerb\right) &\text{if }\act_\player=1\text{ and }\paystate\in\Paystate_l,\\
\frac{1}{\pi(\Theta_h)}\left(2 - \frac{1}{n-1}\sum_{\playerb\neq \player}\act_\playerb\right) &\text{if }\act_\player=1\text{ and }\paystate\in\Paystate_h,\\
-1 &\text{if }\act_\player>1.
\end{cases}
\]
Thus, action $0$ is a safe action. Action $1$ has a payoff that depends both on the state and on the average action of the opponents. For states in $\Theta_l$, action $1$ generates a \emph{negative} baseline payoff of $-1$, but there is also a \emph{positive} externality from the actions of others; these payoffs are scaled by $1/\payprior(\Paystate_l)$. For states in $\Theta_h$, action $1$ generates a \emph{positive} baseline payoff of $2$, but there is also a \emph{negative} externality from the actions of others; these payoffs are scaled by $1/\payprior(\Paystate_h)$.  Any action outside  $\{0,1\}$ is strictly dominated by $0$. 

Let $\out^*$ be the outcome such that all players take action $0$ when $\paystate\in\Paystate_l$, and all players take action $1$ when $\paystate\in\Paystate_h$. It is clear  that conditions (i) and (ii) of Lemma~\ref{lem:robust worst case difference sufficiency} are satisfied.  

All that remains is to verify (iii). To do so, note first that $\out^*$ is the unique minimizer of $\sum_\player \sum_{\act,\paystate}\util^*_\player(1,\act_{-\player},\paystate)\out(\act,\paystate)$ over all $\out\in \Delta_\payprior(\{0,1\}^\Player\times \Paystate)$. Since any action outside $\{0,1\}$ is strictly dominated, we deduce that $\out*$ is the unique minimizer of $\sum_\player \sum_{\act,\paystate}\util^*_\player(1,\act_{-\player},\paystate)\out(\act,\paystate)$ over all $\out\in\BCE(\util^*)$. Moreover, simple algebra shows that for all players $\player$,
\[
\max_{\actb_\player}\sum_{\act,\paystate}\util^*_\player(\actb_\player,\act_{-\player},\paystate)\out^*(\act,\paystate)= \sum_{\act,\paystate}\util^*_\player(1,\act_{-\player},\paystate)\out^*(\act,\paystate).
\]
In turn, this implies that 
\[
\Welfen(\out^*,\util^*)=\sum_\player\sum_{\act,\paystate}\util^*_\player(1,\act_{-\player},\paystate)\out^*(\act,\paystate).
\]
Therefore, every $\out \in \BCE(\util^*) \setminus \{\out^*\}$ has
\begin{align*}
\welfen(\out,\util^*) 
& = \sum_\player \max_{\actb_\player}\sum_{\act,\paystate}\util^*_\player(\actb_\player,\act_{-\player},\paystate)\out(\act,\paystate)
\\ & \geq \sum_\player \sum_{\act,\paystate}\util^*_\player(1,\act_{-\player},\paystate)\out(\act,\paystate)
\\ & > \sum_\player \sum_{\act,\paystate}\util^*_\player(1,\act_{-\player},\paystate)\out^*(\act,\paystate) 
= \welfen(\out^*,\util^*).
\end{align*}
we conclude that $\Welfen(\util^*)=\Welfen(\out^*,\util^*)$ if and only if $\out=\out^*$, that is, (iii) of Lemma~\ref{lem:robust worst case difference sufficiency} holds. The proof is now complete.

\section{Proofs of the results in Section~\ref{sec:welfare analysis}}

\subsection{Proof of Lemma \ref{lem:symmetric worst case}}\label{sec:proof_lemma_sym}

We prove a slightly stronger result:
\begin{claim}
For every BCE $\out$, there is a symmetric BCE $\outb$ such that 
\[
\welfex(\outb)=\welfex(\out)\quad\text{and}\quad\welfen(\outb)\leq\welfen(\out).
\]
\end{claim}
\begin{proof}
Fix a BCE $\out$. Let $\Phi$ be the set of permutations of $\Player$. For every permutation $\phi\in\Phi$, we define the outcome $\out_\phi$ by 
\[
\out_\phi(\act,\paystate)=\out(\act_\phi,\paystate).
\]
Note that player $\player$ in $\out_\phi$ behaves as player $\playerb=\phi^{-1}(\player)$ in $\out$. One can verify that $\out_\phi$ because $\out$ is a BCE and the game is symmetric.

We define the outcome $\outb$ by 
\[
\outb = \frac{1}{\vert\Phi\vert}\sum_{\phi\in\Phi} \out_\phi,
\]
where $\vert\Phi\vert$ is the cardinality of $\Phi$. As noted above, each $\out_\phi$ is a BCE. Since the BCE set is convex, $\outb$ is a BCE. 

The outcome $\outb$ is symmetric. Indeed, for every permutation $\psi\in \Phi$,
\[
\Phi = \left\{\psi^{-1}\circ\phi:\phi\in\Phi\right\}.
\]
Thus, we deduce that
\begin{align*}
\outb (\act_\psi,\paystate) &=  \frac{1}{\vert\Phi\vert}\sum_{\phi\in\Phi} \out_\phi (\act_\psi,\paystate)=\frac{1}{\vert\Phi\vert}\sum_{\phi\in\Phi} \out_{(\psi^{-1}\circ\phi)}(\act_{\psi},\paystate)\\
&=  \frac{1}{\vert\Phi\vert}\sum_{\phi\in\Phi} \out (\act_\phi,\paystate)=\frac{1}{\vert\Phi\vert}\sum_{\phi\in\Phi} \out_{\phi}(\act,\paystate)=\outb (\act,\paystate).
\end{align*}
We conclude $\outb$ is symmetric.

To conclude the proof, we observe that
\[
\welfex (\outb) = \frac{1}{\vert\Phi\vert}\sum_{\phi\in\Phi} \welfex (\out_\phi)= \frac{1}{\vert\Phi\vert}\sum_{\phi\in\Phi} \welfex (\out)=\welfex (\out),
\]
where the first equality holds because $\welfex (\out_\phi)$ is affine in $\out_\phi$, and the second equality because the game is symmetric. Finally, we note that
\[
\welfen (\outb) \leq \frac{1}{\vert\Phi\vert}\sum_{\phi\in\Phi} \welfen (\out_\phi)= \frac{1}{\vert\Phi\vert}\sum_{\phi\in\Phi} \welfen (\out)=\welfen (\out),
\]
where the first inequality holds because $\welfen (\out_\phi)$ is convex in $\out_\phi$, and the second equality because the game is symmetric.
\end{proof}

\subsection{Proof of Proposition \ref{pro:welfare_binary_symmetric}}\label{sec:proof_prop_sym_welfare}
By Lemma \ref{lem:symmetric worst case}, $\welfex$ is the value of the optimization problem
\begin{equation}\label{eq:sym_opt_ex}
\min_{\out\in\BCEsym}\welfex(\out),
\end{equation}
and $\welfen$ is the value of the following optimization
\begin{equation}\label{eq:sym_opt_en}
\min_{\out\in\BCEsym}\welfen(\out).
\end{equation}
We consider also the following optimization problem:
\begin{equation}\label{eq:sym_opt_out}
\min_{\out\in\Outsym}\welfen(\out),
\end{equation}

We proceed by successive claims.
\begin{claim}
The following conditions are equivalent:
\begin{itemize}
\item[(i)] $\welfen<\welfex$ 
\item[(ii)] $\noinfoval_\player(\out)<\grossval_\player(\out)$ for all players $\player$  and optimal solutions $\out$ of (\ref{eq:sym_opt_en}).
\end{itemize}
\end{claim}
\begin{proof}
We first prove that (i) implies (ii). Suppose $\welfen<\welfex$ and let $\out$ be an optimal solution of (\ref{eq:sym_opt_en}). Then, $\welfen(\out)=\welfen<\welfex\leq\welfex(\out) $. The inequality  $\welfen(\out)<\welfex(\out)$ implies that $\noinfoval_\player(\out)<\grossval_\player(\out)$ for some player $\player$. Since the game is symmetric and $\out$ is symmetric, $\noinfoval_\player(\out)<\grossval_\player(\out)$ for all players $\player$.

We now prove that (ii) implies (i). Suppose $\noinfoval_\player(\out)<\grossval_\player(\out)$ for all players $\player$  and optimal solutions $\out$ of (\ref{eq:sym_opt_en}). Let $\out\in \BCE$ be an optimal solution of (\ref{eq:sym_opt_ex}). If $\out$ is also an optimal solution of (\ref{eq:sym_opt_en}), then $\welfen(\out)=\sum_\player \noinfoval_\player(\out)<\sum_\player \grossval_\player(\out)=\welfex(\out)$ by hypothesis; thus, $\welfen<\welfex$. If instead $\out$ is not an optimal solution of (\ref{eq:sym_opt_en}), then $\welfen<\welfen(\out)\leq \welfex(\out)=\Welfex$; thus, $\welfen<\welfex$.
\end{proof}

\begin{claim}
For every $\out\in \BCE$ and $\player\in\Player$, the following conditions are equivalent:
\begin{itemize}
\item[(i)] $\noinfoval_\player(\out)<\grossval_\player(\out)$.
\item[(ii)] $\act_\player\in\supp_\player(\out)$ and  $\BR(\out_{\act_\player}) = \{\act_\player\}$ for all $\act_\player\in\Act_\player$.
\end{itemize}
\end{claim}
\begin{proof}
Condition (i) holds if and only if player $\player$ is strictly better by following the action recommendation of the mediator rather then best responding ex ante. In other terms, player $\player$ has no action $\act_\player$ such that for all $\actb_\player\in\supp_\player (\out)$, $\act_\player\in \BR(\out_{\actb_\player})$. Given that $\Act_\player$ has two elements, this is equivalent to condition (ii). 
\end{proof}

\begin{claim} The following conditions are equivalent:
\begin{itemize}
\item[(i)] All optimal solutions of (\ref{eq:sym_opt_en}) satisfy (\ref{eq:welfare_binary_symmetric}). 
\item[(ii)] All optimal solutions of (\ref{eq:sym_opt_out}) satisfy (\ref{eq:welfare_binary_symmetric}).
\end{itemize}
In either case, 
\[
\argmin_{\out\in\BCEsym}\welfen(\out)=\argmin_{\out\in\Outsym}\welfen(\out).
\]
\end{claim}
\begin{proof}
First we show that (i) implies (ii). Let $\out$ be an optimal solution of (\ref{eq:sym_opt_en}) and let $\outb$ be an optimal solution of (\ref{eq:sym_opt_out}). For every $t\in [0,1]$, define $\out^t=(1-t)\out+t\outb$. Furthermore, set 
\[
s = \max \{t:\out^t\in \BCEsym\}.
\]
Note that $s$ is well defined: the set $\BCEsym$ is closed and $\out^0=\out\in \BCEsym$. 

We observe that $\out^s$ is an optimal solution of (\ref{eq:sym_opt_en}): since $\welfen(\out^t)$ is convex in $t$,
\[
\welfen(\out^s)\leq (1-s)\welfen(\out)+s\welfen(\outb)\leq \welfen(\out)=\welfen.
\]
Thus, $\out^s$ must satisfy (\ref{eq:welfare_binary_symmetric}). But this implies that $s=1$; otherwise, one could find $\epsilon>0$ sufficiently small so that $\out^{s+\epsilon}\in \BCEsym$, contradicting the definition of $\out^s$. This implies that $\outb=\out^s$ satisfies (\ref{eq:welfare_binary_symmetric}).

Now we show that (ii) implies (i). Let $\out$ be an optimal solution of (\ref{eq:sym_opt_en}) and let $\outb$ be an optimal solution of (\ref{eq:sym_opt_out}). Since $\outb$ satisfies (\ref{eq:welfare_binary_symmetric}), $\outb$ is a BCE. Thus, $\outb$ is an optimal solution of (\ref{eq:sym_opt_en}). This implies that $\out$ is an optimal solution of (\ref{eq:sym_opt_out}), and therefore satisfies (\ref{eq:welfare_binary_symmetric}).
\end{proof}

By combining the three claims above, we obtain Proposition \ref{pro:welfare_binary_symmetric}.

\subsection{Proof of Claim \ref{claim:br unconstrained min}}

We begin with a result that establishes a necessary condition for an outcome to solve the relaxed program from Proposition~\ref{pro:welfare_binary_symmetric}. To state the result, let $\actval_\player(\act_\player,\out)$ be player $\player$'s payoff if she always takes action $\act_\player$ while $(\act_{-\player},\paystate)$ is distributed according to $\out$:
\[
\actval_\player(\act_\player,\out):=\sum_{\actb_\player,\act_{-\player},\paystate}\util_\player(\act_\player,\act_{-\player},\paystate)\out(\actb_\player,\act_{-\player},\paystate).
\]
Note that for all $\out\in \Outsym$ and $\player\in\Player$,
\[
\welfen(\out)=n\max\{\actval_\player(0,\out),\actval_\player(1,\out)\}.
\]
Thus,
\[
\argmin_{\out\in \Outsym} \welfen(\out) = \argmin_{\out\in \Outsym}\max\{\actval_\player(0,\out),\actval_\player(1,\out)\}.
\]
\begin{claim}\label{claim:br equality in unconstrained min}
    Every $\out^* \in \argmin_{\out \in \Outsym} \welfen(\out)$ has $\actval_i(0,\out^*) = \actval_i(1,\out^*)$ for all players $\player$.
\end{claim}
\begin{proof}
    We prove the contrapositive: if $\out^*\in \Outsym$ has $\actval_i(0,\out^*) \neq \actval_i(1,\out^*)$, then 
    \[
    \out^* \notin \argmin_{\out \in \Outsym} \welfen(\out).
    \]
    
    We first consider the case in which $\actval_\player(0,\out^*)>\actval_\player(1,\out^*)$. Let $\outb$ be the outcome where all investors attack with probability one. Observe that $
    \outb \in \argmin_{\out\in \Outsym} \actval_\player(0,\out),
    $
    and that every $\outc \in \argmin_{\out\in \Outsym} \actval_\player(0,\out)$ has the speculative attack succeeding with probability one. Hence, every such $\outc$ has $\actval_\player(0,\outc)<\actval_\player(1,\outc)$, which implies that $\out^*\notin \argmin_{\out\in \Outsym} \actval_\player(0,\out)$. We deduce that 
    $
    \actval_\player(0,\outb) < \actval_\player(0,\out^*).
    $
    
    For every $\epsilon \in (0,1)$, we define $\outb^\epsilon:=\epsilon\outb + (1-\epsilon)\out^* \in \Outsym$. Using the inequality $\actval_\player(0,\outb) < \actval_\player(0,\out^*)$, we obtain that for all $\epsilon>0$ small enough,
    \[
    \welfen(\outb^\epsilon) 
    = \Playernum \actval_\player (0,\out^\epsilon) 
    = \Playernum \left(\epsilon \actval_\player (0,\outb) +(1-\epsilon) \actval_\player (0,\out^*)\right)
    < \Playernum\actval_\player (0,\out^*)=\welfen(\out^*).
    \]
    We conclude that $\out^* \notin \argmin_{\out \in \Outsym} \welfen(\out)$.
    
     The argument for the case $\actval_\player(0,\out^*) < \actval_\player(1,\out^*)$ is similar, but with $\outb$ being replaced by the outcome where no one ever speculates.
\end{proof}

Thanks to Claim \ref{claim:br equality in unconstrained min}, to determine $\argmin_{\out \in \Outsym} \welfen(\out)$, we can study the following ``simpler'' optimization problem:
\begin{equation}\label{eq:aux_opt_prob}
\min_{\out \in \Outsym} \actval_\player(0,\out)\quad\text{s.t.}\quad\actval_i(0,\out) = \actval_i(1,\out).
\end{equation}

\begin{claim}
An outcome $\out\in \Outsym$ is an optimal solution of (\ref{eq:aux_opt_prob}) if and only if
\[
    \out\left(\sum_{\playerb\neq \player} \act_\playerb = \paystate -1\right)= 0\quad\text{and}\quad
    \out\left(\sum_{\playerb\neq \player} \act_\playerb\geq \paystate \right)= \frac{\brcost}{1+\brext}.
\]
\end{claim}

\begin{proof}
    Writing the constraint $\actval_\player (1,\out) = \actval_\player (0,\out)$ in terms of $\out$ gives
    \[
    \out\left(\sum_{\playerb\neq \player} \act_\playerb \geq \paystate -1\right) - \brcost = -\brext\out\left(\sum_{\playerb\neq \player} \act_\playerb\geq \paystate\right),
    \]
    which is equivalent to
    \[
     \out\left(\sum_{\playerb\neq \player} \act_\playerb\geq \paystate \right)=\frac{\brcost-\out\left(\sum_{\playerb\neq \player} \act_\playerb = \paystate -1\right)}{1+\brext}.
    \]
    It follows that every $\out$ that satisfies $\actval_\player (1,\out) = \actval_\player (0,\out)$ must yield 
    \[
    \actval_\player (0,\out) = -\brext\out\left(\sum_{\playerb\neq \player} \act_\playerb\geq \paystate \right) = -\brext\left[\frac{\brcost-\out\left(\sum_{\playerb\neq \player} \act_\playerb = \paystate -1\right)}{1+\brext}\right].
    \]
    Hence, we get that (\ref{eq:aux_opt_prob}) is the same as
    \begin{align*}
     \min_{\out \in \Outsym}  &\quad\out\left(\sum_{\playerb\neq \player} \act_\playerb = \paystate -1\right)
    \\ \text{s.t.}\quad\,\,  & \quad \out\left(\sum_{\playerb\neq \player} \act_\playerb\geq \paystate \right) = \frac{\brcost - \out\left(\sum_{\playerb\neq \player} \act_\playerb = \paystate -1\right)}{1+\brext}.
    \end{align*}
    Hence, to complete the proof, we only need to be sure that there is $\out \in \Outsym$ such that \[
    \out\left(\sum_{\playerb\neq \player} \act_\playerb = \paystate -1\right)= 0\quad\text{and}\quad
    \out\left(\sum_{\playerb\neq \player} \act_\playerb\geq \paystate \right)= \frac{\brcost}{1+\brext}.
\]
    Such an outcome is easy to construct: with probability $k/1+x$, all players attack; with the remaining probability, no player attacks.
\end{proof}

The following result connects what we have just found with the conditions in the statement of Claim \ref{claim:br unconstrained min}.

\begin{claim}\label{claim:equiv_cond_boc}
For an outcome $\out \in \Out$, the following conditions are equivalent:
\begin{itemize}
\item[(i)] For all players $\player$ and payoff states $\paystate$,
\begin{align}
    \out\left(\sum_{\playerb\neq \player} \act_\playerb = \paystate -1\right)& = 0,\label{eq:bocconi_mon}\\
    \out\left(\sum_{\playerb\neq \player} \act_\playerb\geq \paystate \right) &= \frac{\brcost}{1+\brext}\label{eq:bocconi_tue}.
\end{align}
\item[(ii)] For all payoff states $\paystate$,
\begin{align}
    \out\left(\paystate -1\leq \sum_{\player} \act_\player \leq\paystate \right)& = 0,\label{eq:bocconi_noom}\\
    \out\left(\sum_{\player} \act_\player>\paystate \right) &= \frac{\brcost}{1+\brext}\label{eq:bocconi_non}.
\end{align}
\end{itemize}
\end{claim}

\begin{proof}
First we show that (i) implies (ii). Since $\max\Theta<n$, 
\[
\out\left(\sum_{\player} \act_\player=\paystate-1\right)=\out\left(\sum_{\player} \act_\player=\paystate-1,\text{ and }\act_\player=0\text{ for some }\player\right).
\]
Thus,  
\[
\out\left(\sum_{\player} \act_\player=\paystate-1\right)\leq \sum_\player\out\left(\sum_{\playerb\neq\player} \act_\playerb=\paystate-1,\text{ and }\act_\player=0\right)=0,
\]
where the last equality follows from (\ref{eq:bocconi_mon}). Moreover, since $\min\Theta>0$,
\[
\out\left(\sum_{\player} \act_\player=\paystate\right)=\out\left(\sum_{\player} \act_\player=\paystate,\text{ and }\act_\player=1\text{ for some }\player\right).
\]
Thus,  
\[
\out\left(\sum_{\player} \act_\player=\paystate\right)\leq \sum_\player\out\left(\sum_{\playerb\neq\player} \act_\playerb=\paystate-1,\text{ and }\act_\player=1\right)=0,
\]
where the last equality follows from (\ref{eq:bocconi_mon}). We conclude that (\ref{eq:bocconi_noom}) holds.

To prove (\ref{eq:bocconi_non}), notice that 
\[
\out\left(\sum_{\player} \act_\player>\paystate\right)=\out\left(\sum_{\player} \act_\player\geq\paystate\right),
\]
because we have just verified that $\out\left(\sum_{\player} \act_\player=\paystate\right)=0$. Then, fixing some player $\player^*$,
\begin{align*}
\out\left(\sum_{\player} \act_\player\geq\paystate\right)& =\out\left(\sum_{\player\neq\player^*} \act_\player\geq\paystate,\text{ and }\act_{\player^*}=0\right)+\out\left(\sum_{\player\neq\player^*} \act_\player\geq\paystate-1,\text{ and }\act_{\player*}=1\right)\\
& =\out\left(\sum_{\player\neq\player^*} \act_\player\geq\paystate,\text{ and }\act_{\player^*}=0\right)+\out\left(\sum_{\player\neq\player^*} \act_\player\geq\paystate,\text{ and }\act_{\player^*}=1\right)\\
& =\out\left(\sum_{\player\neq\player^*} \act_\player\geq\paystate\right)=\frac{\brcost}{1+\brext},
\end{align*}
where the second equality holds by (\ref{eq:bocconi_mon}), and the last equality by (\ref{eq:bocconi_tue}). We deduce (\ref{eq:bocconi_non}). This completes the proof that (i) implies (ii).

Now we show that (ii) implies (i). Observe that
\[
\out\left(\sum_{\playerb\neq \player} \act_\playerb = \paystate -1\right)=\out\left(\sum_{\playerb} \act_\playerb = \paystate -1,\text{ and }\act_\player=0\right)+\out\left(\sum_{\playerb} \act_\playerb = \paystate,\text{ and }\act_\player=1\right).
\]
By (\ref{eq:bocconi_noom}), the right-hand side is equal to zero: we deduce (\ref{eq:bocconi_mon}). We obtain (\ref{eq:bocconi_tue}) from the following chain of equalitites:
\begin{align*}
\out\left(\sum_{\playerb\neq \player} \act_\playerb \geq \paystate \right) & = \out\left(\sum_{\playerb} \act_\playerb \geq \paystate,\text{ and }\act_\player=0 \right)+\out\left(\sum_{\playerb} \act_\playerb > \paystate,\text{ and }\act_\player=1 \right)\\
& = \out\left(\sum_{\playerb} \act_\playerb > \paystate,\text{ and }\act_\player=0 \right)+\out\left(\sum_{\playerb} \act_\playerb > \paystate,\text{ and }\act_\player=1 \right)\\
& = \out\left(\sum_{\playerb} \act_\playerb > \paystate \right)=\frac{\brcost}{1+\brext},
\end{align*}
where the second equality follows from (\ref{eq:bocconi_noom}), and the last equality from (\ref{eq:bocconi_non}). This completes the proof that (ii) implies (i). 
\end{proof}

Combining the three results above, we obtain Claim \ref{claim:br unconstrained min}.

\subsection{Proof of Claim \ref{claim:br two states char}}

First, we obtain necessary and sufficient conditions for $\Welfen < \Welfex$ in the regime change game for an arbitrary number of states. 

\begin{claim}\label{claim:to_refine}
The inequality $\Welfen < \Welfex$ holds if and only if all symmetric outcomes $\out$ that satisfy (\ref{eq:br not pivotal}) and (\ref{eq:br ex-ante fail prob}), also satisfy
\begin{equation}\label{eq:bocconi_sem}
\out_{\act_\player=1}\left(\sum_{\playerb}\act_\playerb\geq \theta \right)>\frac{k}{1+x},
\end{equation}
where $\out_{\act_\player=1}$ is the conditional probability of $(\act_{-\player},\paystate)$ given $\act_\player=1$.
\end{claim}
\begin{proof}
By Proposition \ref{pro:welfare_binary_symmetric}, $\Welfen < \Welfex$ if and only if all optimal solutions of $\min_{\out\in\Outsym}\welfen(\out)$ satisfy (\ref{eq:welfare_binary_symmetric}). By Claim \ref{claim:br unconstrained min}, the latter condition is equivalent to the following statement: all symmetric outcomes $\out$ that satisfy (\ref{eq:br not pivotal}) and (\ref{eq:br ex-ante fail prob}), also satisfy (\ref{eq:welfare_binary_symmetric}). Next we verify that, for all symmetric outcomes $\out$ that satisfy (\ref{eq:br not pivotal}) and (\ref{eq:br ex-ante fail prob}), the conditions (\ref{eq:welfare_binary_symmetric}) and (\ref{eq:bocconi_sem}) are equivalent.

Let $\out$ be a symmetric outcome that satisfy (\ref{eq:br not pivotal}) and (\ref{eq:br ex-ante fail prob}). First, note \eqref{eq:br ex-ante fail prob} implies the attack succeeds with a probability strictly between $0$ and $1$, and so players must both attack and and not attack with positive probability due to symmetry. Hence $\supp_\player(\out)=\{0,1\}=\Act_\player$. 

Given $\supp_\player(\out)=\{0,1\}$, $\player$'s obedience constraints are strict when
\begin{align}
    \out_{\act_\player=1}\left(\sum_{\playerb\neq \player} \act_\playerb \geq \paystate-1\right)-\brcost & > -\brext \out_{\act_\player=1}\left(\sum_{\playerb\neq \player} \act_\playerb \geq \paystate\right),\label{eq:bocconi-after}\\
    \out_{\act_\player=0}\left(\sum_{\playerb\neq \player} \act_\playerb \geq \paystate-1\right)-\brcost & < -\brext \out_{\act_\player=0}\left(\sum_{\playerb\neq \player} \act_\playerb \geq \paystate\right)\label{eq:bocconi-before}.
\end{align}
By (\ref{eq:br not pivotal})---see also Claim \ref{claim:equiv_cond_boc}---
\[
    \out_{\act_\player=1}\left(\sum_{\playerb\neq \player} \act_\playerb \geq \paystate-1\right) = \out_{\act_\player=1}\left(\sum_{\playerb\neq \player} \act_\playerb \geq \paystate\right)\quad\text{and}\quad
    \out_{\act_\player=0}\left(\sum_{\playerb\neq \player} \act_\playerb \geq \paystate-1\right) = \out_{\act_\player=0}\left(\sum_{\playerb\neq \player} \act_\playerb \geq \paystate\right).
\]
Thus, (\ref{eq:bocconi-after}) and (\ref{eq:bocconi-before}) hold if and only if 
\[
\out_{\act_\player=1}\left(\sum_{\playerb\neq \player} \act_\playerb \geq \paystate\right)>\frac{\brcost}{1+\brext}>
\out_{\act_\player=0}\left(\sum_{\playerb\neq \player} \act_\playerb \geq \paystate\right).
\]
By (\ref{eq:br not pivotal}) and (\ref{eq:br ex-ante fail prob})---see also Claim \ref{claim:equiv_cond_boc}---
\[
\out\left(\sum_{\playerb\neq \player} \act_\playerb \geq \paystate\right)=\frac{\brcost}{1+\brext}.
\]
Thus, by the law of total probability, (\ref{eq:bocconi-after}) and (\ref{eq:bocconi-before}) hold if and only if
\[
\out_{\act_\player=1}\left(\sum_{\playerb} \act_\playerb \geq \paystate\right)=\out_{\act_\player=1}\left(\sum_{\playerb\neq \player} \act_\playerb \geq \paystate-1\right)=\out_{\act_\player=1}\left(\sum_{\playerb\neq \player} \act_\playerb \geq \paystate\right)>\frac{\brcost}{1+\brext}.
\]
Overall, we conclude that, for all symmetric outcomes $\out$ that satisfy (\ref{eq:br not pivotal}) and (\ref{eq:br ex-ante fail prob}), the conditions (\ref{eq:welfare_binary_symmetric}) and (\ref{eq:bocconi_sem}) are equivalent.
\end{proof}

Next we refine the characterization $\Welfen < \Welfex$ obtained in Claim \ref{claim:to_refine}. To state this refinement, denote the CDF of $\paystate$ by $\ppcdf(\paystate) := \sum_{\paystateb \leq \paystate}\payprior(\paystate)$. Define also the cutoff state $\paystate^*$ by
\[
\paystate^* := \min\left\{\paystate \in \Paystate: \ppcdf(\paystate) \geq \frac{\brcost}{1+\brext}\right\}
\]

\begin{claim}\label{claim:br general welfare difference inequality}
    The inequality $\Welfex>\Welfen$ holds if and only if 
    \begin{equation}\label{eq:2023-05-16a}
    \ppcdf(\paystate^*)\left(\paystate^* - \bbE[\paystate|\paystate\leq \paystate^*]\right) < \frac{\brcost}{1+\brext}\left( 3-\frac{3\brcost}{1+\brext} + \paystate^* - \bbE[\paystate] \right).
    \end{equation}
\end{claim}
\begin{proof}
    By Claim~\ref{claim:to_refine},  $\Welfen<\Welfex$ is equivalent to 
    \begin{equation}\label{eq:2023-05-16b}
        \begin{split}
         \frac{\brcost}{1+\brext} < \min_{\out \in \Outsym} & \out_{\act_\player=1}\left(\sum_{\playerb\neq\player} \act_\playerb \geq \paystate-1\right) \\
        \text{s.t.}\quad  & \eqref{eq:br not pivotal} \text{ and }\eqref{eq:br ex-ante fail prob}. 
        \end{split}
    \end{equation}
    Hence, showing \eqref{eq:2023-05-16a} and \eqref{eq:2023-05-16b} are equivalent is sufficient. To show this equivalence, we first characterize the unique solution to the program on the right hand side of \eqref{eq:2023-05-16b}. This solution gives the value of the program, which we then compare to $\brcost/(1+\brext)$.
    
    We begin with an alternative way of representing symmetric outcomes. This representation is based on the observation that an outcome $\out \in \Out$ is symmetric if and only if, conditional on the state, all action profiles with the same number of attackers have the same probability. Consequently, $\out \in \Outsym$ if and only if there is $\brsout:\Paystate \rightarrow \Delta\left(\{0,\ldots,\Playernum\}\right)$ such that
    \[
    \out(\act,\paystate) = \binom{\Playernum}{\sum_{\playerb}\act_\playerb} \brsout\left(\sum_{\playerb}\act_\playerb\Big|\paystate\right)\payprior(\paystate),
    \]
    where $\binom{\Playernum}{\sum_{\playerb}\act_\playerb}$ is the binomial coefficient. Thus, one can write
    \[
    \out(\act_\player=1) = \sum_\paystate \payprior(\paystate)\sum_{m=1}^{\Playernum}\frac{m}{\Playernum}\brsout (m|\paystate).
    \]
    Moreover, condition \eqref{eq:br not pivotal} is equivalent to $\brsout(\paystate-1|\paystate) = \brsout(\paystate|\paystate) = 0$. Therefore, 
    \[
    \out\left(\sum_{\playerb\neq\player}\act_\playerb \geq \paystate-1 \text{ and } \act_\player=1\right) = \sum_\paystate \payprior(\paystate)\sum_{m\geq \paystate+1}\frac{m}{\Playernum}\brsout (m|\paystate),
    \]
    and
    \[
    \out\left(\sum_{\playerb}\act_\playerb \geq \paystate\right) = \sum_\paystate \payprior(\paystate)\sum_{m\geq \paystate+1}\brsout (m|\paystate).
    \]
    Hence, letting
    \[
    \func(\brsout) = \frac{\sum_\paystate \payprior(\paystate)\sum_{m\geq \paystate+1}m\brsout (m|\paystate)}{\sum_\paystate \payprior(\paystate)\sum_{m=1}^{\Playernum}m\brsout (m|\paystate)},
    \]
    we can write the program on the right hand side of \eqref{eq:2023-05-16a} as
    \begin{equation}\label{eq:2023-05-17a}
    \begin{split}
    \min_{\brsout: \Paystate \rightarrow \Delta(\{0,\ldots,\Playernum\})} & 
    \func(\brsout)
    \\ \text{s.t.} \quad
    & \sum_\paystate \payprior(\paystate)\sum_{m\geq \paystate+1}\brsout (m|\paystate) = \frac{\brcost}{1+\brext}, 
    \\
    & \brsout(\paystate-1|\paystate) = \brsout(\paystate|\paystate) = 0 \text{ for all }\paystate.
    \end{split}
    \end{equation}
    Since the constraint set is compact and the objective continuous, the above program admits a solution, $\brsout^*$. We now use perturbation-based arguments to show $\brsout^*$ must satisfy a few properties:
    \begin{enumerate}
        \item $\brsout^*(m|\paystate)=0$ whenever $m \notin \{\paystate-2,\paystate+1\}$: if $\brsout^*(m|\paystate)>0$ for $m > \paystate+1$ (resp., $m < \paystate-2$), one can reduce the objective without violating the constraints by moving $\epsilon>0$ mass from $\brsout^*(m|\paystate)$ to $\brsout^*(\paystate+1|\paystate)$ (resp., $\brsout^*(\paystate-2|\paystate)$).
        
        \item If $\brsout^*(\paystate+1|\paystate)>0$, then $\brsout^*(\paystateb+1|\paystateb)=1$ for all $\paystateb < \paystate$: For a contradiction, suppose $\brsout^*(\paystate+1|\paystate)>0$, but $\brsout^*(\paystateb+1|\paystateb)<1$ for some $\paystateb < \paystate$. For every $\epsilon>0$, define the following perturbation $\brsout^{\epsilon}$ of $\brsout$: 
        \[
        \brsout^{\epsilon}(m|\paystatec) = 
        \begin{cases}
            \brsout^*(\paystate+1|\paystate) - \epsilon & \text{if }m=\paystate+1, \paystatec = \paystate, \\ 

            \brsout^*(\paystate-2|\paystate) + \epsilon & \text{if }m=\paystate-2, \paystatec = \paystate, \\

            \brsout^*(\paystateb+1|\paystateb) + \epsilon\frac{\payprior(\paystate)}{\payprior(\paystateb)} & \text{if }m=\paystateb+1, \paystatec = \paystateb, \\ 

            \brsout^*(\paystateb-2|\paystateb) - \epsilon\frac{\payprior(\paystate)}{\payprior(\paystateb)} & \text{if }m=\paystateb-2, \paystatec = \paystateb, \\

            \brsout^*(m|\paystatec) & \text{otherwise.}
        \end{cases}
        \]
        The contradiction assumption means $\brsout^\epsilon$ is feasible for all sufficiently small $\epsilon>0$. Direct computation shows
        \[
        \lim_{\epsilon \searrow 0}\frac{1}{\epsilon}\left(\func(\brsout^\epsilon) - \func(\brsout^*)\right) = \frac{\payprior(\paystate)(\paystateb - \paystate)}{\sum_{\paystatec} \payprior(\paystatec)\sum_{m=1}^{\Playernum}m\brsout^* (m|\paystatec)} <0 ,
        \]
        contradicting the optimality of $\brsout$.
        
        \item If $\brsout^*(\paystate-2|\paystate)>0$, then $\brsout^*(\paystateb-2|\paystateb)=1$ for all $\paystateb > \paystate$:  For a contradiction, suppose $\brsout^*(\paystate-2|\paystate)>0$, but $\brsout^*(\paystateb-2|\paystateb)<1$ for some $\paystateb > \paystate$. For every $\epsilon>0$, define the following perturbation $\brsout^{\epsilon}$ of $\brsout$: 
        \[
        \brsout^{\epsilon}(m|\paystatec) = 
        \begin{cases}
            \brsout^*(\paystate-2|\paystate) - \epsilon & \text{if }m=\paystate-2, \paystatec = \paystate, \\ 

            \brsout^*(\paystate+1|\paystate) + \epsilon & \text{if }m=\paystate+1, \paystatec = \paystate, \\

            \brsout^*(\paystateb-2|\paystateb) + \epsilon\frac{\payprior(\paystate)}{\payprior(\paystateb)} & \text{if }m=\paystateb-2, \paystatec = \paystateb, \\ 

            \brsout^*(\paystateb+1|\paystateb) - \epsilon\frac{\payprior(\paystate)}{\payprior(\paystateb)} & \text{if }m=\paystateb+1, \paystatec = \paystateb, \\

            \brsout^*(m|\paystatec) & \text{otherwise.}
        \end{cases}
        \]
        The contradiction assumption means $\brsout^\epsilon$ is feasible for all sufficiently small $\epsilon>0$. Direct computation shows
        \[
        \lim_{\epsilon \searrow 0}\frac{1}{\epsilon}\left(\func(\brsout^\epsilon) - \func(\brsout^*)\right) = \frac{\payprior(\paystate)(\paystate-\paystateb)}{\sum_{\paystatec} \payprior(\paystatec)\sum_{m=1}^{\Playernum}m\brsout (m|\paystatec)} <0,
        \]
        contradicting the optimality of $\brsout^*$.        
    \end{enumerate}
    The above conditions imply the optimal $\brsout^*$ admits a cutoff $\tilde\paystate$ such that $\brsout^*(\paystate+1|\paystate)=1$ for all $\paystate<\tilde\paystate$, $\brsout^*(\paystate-2|\paystate)=1$ for all $\paystate>\tilde\paystate$, and $\brsout^*(\{\tilde\paystate+1,\tilde\paystate-2\}|\tilde\paystate) =1$. Then, the constraint 
    \[
    \sum_\paystate \payprior(\paystate)\sum_{m\geq \paystate+1}\brsout (m|\paystate) = \frac{\brcost}{1+\brext}
    \] pins down the optimum: we must have $\tilde\paystate=\paystate^*$, and  
    \[
    \brsout^*(\paystate^*+1|\paystate^*) = \frac{1}{\payprior(\paystate^*)}\left(\frac{\brcost}{1+\brext} - \ppcdf(\paystate^*-1)\right).
    \]
    Therefore, the inequality \eqref{eq:2023-05-16b} becomes
    \begin{align*}
    \frac{\brcost}{1+\brext} < \func(\brsout^*) 
    & = \frac{\sum_\paystate \payprior(\paystate)\sum_{m\geq \paystate+1}m\brsout^* (m|\paystate)}{\sum_\paystate \payprior(\paystate)\sum_{m=1}^{\Playernum}m\brsout^* (m|\paystate)} \\
    & = \frac{\ppcdf(\paystate^*)\bbE[\paystate+1|\paystate\leq \paystate^*] - \left(\ppcdf(\paystate^*) - \frac{\brcost}{1+\brext}\right)(\paystate^*+1) 
    }
    {
    \bbE[\paystate] + \frac{\brcost}{1+\brext} -2\left(1-\frac{\brcost}{1+\brext}\right)
    }.
    \end{align*}
    Rearranging the above equation gives \eqref{eq:2023-05-16a}.
\end{proof}

Finally, we prove Claim~\ref{claim:br two states char} by specializing Claim~\ref{claim:br general welfare difference inequality} to two states.

\begin{proof}[Proof of Claim~\ref{claim:br two states char}]
We begin the proof by explicitly stating the implication of \eqref{eq:2023-05-16a} for the binary  case. In particular, we show $\Welfen < \Welfex$ if and only if one of the following two conditions hold:
\begin{enumerate}[(i)]
    \item $\frac{\brcost}{1+\brext} > \payprior(\minpaystate)$ and $\frac{\brcost}{1+\brext} > \frac{1}{3}\payprior(\minpaystate)(\maxpaystate-\minpaystate)$.
    \item $\frac{\brcost}{1+\brext} \leq \payprior(\minpaystate)$ and $\frac{\brcost}{1+\brext} < 1 - \frac{1}{3}\payprior(\maxpaystate)(\maxpaystate-\minpaystate).$
\end{enumerate}
To prove the above, we consider two cases, depending on the value of $\paystate^*$:
\begin{itemize}
    \item \emph{Case 1:} $\brcost/(1+\brext) > \payprior(\minpaystate)$. Then $\paystate^*=\maxpaystate$, and the inequality \eqref{eq:2023-05-16a} specializes to
\[
\maxpaystate - \bbE[\paystate] < \frac{\brcost}{1+\brext}\left(3 - \frac{3\brcost}{1+\brext} + \maxpaystate - \bbE[\paystate]\right).
\]
Substituting $\maxpaystate - \bbE[\paystate] = \payprior(\minpaystate)(\maxpaystate-\minpaystate)$ and rearranging gives
\[
\left(1-\frac{\brcost}{1+\brext}\right)\payprior(\minpaystate)(\maxpaystate-\minpaystate) <
3\frac{\brcost}{1+\brext}\left(1-\frac{\brcost}{1+\brext}\right),
\]
which is equivalent to
\[
\frac{1}{3}\payprior(\minpaystate)(\maxpaystate-\minpaystate) < \frac{\brcost}{1+\brext}.
\]
Thus, we have established (i) is sufficient for $\Welfen < \Welfex$, and necessary if $\payprior(\minpaystate)\geq \frac{\brcost}{1+\brext}$. 

\item \emph{Case 2:} Suppose now $\brcost/(1+\brext) \leq \payprior(\minpaystate)$. Then $\paystate^* = \minpaystate$. Thus, the inequality \eqref{eq:2023-05-16a} is now
\[
0 < \frac{\brcost}{1+\brext}\left(3 - \frac{3\brcost}{1+\brext} + \minpaystate - \bbE[\paystate]\right).
\]
Note $\minpaystate - \bbE[\paystate] = -\payprior(\maxpaystate)(\maxpaystate-\minpaystate)$. Therefore, the above inequality is equivalent to
\[
\frac{\brcost}{1+\brext} < 1 - \frac{1}{3}\payprior(\maxpaystate)(\maxpaystate-\minpaystate).
\]
Hence, (ii) is sufficient for $\Welfen < \Welfex$, and necessary if $\payprior(\minpaystate)\leq\frac{\brcost}{1+\brext}$.
\end{itemize}

Next, we argue that a violation of one of the claim's conditions implies that either (i) or (ii) above hold. Suppose first $\maxpaystate-\minpaystate<3$. In this case, $\frac{1}{3}\payprior(\minpaystate)(\maxpaystate-\minpaystate) < \payprior(\minpaystate)$, and so (i) holds whenever $\frac{\brcost}{1+\brext} > \payprior(\minpaystate)$. If $\frac{\brcost}{1+\brext} \leq \payprior(\minpaystate)$, then (ii) holds, because 
\[
1 - \frac{1}{3}\payprior(\maxpaystate)(\maxpaystate-\minpaystate) > 1-\payprior(\maxpaystate) = \payprior(\minpaystate) \geq \frac{\brcost}{1+\brext}.
\]
Suppose now $\maxpaystate-\minpaystate\geq 3$, but \eqref{eq:br binary welfare diff condition} fails. Then one of the following inequality chains must hold: either  
\[
\frac{\brcost}{1+\brext} > \frac{1}{3}(\maxpaystate-\minpaystate)(1-\payprior(\maxpaystate)) =\frac{1}{3}(\maxpaystate-\minpaystate)\payprior(\minpaystate) \geq  \payprior(\minpaystate),
\]
or
\[
\frac{\brcost}{1+\brext} < 1-\frac{1}{3}(\maxpaystate-\minpaystate)\payprior(\maxpaystate) \leq 1- \payprior(\maxpaystate) = \payprior(\minpaystate).
\]
Either way, $\Welfen < \Welfex$ holds: the first inequality chain implies (i), whereas the second inequality chain implies (ii).

To conclude the proof, we show that the claim's condition must hold if neither (i) nor (ii) hold. Suppose first that $\frac{\brcost}{1+\brext}>\payprior(\minpaystate)$, but (i) fails. Then 
\[
\frac{1}{3}(\maxpaystate-\minpaystate)(1-\payprior(\maxpaystate)) = \frac{1}{3}(\maxpaystate-\minpaystate)\payprior(\minpaystate) \geq \frac{\brcost}{1+\brext}>\payprior(\minpaystate),
\]
meaning $\maxpaystate-\minpaystate\geq 3$, and the right inequality in \eqref{eq:br binary welfare diff condition} holds. For the left inequality, note that
\[
1-\frac{1}{3}(\maxpaystate-\minpaystate)\payprior(\maxpaystate) \leq 1-\payprior(\maxpaystate) = \payprior(\minpaystate) < \frac{\brcost}{1+\brext}.
\]
Suppose now $\frac{\brcost}{1+\brext}\leq \payprior(\minpaystate)$, but (ii) fails. Then,
\[
1-\payprior(\maxpaystate) \geq \frac{\brcost}{1+\brext} \geq 1 - \frac{1}{3}(\maxpaystate-\minpaystate)\payprior(\maxpaystate).
\]
The right inequality above delivers the left inequality in \eqref{eq:br binary welfare diff condition}. Moreover, the implied inequality between the left most expression and the right most expression implies 
\[
\frac{1}{3}(\maxpaystate-\minpaystate)\payprior(\maxpaystate) \geq \payprior(\maxpaystate),
\]
and so $\maxpaystate-\minpaystate \geq 3$. Finally, to get the right inequality in \eqref{eq:br binary welfare diff condition}, notice that
\[
\frac{\brcost}{1+\brext} \leq 1-\payprior(\maxpaystate) \leq \frac{1}{3}(\maxpaystate-\minpaystate)(1-\payprior(\maxpaystate)),
\]
where the last inequality holds because $\maxpaystate-\minpaystate \geq 3$.
\end{proof}

\bibliography{references}

\newpage

\setcounter{page}{1} 

\begin{center}
\huge{Online Appendix}
\end{center}

\section{An Example Where Separation Binds}\label{sec:An Example Where the Separation constraint bind}

Next we present a base game where (i) the set of separated BCEs is \emph{nowhere dense} in the set of BCEs, (ii) for every player $i$, the utility function $u_i:A\rightarrow \mathbb{R}$ is \emph{one-to-one} (i.e., no ties in the matrix below), and (iii) no action is weakly dominated. 

\begin{center}
\begin{game}{3}{3}   &   $a_2$       &   $b_2$    &    $c_2$ \\
     $a_1$           &   $8,8$      &   $3,7$      &   $2,6$   \\ 
     $b_1$          &   $7,3$     &   $5,1$      &   $0,5$   \\  
     $c_1$           &   $6,2$     &   $1,4$      &   $4,0$  \\
\end{game}
\end{center}

The game ($\Theta$ is a sigleton) has one pure Nash equilibrium and one mixed Nash equilibrium:
\[
(a_1,a_2)\quad\text{and}\quad\left(\frac{1}{2}b_1+\frac{1}{2}c_1,\frac{1}{2}b_2+\frac{1}{2}c_2\right).
\]
The set of BCE is the set of convex combinations of the two Nash equilibria: for $t\in [0,1]$,
\[
p^t = t (a_1,a_2) + (1-t)\left(\frac{1}{2}b_1+\frac{1}{2}c_1,\frac{1}{2}b_2+\frac{1}{2}c_2\right).
\]

The game has only two separated BCE, namely, the two Nash equilibria. Indeed, for every $t\in (0,1)$ and every player $i$, the action recommendations $a_i$ and $b_i$ (or $c_i$) induce distinct posterior beliefs about the action of the opponent:
\[
p^t_{a_i}(a_j)=1 \quad\text{while}\quad p^t_{b_i} (b_j)=p^t_{b_i} (c_j)=\frac{1}{2}.
\]
Yet, $a_i$ is best response to the belief induced by $b_i$:
\[
\frac{1}{2}u_i(a_i,b_j)+\frac{1}{2}u_i(a_i,c_j)=\frac{5}{2}=\frac{1}{2}u_i(b_i,b_j)+\frac{1}{2}u_i(b_i,c_j).
\]
Thus, $p^t$ is not separated: player $i$ does not have an incentive to acquire information about the correlation device; they could just play $a_i$ without acquiring any information.

\section{Arbitrary Information Technologies}\label{sec:Arbitrary information technologies}

In this section, we characterize the predictions attainable as one ranges across \emph{all} information technologies. In particular, we do not require the information technology to be flexible or monotone. We also show it is without loss to require the technology to be flexible, and costs to be weakly monotone. Formally, a cost function $\icost_\player$ is \textbf{weakly monotone} if less informative experiment are weakly cheaper to acquire: if $\exper_\player, \experb_\player \in \Exper_{\player}$ are such that $\exper_\player \succsim \experb_\player$, then $\icost_\player(\exper_\player) \geq \icost_\player(\experb_\player)$.

\begin{samepage}
\begin{proposition}\label{pro:all_tech}
Fix a base game $\BGame$. An information technology $\IT$ exists that induces the outcome-value pair $(\out,\payvector)$ in an equilibrium of $(\BGame,\IT)$ if and only if 
\begin{enumerate}[(i)]
\item $\out$ is a BCE, and 
\item for every $\player\in \Player$, $\payvector_\player \in [\noinfoval_{\player}(\out),\grossval_{\player}(\out)].$
\end{enumerate}
In addition, for every player $\player$, one can choose $\Exper_\player$ flexible and $\icost_\player$ weakly monotone.
\end{proposition}
\end{samepage}

\begin{proof}
``If.'' Let $(\out,\payvector)$ be an outcome-value pair such that $\out$ is a BCE and, for every $\player\in\Player$, $\payvector_\player\in [\noinfoval_\player(\out),\grossval_\player(\out)]$. Since $\out$ is a BCE, by \cite{bergemann2016bayes} there exist an information structure $\mathcal{S}=(\Corstate,\corprior,(\Signal_\player,\exper_\player)_{\player\in\Player})$ and a profile of action plans $\aplan=(\aplan_\player)_{\player\in\Player}$ such that $\out$ is the outcome of $(\exper,\aplan)$, and for every player $\player$, $\aplan_\player$ maximizes 
\begin{equation}\label{eq:optimality_BCE}
\sum_{\act,\signal,\corstate,\paystate} \util_\player(\act,\paystate) \left(\aplan^\prime_{\player}(\act_\player|\signal_\player)\exper_\player(\signal_{\player}|\corstate,\paystate)\prod_{\playerb\neq \player} \aplan_{\playerb}(\act_{\playerb}|\signal_\playerb)\exper_{\playerb}(\signal_{\playerb}|\corstate,\paystate)\right)\corprior(\corstate|\paystate)\payprior(\paystate)
\end{equation}
over all $\aplan^\prime_\player\in \Aplan_\player$.

For every player $\player$, let $\Exper_\player=\{\experb_\player:\exper_\player\succeq \experb_\player\}$. In addition, take $\lambda_\player\in [0,1]$ such that 
\[
\payvector_\player=\lambda_\player\noinfoval_\player(\out) +(1-\lambda_\player)\grossval_\player(\out).
\]
For every $\experb_\player\in\Exper_\player$, set  $\icost_\player(\experb_\player)$ equal to
\[
\lambda_\player\left[\max_{\aplan_\player^\prime}\sum_{\act,\signal,\corstate,\paystate} \util_\player(\act,\paystate) \left(\aplan^\prime_{\player}(\act_\player|\signal_\player)\experb_\player(\signal_{\player}|\corstate,\paystate)\prod_{\playerb\neq \player} \aplan_{\playerb}(\act_{\playerb}|\signal_\playerb)\exper_{\playerb}(\signal_{\playerb}|\corstate,\paystate)\right)\corprior(\corstate|\paystate)\payprior(\paystate)-\noinfoval_\player(\out)\right].
\]
Notice that $\Exper_\player$ is flexible and $\icost_\player$ is weakly monotone. 

It follows from (\ref{eq:optimality_BCE}) that
$
\icost_\player(\exper_\player)=\lambda_\player\left(\grossval_\player(\out)-\noinfoval_\player(\out)\right),
$
which in turn implies that
\[
\sum_{\act,\signal,\corstate,\paystate} \util_\player(\act,\paystate) \left(\prod_{\playerb} \aplan_{\playerb}(\act_{\playerb}|\signal_\playerb)\exper_{\playerb}(\signal_{\playerb}|\corstate,\paystate)\right)\corprior(\corstate|\paystate)\payprior(\paystate)-\icost_\player(\exper_\player)=\payvector_\player.
\]
We also see that for every $\experb_\player\in\Exper_\player$
\begin{align*}
    & \sum_{\act,\signal,\corstate,\paystate} \util_\player(\act,\paystate) \left(\prod_{\playerb} \aplan_{\playerb}(\act_{\playerb}|\signal_\playerb)\exper_{\playerb}(\signal_{\playerb}|\corstate,\paystate)\right)\corprior(\corstate|\paystate)\payprior(\paystate)-\icost_\player(\exper_\player)\\
=\, & \max_{\aplan^\prime_\player}\sum_{\act,\signal,\corstate,\paystate} \util_\player(\act,\paystate)\left( \aplan^\prime_{\player}(\act_\player|\signal_\player)\exper_\player(\signal_{\player}|\corstate,\paystate)\prod_{\playerb\neq \player} \aplan_{\playerb}(\act_{\playerb}|\signal_\playerb)\exper_{\playerb}(\signal_{\playerb}|\corstate,\paystate)\right)\corprior(\corstate|\paystate)\payprior(\paystate)-\icost_\player(\exper_\player)\\
=\, & \lambda_\player \noinfoval_\player(\out)+(1-\lambda_\player)\max_{\aplan^\prime_\player}\sum_{\act,\signal,\corstate,\paystate} \util_\player(\act,\paystate) \left(\aplan^\prime_{\player}(\act_\player|\signal_\player)\exper_\player(\signal_{\player}|\corstate,\paystate)\prod_{\playerb\neq \player} \aplan_{\playerb}(\act_{\playerb}|\signal_\playerb)\exper_{\playerb}(\signal_{\playerb}|\corstate,\paystate)\right)\corprior(\corstate|\paystate)\payprior(\paystate) \\  
\geq \, & \lambda_\player \noinfoval_\player(\out)+(1-\lambda_\player)\max_{\aplan^\prime_\player}\sum_{\act,\signal,\corstate,\paystate} \util_\player(\act,\paystate) \left(\aplan^\prime_{\player}(\act_\player|\signal_\player)\experb_\player(\signal_{\player}|\corstate,\paystate)\prod_{\playerb\neq \player} \aplan_{\playerb}(\act_{\playerb}|\signal_\playerb)\exper_{\playerb}(\signal_{\playerb}|\corstate,\paystate)\right)\corprior(\corstate|\paystate)\payprior(\paystate) \\
=\, & \max_{\aplan^\prime_\player}\sum_{\act,\signal,\corstate,\paystate} \util_\player(\act,\paystate) \left(\aplan^\prime_{\player}(\act_\player|\signal_\player)\experb_\player(\signal_{\player}|\corstate,\paystate)\prod_{\playerb\neq \player} \aplan_{\playerb}(\act_{\playerb}|\signal_\playerb)\exper_{\playerb}(\signal_{\playerb}|\corstate,\paystate)\right)\corprior(\corstate|\paystate)\payprior(\paystate)-\icost_\player(\experb_\player),
\end{align*}
where the first equality follows from (\ref{eq:optimality_BCE}) and the weak inequality from $\exper_\player\succeq\experb_\player$. We conclude $(\exper,\aplan)$ is an equilibrium of $(\BGame,\IT)$ with $\IT:=(\Corstate, \corprior, (\Signal_\player,\Exper_\player,\icost_\player)_{\player\in\Player})$; in addition, $(\out,\payvector)$ is the outcome-value pair corresponding to $(\exper,\aplan)$.

``Only if.'' Let $(\out,\payvector)$ be the outcome-value pair of an equilibrium $(\exper,\aplan)$ of an information acquisition game $(\BGame,\IT)$, with $\IT=(\Corstate, \corprior, (\Signal_\player,\Exper_\player,\icost_\player)_{\player\in\Player})$. Define the information structure $\mathcal{S}=(\Corstate, \corprior, (\Signal_\player,\exper_\player)_{\player\in\Player})$. Since $(\exper,\aplan)$ is an equilibrium of $(\BGame,\IT)$, $\aplan$ is an equilibrium of $(\BGame,\mathcal{S})$. It follows from \cite{bergemann2016bayes} that $\out$ is a BCE.

For every player $i$, $\icost_\player(\exper_\player)\geq 0$, which implies that $\payvector_\player\leq \grossval_\player(\out)$. In addition, by hypothesis there exists an experiment $\exper_\player^\prime$ such that  $\icost_\player(\exper_\player^\prime) = 0$. Thus, since $(\exper_\player,\aplan_\player)$ is a best response to $(\exper_{-\player},\aplan_{-\player})$, we have that
\[
\payvector_\player  \geq \max_{\aplan^\prime_\player}\sum_{\act,\signal,\corstate,\paystate} \util_\player(\act,\paystate) \aplan^\prime_{\player}(\act_\player|\signal_\player)\exper_\player(\signal_{\player}|\corstate,\paystate)\prod_{\playerb\neq \player} \aplan_{\playerb}(\act_{\playerb}|\signal_\playerb)\exper_{\playerb}(\signal_{\playerb}|\corstate,\paystate)\corprior(\corstate|\paystate)\payprior(\paystate)
 \geq \noinfoval_\player(\out).
\]
We conclude that $\payvector_\player\in [\grossval_\player(\out),\noinfoval_\player(\out)]$.

\end{proof}

\section{Strict BCE: Single-Agent Settings}\label{sec:generic_singla_agent}

A BCE $\out$ is \textbf{strict} if all $\player\in\Player$, $\act_\player\in\supp_\player(\out)$, and $\actb_\player \in \Act_\player$ with $\actb_\player\neq \act_\player$,
\[
\sum_{\act_{-\player},\paystate}\left(\util_{\player}(\act_{\player},\act_{-\player},\paystate) - \util_{\player}(\actb_{\player},\act_{-\player},\paystate) \right)\out(\act_{\player},\act_{-\player},\paystate) > 0. 
\]
In the main text, discussing Theorem \ref{thm:genericity}, we mentioned the following result:  
\begin{proposition}\label{pro:strict_single}
Let $\Player=\{\player\}$ be a singleton. For generic $\util_\player$, the set of strict BCE is dense in the BCE set.
\end{proposition}
We expect the result to be known in the literature. However, we could not find a good reference. Thus, next we provide a self-contained proof. The proof relies on two lemmas on dominated actions. A mixed action $\mact_\player\in\Delta(\Act_\player)$ \textbf{weakly dominates} a pure action $\act_\player\in\Act_\player$ if for all $\act_{-\player}\in\Act_{-\player}$ and $\paystate\in\Paystate$,
\[
\sum_{\actb_\player}\util_\player(\actb_\player,\act_{-\player},\paystate)\mact_\player(\actb_\player)\geq \util_\player(\act_\player,\act_{-\player},\paystate). 
\]
\begin{lemma}\label{lem:minmax_single}
The following statements are equivalent:
\begin{itemize}
    \item[(i)] There is no belief $\mu_{\act_\player}\in\Delta(\Act_{-\player}\times\Paystate)$ for which $\act_\player$ is the unique best response.
    \item[(ii)] There is a mixed action $\mact_\player\in\Delta(\Act_\player\setminus\{\act_\player\})$ that weakly dominates $\act_\player$.
\end{itemize}
\end{lemma}
\begin{proof}
Condition (i) can be rewritten as 
\[
\max_{\mu_\player\in \Delta(\Act_{-\player}\times\Paystate)}\min_{\actb_\player\in\Act_\player\setminus\{\act_\player\}} \sum_{\act_{-\player},\paystate} (\util_\player(\act_\player,\act_{-\player},\paystate)-\util_\player(\actb_\player,\act_{-\player},\paystate))\mu_\player(\act_{-\player},\paystate)\leq 0.
\]
Equivalently,
\[
\max_{\mu_\player\in \Delta(\Act_{-\player}\times\Paystate)}\min_{\mact_\player\in\Delta(\Act_\player\setminus\{\act_\player\})} \sum_{\act_{-\player},\paystate} (\util_\player(\act_\player,\act_{-\player},\paystate)-\util_\player(\actb_\player,\act_{-\player},\paystate))\mu_\player(\act_{-\player},\paystate)\mact_\player(\actb_\player)\leq 0.
\]
By the minimax theorem (e.g., \citealp[Corollary 37.3.2]{rockafellar1970convex}), the above inequality holds if and only if
\[
\min_{\mact_\player\in\Delta(\Act_\player\setminus\{\act_\player\})} \max_{\mu_\player\in \Delta(\Act_{-\player}\times\Paystate)} \sum_{\act_{-\player},\paystate} (\util_\player(\act_\player,\act_{-\player},\paystate)-\util_\player(\actb_\player,\act_{-\player},\paystate))\mu_\player(\act_{-\player},\paystate)\mact_\player(\actb_\player)\leq 0.
\]
Equivalently, 
\[
\min_{\mact_\player\in\Delta(\Act_\player\setminus\{\act_\player\})} \max_{\act_{-\player},\paystate} \sum_{\act_{-\player},\paystate} (\util_\player(\act_\player,\act_{-\player},\paystate)-\util_\player(\actb_\player,\act_{-\player},\paystate))\mact_\player(\actb_\player)\leq 0.
\]
which is another way of expressing condition (ii).\end{proof}

A mixed action $\mact_\player\in\Delta(\Act_\player)$ \textbf{strictly dominates} a pure action $\act_\player\in\Act_\player$ if for all $\act_{-\player}\in\Act_{-\player}$ and $\paystate\in\Paystate$,
\[
\sum_{\actb_\player}\util_\player(\actb_\player,\act_{-\player},\paystate)\mact_\player(\actb_\player)> \util_\player(\act_\player,\act_{-\player},\paystate). 
\]

\begin{lemma}\label{lem:dominance_single}
Let $\Player=\{\player\}$ be a singleton. For generic $\util_\player$, if an action $\act_\player$ is weakly dominated by some mixed action $\mact_\player\in\Delta(\Act_\player\setminus\{\act_\player\})$, then it is strictly dominated by some mixed action $\beta_\player\in \Delta(\Act_\player)$. 
\end{lemma}
\begin{proof}
Let $\act_\player$ be an action that is weakly dominated by a mixed action $\mact_\player\in\Delta(\Act_\player\setminus\{\act_\player\})$. Let $\Act^\prime_\player$ be the support of $\mact_\player$, and let  $\Paystate^\prime$ be set of states $\theta$ for which 
\begin{equation}\label{eq:linear_sys}
\util_\player(\act_\player,\paystate)=\sum_{\actb_\player}\util_\player(\actb_\player,\paystate)\mact_\player(\actb_\player).
\end{equation}
Let $m$ be the cardinality of $\Act^\prime_\player$, and let $n$ be the cardinality of $\Paystate^\prime$. We consider the $m\times n$ matrix $M\in \mathbb{R}^{\Act^\prime_\player\times\Paystate^\prime}$ given by
\[
M(\actb_\player,\paystate)=\util_\player(\act_\player,\paystate)-\util_\player(\actb_\player,\paystate).
\]
For generic $\util_\player$, the matrix $M$ has full rank. By (\ref{eq:linear_sys}), the rows of $M$ are linearly dependent. Thus, the rank of $M$ must be $n$, the number of columns. We obtain that the row space of $M$ has dimension $n$. Hence, we can find $\beta_\player\in\mathbb{R}^{\Act^\prime_\player}$ such that for every $\paystate\in \Paystate^\prime$
\[
\sum_{\actb_\player}(\util_\player(\act_\player,\paystate)-\util(\actb_\player,\paystate))\beta_\player(\actb_\player)<0.
\]
For every $t>0$, we define $\mact^t_\player\in \mathbb{R}^{\Act^\prime_\player}$ by
\[
\mact^t_\player(\actb_\player)=\frac{\mact_\player(\actb_\player)+t\beta_\player(\actb_\player)}{\sum_{\actc_\player} \mact_\player(\actc_\player)+t\beta_\player(\actc_\player)}.
\]
For $t$ sufficiently small, $\mact^t_\player$ is a mixed action that strictly dominates $\act_\player$.
\end{proof}

We are now ready to prove the proposition on strict BCE.

\begin{proof}[Proof of Proposition \ref{pro:strict_single}]
Let $\Act_\player^*$ be the set of actions that are not strictly dominated. Since $\util_\player$ is generic, it follows from Lemma \ref{lem:dominance_single} that each $\act_\player\in \Act_\player^*$ is not weakly dominated by a mixed action $\mact_\player\in\Delta(\Act_\player\setminus\{\act_\player\})$. By Lemma \ref{lem:minmax_single}, there is a belief $\mu_{\act_\player}\in\Delta(\Theta)$ for which $\act_\player$ is the unique best response.

Since $\payprior$ has full support, we can find $\nu\in\Delta(\Paystate)$ and for every $\act_\player\in \Act_\player^*$, $t_{\act_\player}\in (0,1)$---with $\sum_{\act_\player\in \Act_\player^*}t_{\act_\player}\leq 1$---such that 
\[
\payprior=\sum_{\act_\player\in\Act_\player^*}t_{\act_\player}\mu_{\act_\player} +\left(1-\sum_{\act_\player\in\Act_\player^*} t_{\act_\player}\right)\nu.
\]
Let $\act_\player^*$ be a best response to $\nu$; necessarily, $\act_\player^*\in \Act_\player^*$. Define the outcome $\out\in\Delta_\pi(\Act_\player\times\Paystate)$ as follows:
\[
\out(\act_\player,\theta)=\begin{cases}
t_{\act_\player} \mu_{\act_\player}(\theta) &\text{if }\act_\player\in\Act_\player^*\setminus\{\act_\player^*\},\\
t_{\act^*_\player} \mu_{\act^*_\player}(\theta) + \left(1-\sum_{\act_\player\in \Act_\player^*}t_{\act_\player}\right)\nu(\theta) &\text{if }\act_\player=\act^*_\player,\\
0 &\text{otherwise.}
\end{cases}
\]
The outcome $\out$ is a strict BCE. Moreover, if $\outb$ is a BCE, then
\[
\supp_\player(\outb)\subseteq \Act^*_\player = \supp_\player(\out).
\]
Thus, the set of outcomes
\[
\left\{s \outb +(1-s)\out: s\in (0,1)\text{ and }\outb\in \BCE \right\}
\]
is a subset of the set of strict BCE, and it is dense in the BCE set. We conclude that (for generic $\util_\player$) the set of strict BCE is dense in the BCE set.\end{proof}

\section{Non-generic Environments}

\subsection{Proofs of Proposition \ref{pro:BCE equals cl sBCE iff stuff} and Theorem~\ref{thm:Dense or nowhere dense}}

In this section we prove generalizations of Proposition \ref{pro:BCE equals cl sBCE iff stuff} and Theorem~\ref{thm:Dense or nowhere dense} that apply locally to a closed convex set of BCEs. 

Fix a base game $\BGame$; denote by $\BCE$ the set of all BCEs, and by $\sBCE$ the set of all sBCEs. Let $\BCEset\subseteq\BCE$ be a non-empty closed convex set. For a player $\player$, an action $\act_\player$ $\BCEset$-\textbf{jeopardizes} an action $\actb_\player$ if, for every $\out \in \BCEset$ with $\actb\in\supp_\player(\out)$, $\act_\player\in BR(\out_{\actb_\player})$. We denote by $\Jeop_{\BCEset}(\actb_\player)$ the set of actions that $\BCEset$-jeopardizes $\actb_\player$. Just like the standard jeopardization concept, one has $\actb_\player \in J_{\BCEset}(\actb_\player)$ for all $\actb_\player$. 

An outcome $\out \in \BCEset$ has $\BCEset$-\textbf{maximal support} if the support of every other $\outb\in \BCEset$ is contained by the support of $\out$. An outcome $\out\in \BCEset$ is $\BCEset$-\textbf{minimally mixed} if it has $\BCEset$-maximal support and
\[
\outb_{\act_\player}\neq\outb_{\actb_\player}\quad\text{implies}\quad\out_{\act_\player}\neq\out_{\actb_\player}.
\]
if for every $\outb \in \BCEset$, $\player\in \Player$, and $\act_\player,\actb_\player\in\supp_\player(\out)$, 

We are now ready to state the local versions of Proposition~\ref{pro:BCE equals cl sBCE iff stuff} and Theorem~\ref{thm:Dense or nowhere dense} that we prove in this section. 

\begin{proposition}\label{pro:B-BCE equals cl B-sBCE iff stuff}
The following statements are equivalent:
\begin{enumerate}[(i)]
\item\label{pro:B-BCE equals cl B-sBCE: closure} The set $\sBCE \cap \BCEset$ is dense in $\BCEset$. 
\item\label{pro:B-BCE equals cl B-sBCE: min-mixed sBCE} A $\BCEset$-minimally mixed sBCE exists. 
\item\label{pro:B-BCE equals cl B-sBCE: jeopardization} For all $p \in \BCEset$, $\player\in \Player$, $\act_\player,\actb_\player\in \supp_\player(\out)$, 
\[
\out_{\act_\player}\neq\out_{\actb_\player}\quad\text{implies}\quad \Jeop_{\BCEset}(\act_\player)\cap \Jeop_{\BCEset}(\actb_\player)=\varnothing.
\] 
\end{enumerate}
\end{proposition}

\begin{theorem}\label{thm:B-Dense or nowhere dense}
The set $\sBCE \cap \BCEset$ is either dense or nowhere dense in $\BCEset$.
\end{theorem}

As a first step, we prove a basic lemma about best responses. In what follows, for $\out \in \Delta(\Act\times\Paystate)$, $\act_{\player}\in\supp_\player(\out)$, and $\actb_{\player}\in \Act_{\player}$, take
$\util_{\player}\left(\actb_{\player},\out_{\act_{\player}}\right)\in\mathbb{R}$ to be
\[
\util_{\player}\left(\actb_{\player},\out_{\act_\player}\right)=\sum_{\act_{-\player},\paystate}\util_{\player}(\actb_{\player},\act_{-\player},\paystate)\out_{\act_{\player}}(\act_{-\player},\paystate).
\]

\begin{lemma}\label{lem:BR inclusion under convex combinations}
For every $\wt\in(0,1)$, $\out,\outb\in \BCE$, $\player\in \Player$,
and $\act_{\player}\in \supp_\player (\out)$,
\[
\BR\left(\left(\wt\out+(1-\wt)\outb\right)_{\act_{\player}}\right)\subseteq \BR\left(\out_{\act_{\player}}\right).
\]
\end{lemma}

\begin{proof}
Take $\actb_{\player}\in \BR\left(\left(\wt\out+(1-\wt)\outb\right)_{\act_{\player}}\right)$. If $\act_{\player}\notin \supp_\player (\outb)$, then 
\[
\left(\wt\out+(1-\wt)\outb\right)_{\act_{\player}}=\out_{\act_{\player}},
\]
which immediately implies the desired result.

Suppose now that $\act_{\player}\in \supp_\player (\outb)$. Since $\out,\outb\in \BCE$, we have
\[
\util_{\player}(\act_{\player},\out_{\act_{\player}})\geq \util_{\player}(\actb_{\player},\out_{\act_{\player}})\quad\text{and}\quad \util_{i}(\act_{\player},\outb_{\act_{\player}})\geq \util_{\player}(\actb_{\player},\outb_{\act_{\player}}).
\]
Simple algebra shows that there exists $\wtb\in(0,1)$ such that
\[
\left(\wt \out + (1-\wt) \outb \right)_{\act_{\player}}=\wtb \out_{\act_{\player}}+(1-\wtb)\outb_{\act_{\player}}.
\]
Since $\actb_{\player}\in \BR\left(\left(\wt \out +(1-\wt)\outb\right)_{\act_{\player}}\right)$,
we obtain that
\begin{align*}
\wtb \util_{\player}(\actb_{\player},\out_{\act_{\player}})+(1-s)\util_{\player}(\actb_{\player},\outb_{\act_{\player}}) & =\util_{\player}(\actb_{\player},\wtb \out_{\act_{\player}}+(1-\wtb)\outb_{\act_{\player}})\\
 & \geq \util_{\player}(\act_{\player},\wtb \out_{\act_{\player}}+(1-\wtb)\outb_{\act_{\player}})\\
 & =\wtb\util_{\player}(\act_{\player},\out_{\act_{\player}})+(1-\wtb)\util_{\player}(\act_{\player},\outb_{\act_{\player}}).
\end{align*}
We conclude that $\util_{\player}(\act_{\player},\out_{\act_{\player}})=\util_{\player}(\actb_{\player},\out_{\act_{\player}})$
and $\util_{\player}(\act_{\player},\outb_{\act_{\player}})=\util_{\player}(\actb_{\player},\outb_{\act_{\player}})$. It follows from
$\out\in \BCE$ that
$
\actb_\player \in \BR\left(\out_{\act_{\player}}\right).
$
\end{proof}

Next, we show that taking convex combinations of BCEs usually preserve the set of action recommendations that lead to different beliefs. 

\begin{lemma}\label{lem:most convex combinations preserve difference in posteriors}For every $\out,\outb\in\Delta(\Act \times \Paystate)$, $\player \in \Player$, and
$\act_{\player},\actb_{\player}\in \supp_\player(\out)$ with $\out_{\act_{\player}}\neq \out_{\actb_{\player}}$,
there are at most two $\wt\in(0,1)$ such that
\begin{equation}
(\wt\out+(1-\wt)\outb)_{\act_{\player}}=(\wt\out+(1-\wt)\outb)_{\actb_{\player}}.\label{eq:convex combination of outcomes actually yields equal beliefs}
\end{equation}
\end{lemma}
\begin{proof}
Note that $t\in(0,1)$ is a solution of (\ref{eq:convex combination of outcomes actually yields equal beliefs})
if and only if for every $\act_{-\player}\in \Act_{-i}$ and $\paystate \in \Paystate$,
\begin{align}
 & (\wt\out(\act_{\player},\act_{-\player},\paystate)+(1-\wt)\outb(\act_{\player},\act_{-\player},\paystate))\left(\wt\out(\actb_{\player})+(1-\wt)\outb(\actb_{\player})\right)\nonumber \\
= & (\wt\out(\actb_{\player},\act_{-\player},\paystate)+(1-\wt)\outb(\actb_{\player},\act_{-\player},\paystate))\left(\wt\out(\act_{\player})+(1-\wt)\outb(\act_{\player})\right).\label{eq:when does convex combination of outcomes yield equal beliefs}
\end{align}
Each equation (\ref{eq:when does convex combination of outcomes yield equal beliefs}) is polynomial in $\wt$, with degree at most two. Since
$\out_{\act_{\player}}\neq \out_{\actb_{\player}}$, at least one such polynomial equation does not have degree zero and,
therefore, has at most two solutions. We deduce that (\ref{eq:convex combination of outcomes actually yields equal beliefs})
has at most two solutions for $\wt\in(0,1)$.
\end{proof}

Our next goal is to show that $\BCEset$-minimally mixed BCEs are the norm rather than the exception. As an intermediate step, we first show the set of $\BCEset$-minimally mixed BCEs is non-empty.

\begin{lemma}\label{lem:a min-mix BCE exists}
    A $\BCEset$-minimally mixed BCE exists. 
\end{lemma}
\begin{proof}
For every $\out\in \BCEset$, define the set
\[
X(\out)=\bigcup_{\player}\left\{ (\act_{\player},\actb_{\player}):\act_{\player},\actb_{\player}\in\supp_\player(\out)\text{ and }\out_{\act_{\player}}\neq \out_{\actb_{\player}}\right\} .
\]
Note that $\out\in\BCEset$ is $\BCEset$-minimally mixed if and only if it has $\BCEset$-maximal support and for every $\outb\in \BCEset$, $X(\outb)\subseteq X(\out)$.

Since the set $\Act\times\Paystate$ is finite and the set $\BCEset$ is convex, we can
find a $\BCEset$-maximal support $\out\in \BCEset$ such that for every $\BCEset$-maximal-support $q\in \BCEset$,
the cardinality of $X(\out)$ is larger than the cardinality of $X(\outb)$.

We now show that $\out$ is $\BCEset$-minimally mixed. Fix an arbitrary $\outb\in \BCEset$. For every $\wt\in(0,1)$, define $\out^{\wt}=\wt \out+(1-\wt)\outb$, which belongs to $\BCEset$ because $\BCEset$ is convex. Since $\out$ has $\BCEset$-maximal support, the same is true for $\out^{\wt}$. Thus, the cardinality of
$X(\out)$ is larger than the cardinality of $X(\out^{\wt})$. By Lemma \ref{lem:most convex combinations preserve difference in posteriors}, we can find $\wt\in(0,1)$ such that $X(\out)\subseteq X(\out^{\wt})$ and $X(\outb)\subseteq X(\out^{\wt})$.
This shows that $X(\outb)\subseteq X(\out)$; otherwise, the cardinality
of $X(\out^{\wt})$ would be strictly larger than the cardinality of $X(\out)$.
We conclude that $\out$ is $\BCEset$-minimally mixed.
\end{proof}

We now show that the $\BCEset$-minimally mixed BCEs includes most of the BCEs in $\BCEset$ in a precise sense.

\begin{lemma}\label{lem:B-BCEs are open and dense in B}
    The set of $\BCEset$-minimally mixed BCEs is open and dense in $\BCEset$.
\end{lemma}
\begin{proof}
Let $\BCEset_{M}$ denote the set of $\BCEset$-minimally mixed BCEs. We first argue that $\BCEset_{M}$ is open in $\BCEset$. Towards this goal, note the following sets are open in $\BCEset$ for every $\player\in \Player$ and $\act_\player, \actb_\player \in \Act_\player$:
\[
\{\out \in \BCEset: \out(\act_\player)>0\}, 
\quad \text{ and } \quad 
\{\out \in \BCEset: \out(\act_\player)\out(\actb_\player)>0 \text{ and } \out_{\act_\player} \neq \out_{\actb_\player}\}.
\]
Since $\Act$ is finite, we obtain that $\BCEset_{M}$ equals the intersection of a finite number of open subsets of $\BCEset_{M}$. It follows $\BCEset_{M}$ is open in $\BCEset$.

To see $\BCEset_{M}$ is dense in $\BCEset$, fix some $\outb \in \BCEset$. Take $\out$ to be a $\BCEset$-minimally mixed BCE, which exists by Lemma~\ref{lem:a min-mix BCE exists}. For every $\wt \in (0,1)$, define $\out^{\wt} = \wt \out + (1-\wt)\outb$. Because $\out$ has $\BCEset$-maximal support, the same is true for $\out^{\wt}$ for all $\wt \in (0,1)$. Moreover, by Lemma~\ref{lem:most convex combinations preserve difference in posteriors}, a finite set $T \subseteq (0,1)$ exists such that for all $\wt \in (0,1)\setminus T$, $\player\in \Player$, and $\act_\player, \actb_\player \in \supp_\player(\out)$,
\[
\out_{\act_\player} \neq \out_{\actb_\player}\quad\text{implies}\quad\out^{\wt}_{\act_\player} \neq \out^{\wt}_{\actb_\player}.
\]
Thus, $\out^{\wt}$ is a $\BCEset$-minimally mixed BCE for all $\wt \in (0,1) \setminus T$. Thus, $\outb$ is a limit point of $\{\out^{\wt}: \wt \in (0,1)\setminus T\}$, which implies it is a limit point of $\BCEset_{M}$. 
\end{proof}

We are now ready to prove Proposition~\ref{pro:B-BCE equals cl B-sBCE iff stuff} and Theorem~\ref{thm:B-Dense or nowhere dense}.

\begin{proof}[Proof of Proposition~\ref{pro:B-BCE equals cl B-sBCE iff stuff}] 
That \eqref{pro:B-BCE equals cl B-sBCE: closure} implies \eqref{pro:B-BCE equals cl B-sBCE: min-mixed sBCE} follows from Lemma~\ref{lem:B-BCEs are open and dense in B}. 

We now show \eqref{pro:B-BCE equals cl B-sBCE: min-mixed sBCE} implies \eqref{pro:B-BCE equals cl B-sBCE: jeopardization}. Let $\outb$ be a $\BCEset$-minimally mixed sBCE. Fix any $\out \in \BCEset$, $\player\in \Player$ and $\act_\player, \actb_\player \in \supp_\player(\out)$ such that $\out_{\act_\player}\neq \out_{\actb_\player}$. Since $\outb$ is $\BCEset$-minimally mixed, $\act_\player,\actb_\player \in \supp_\player(\outb)$ (because $\outb$ has $\BCEset$-maximal support) and $\outb_{\act_\player}\neq \outb_{\actb_\player}$. Thus,
\[
\varnothing = \BR (\outb_{\act_\player}) \cap \BR (\outb_{\actb_\player}) \supseteq \Jeop_{\BCEset}(\act_\player)\cap \Jeop_{\BCEset}(\actb_\player),
\]
where first we use the separation constraint, and then the fact that $\Jeop_{\BCEset}(\actc_\player) = \cap_{\tilde\out \in \BCEset} \BR (\tilde\out_{\actc_{\player}})$ for all $\actc_\player \in \Act_\player$. We conclude \eqref{pro:B-BCE equals cl B-sBCE: min-mixed sBCE} implies \eqref{pro:B-BCE equals cl B-sBCE: jeopardization}. 

Finally, we argue \eqref{pro:B-BCE equals cl B-sBCE: jeopardization} implies \eqref{pro:B-BCE equals cl B-sBCE: closure}. Fix any $\out \in \BCEset$. Because $\Act$ is finite and $\BCEset$ is convex, it follows from Lemma~\ref{lem:BR inclusion under convex combinations} that we can find $\outb \in \BCEset$ such that $\outb$ has $\BCEset$-maximal support and 
\begin{equation}\label{eq:iii_to_i}
\BR(\outb_{\act_\player}) = \Jeop_{\BCEset}(\outb_{\act_\player})
\end{equation}
for all $\player\in \Player$ and $\act_\player \in \supp_\player(\outb)$. 

For $\wt \in (0,1)$, let $\out^{\wt} = \wt \out + (1-\wt)\outb$. We claim that $\out^{\wt} \in \sBCE \cap \BCEset$. That $\out^{\wt} \in \BCEset$ follows from $\BCEset$ being convex. To see $\out^{\wt}$ is a sBCE, take any $\player\in\Player$ and $\act_{\player},\actb_{\player} \in \supp_{\player}(\out^{\wt})$ such that $\out^{\wt}_{\act_\player}\neq\out^{\wt}_{\actb_\player}$. Since $\outb$ has maximal support, $\act_{\player},\actb_{\player} \in \supp_{\player}(\outb)$. Then,
\[
\BR(\out^{\wt}_{\act_\player})\cap \BR(\out^{\wt}_{\actb_\player})\subseteq \BR(\outb_{\act_\player})\cap \BR(\outb_{\actb_\player})=\Jeop_{\BCEset}(\act_\player) \cap \Jeop_{\BCEset}(\actb_\player) = \varnothing,
\]
where first we use Lemma~\ref{lem:BR inclusion under convex combinations}, then (\ref{eq:iii_to_i}), and finally Proposition~\ref{pro:B-BCE equals cl B-sBCE iff stuff}-\eqref{pro:B-BCE equals cl B-sBCE: jeopardization}. We conclude $\out^{\wt} \in \sBCE \cap \BCEset$ for all $\wt \in (0,1)$. Proposition~\ref{pro:B-BCE equals cl B-sBCE iff stuff}-\eqref{pro:B-BCE equals cl B-sBCE: closure} then follows from $\out = \lim_{\wt \rightarrow 1} \out^{\wt}$. 
\end{proof}

\begin{proof}[Proof of Theorem~\ref{thm:B-Dense or nowhere dense}]
It is enough to prove that if $\sBCE\cap \BCEset$ is not nowhere dense in $\BCEset$, then it is dense in $\BCEset$. Suppose $\sBCE\cap \BCEset$ is dense in some non-empty set $\tilde\BCEset \subseteq \BCEset$ that is open in $\BCEset$. Let $\BCEset_{M}$ the set of $\out \in \BCEset$ that are $\BCEset$-minimally mixed.
Note $\tilde \BCEset \cap \BCEset_{M}$ is open (in $\BCEset$) and non-empty by Lemma~\ref{lem:B-BCEs are open and dense in B}. Because $\sBCE\cap \BCEset$ is dense in $\tilde\BCEset$, we obtain that $\left(\sBCE\cap \BCEset\right)\cap(\tilde \BCEset \cap \BCEset_{MM})$ is non-empty. Thus, we have found a $\BCEset$-minimally mixed sBCE. That $\sBCE\cap \BCEset$ is dense in $\BCEset$ then follows from Proposition~\ref{pro:B-BCE equals cl B-sBCE iff stuff}. 
\end{proof}

\subsection{Checking for Equal Beliefs}\label{sec:Checking for Equal Beliefs}

To check the conditions of Proposition~\ref{pro:BCE equals cl sBCE iff stuff}, knowing which actions induce different beliefs for some BCE is useful. In this section, we prove a result that shows how to find actions that lead to different beliefs in a closed convex set of outcomes $\BCEset \subseteq \Delta(\Act \times \Paystate)$.\footnote{Neither the obedience nor the separation constraint play any role in this section.} 

For a player $\player$, say an action $\act_\player$ is \textbf{$\BCEset$-coherent} if a $\out \in \BCEset$ exists with $\out(\act_\player) >0$.\footnote{Our notion of $\BCEset$-coherent is inspired by the notion of coherence in \cite{nau1990coherent}.} Let $\textbf{0}$ be the all-zeros vector in $\mathbb{R}^{\Act_{-\player}\times\Paystate}$; in what follows, we use the convention that $\out_{\act_\player}=\textbf{0}$ for every $\out \in \Delta(\Act\times \Paystate)$ and $\act_\player \in \Act_\player$ such that $\out(\act_\player) = 0$. As in the previous section, we say that an outcome $\out\in \BCEset$ has \textbf{$\BCEset$-maximal support} if the support of every other $\outb\in \BCEset$ is contained by the support of $\out$.

\begin{proposition}\label{pro:local equal-belief actions via extreme points}
Fix a player $\player$ and two $\BCEset$-coherent actions $\act_\player, \actb_\player \in \Act_\player$. Then every $\out\in \BCEset$ with $\act_\player,\actb_\player \in \supp_\player(\out)$ has $\out_{\act_\player}=\out_{\actb_\player}$ if and only if one of the following two conditions hold:
\begin{enumerate}[(i)]
\item \label{pro:local equal-belief actions via extreme points: equal beliefs} A $\belief \in \Delta(\Act_{-\player}\times \Paystate)$ exists such that for all $\out \in \ext(\BCEset)$, $\{\out_{\act_\player},\out_{\actb_\player}\}\subseteq\{\belief,\textbf{0}\}$.
\item \label{pro:local equal-belief actions via extreme points: equal ratios} 
 A constant $\const >0$ exists such that for all $\out \in \ext (\BCEset)$, $\out(\act_\player)\out_{\act_\player} = \const\out(\actb_\player)\out_{\actb_\player}$.
\end{enumerate}
\end{proposition}

Thus, to know whether a pair of actions leads to the same beliefs in all outcomes in $\BCEset$, it is enough to check the extreme points of $\BCEset$ for one of two properties. The first property states these actions induce the same beliefs in all of the set's extreme points. The second property requires the likelihood ratio for these actions to be constant across all these extreme points. 

To prove the proposition, we need the following lemma.

\begin{lemma}\label{lem:equal beliefs at a cvx comb of outcomes w same beliefs}
Fix a player $\player$ and two actions $\act_\player,\actb_\player \in \Act_\player$. Let $\out,\outb \in \Delta(\Act\times\Paystate)$ such that $\{\act_\player,\actb_\player \} \subseteq \supp_\player(\out) \cup \supp_\player(\outb)$. Suppose $\outc_{\act_\player}=\outc_{\actb_\player}$ for all $\outc \in \{\out,\outb\}$ with $\{\act_\player,\actb_\player \} \subseteq \supp_\player(\outc)$. If $(\wt \out + (1-\wt)\outb)_{\act_\player} = (\wt \out + (1-\wt)\outb)_{\actb_\player}$ for some $\wt \in (0,1)$, then one of the following two conditions hold:
\begin{enumerate}[(i)]
\item \label{lem:equal beliefs at a cvx comb of outcomes w same beliefs: equal beliefs} A $\belief \in \Delta(\Act_{-\player}\times \Paystate)$ exists such that for all $\outc\in \{\out,\outb\}$, $\{\outc_{\act_\player},\outc_{\actb_\player}\}\subseteq \{\belief,\textbf{0}\}$. 

\item \label{lem:equal beliefs at a cvx comb of outcomes w same beliefs: equal ratios} A constant $\const >0$ exists such that for all $\outc \in \{\out,\outb\}$, $\outc(\act_\player) = \const\outc(\actb_\player)$.
\end{enumerate}

\end{lemma}

\begin{proof}
Let $\out^{\wt}:= \wt \out + (1-\wt)\outb$. We proceed by contradiction: we assume that Lemma~\ref{lem:equal beliefs at a cvx comb of outcomes w same beliefs}-\eqref{lem:equal beliefs at a cvx comb of outcomes w same beliefs: equal beliefs} and Lemma~\ref{lem:equal beliefs at a cvx comb of outcomes w same beliefs}-\eqref{lem:equal beliefs at a cvx comb of outcomes w same beliefs: equal ratios} both fail and show that $\out^{\wt}_{\act_\player}\neq\out^{\wt}_{\actb_\player}$.

We begin by noting that one can rewrite the condition that $\outc_{\act_\player}=\outc_{\actb_\player}$ for all $\outc \in \{\out,\outb\}$ with $\{\act_\player,\actb_\player \} \subseteq \supp_\player(\outc)$ as
\begin{equation}\label{eq:equal beliefs or zero probability}
\outc(\act_\player)\outc(\actb_\player)\outc_{\act_\player} =
\outc(\act_\player)\outc(\actb_\player)\outc_{\actb_\player} \text{ for all }\outc \in \{\out,\outb\}.
\end{equation}
Because $\supp_{\player}(\out^{\wt}) = \supp_\player(\out) \cup \supp_\player(\outb)$ and $\{\act_\player,\actb_\player \}\subseteq \supp_\player(\out) \cup \supp_\player(\outb)$, we have $\{\act_\player,\actb_\player \}\subseteq \supp_{\player}(\out^{\wt})$. Thus, applying Bayes rule, we obtain that $\out^{\wt}_{\act_\player} = \out^{\wt}_{\actb_\player}$ if and only if for every $\act_{-\player}\in \Act_{-\player}$ and $\paystate\in\Paystate$, one has
\[
\out^{\wt}(\act_\player)\out^{\wt}(\actb_\player,\act_{-\player},\paystate) - \out^{\wt}(\actb_\player)\out^{\wt}(\act_\player,\act_{-\player},\paystate)=0.
\]
Expanding the left hand side of the above equation by substituting in the definition of $\out^{\wt}$, rearranging terms as a polynomial in $\wt$, and using \eqref{eq:equal beliefs or zero probability}, delivers that the above display equation is equivalent to
\begin{equation*}
 (\wt - \wt^{2}) \Big[ \out(\act_\player)\outb(\actb_\player,\act_{-\player},\paystate)
+ \outb(\act_\player)\out(\actb_\player,\act_{-\player},\paystate) 
- \outb(\actb_\player)\out(\act_\player,\act_{-\player},\paystate)
- \out(\actb_\player)\outb(\act_\player,\act_{-\player},\paystate) \Big]=0.
\end{equation*}
Since $\wt\in (0,1)$, we get that $\out^{\wt}_{\act_\player} = \out^{\wt}_{\actb_\player}$ if and only if for every $\act_{-\player}\in \Act_{-\player}$ and $\paystate\in\Paystate$, one has
\begin{equation*}
\out(\act_\player)\outb(\actb_\player,\act_{-\player},\paystate)
+ \outb(\act_\player)\out(\actb_\player,\act_{-\player},\paystate) 
- \outb(\actb_\player)\out(\act_\player,\act_{-\player},\paystate)
- \out(\actb_\player)\outb(\act_\player,\act_{-\player},\paystate)=0.
\end{equation*}
Writing the above in vector notation delivers that $\out^{\wt}_{\act_\player} = \out^{\wt}_{\actb_\player}$ is equivalent to
\begin{equation}\label{eq:lem-1118}
 \out(\act_\player)\outb(\actb_\player)\outb_{\actb_\player}
+ \outb(\act_\player)\out(\actb_\player)\out_{\actb_\player}
- \outb(\actb_\player)\out(\act_\player)\out_{\act_\player}
- \out(\actb_\player)\outb(\act_\player)\outb_{\act_\player}=\mathbf{0}.
\end{equation}

We now divide the proof into cases. Consider first the case in which $\{\act_\player,\actb_\player\} \subseteq \supp_{\player}(\out) \cap \supp_{\player}(\outb)$. In this case,  \eqref{eq:equal beliefs or zero probability} implies $\out_{\act_\player}= \out_{\actb_\player}$ and $\outb_{\act_\player}= \outb_{\actb_\player}$, and so we get that 
\[
\begin{split}
\out(\act_\player)\outb(\actb_\player)\outb_{\actb_\player}
+ \outb(\act_\player)\out(\actb_\player)\out_{\actb_\player}
& - \outb(\actb_\player)\out(\act_\player)\out_{\act_\player}
- \out(\actb_\player)\outb(\act_\player)\outb_{\act_\player} =
\\
& = (\out(\act_\player)\outb(\actb_\player) - \out(\actb_\player)\outb(\act_\player) )
(\outb_{\act_\player}-\out_{\act_\player})\neq \mathbf{0},
\end{split}
\]
where the inequality follows from failure of Lemma~\ref{lem:equal beliefs at a cvx comb of outcomes w same beliefs}-\eqref{lem:equal beliefs at a cvx comb of outcomes w same beliefs: equal beliefs} and Lemma~\ref{lem:equal beliefs at a cvx comb of outcomes w same beliefs}-\eqref{lem:equal beliefs at a cvx comb of outcomes w same beliefs: equal ratios}. We conclude \eqref{eq:lem-1118} fails.

Consider now the case in which $\{\act_\player,\actb_\player\} \not\subseteq \supp_{\player}(\out) \cap \supp_{\player}(\outb)$. Because Lemma~\ref{lem:equal beliefs at a cvx comb of outcomes w same beliefs}-\eqref{lem:equal beliefs at a cvx comb of outcomes w same beliefs: equal ratios} fails, we can assume  $\out(\act_{\player}) = 0<\out(\actb_\player)$ without loss of generality. Since the lemma assume $\act_\player \in \supp_{\player}(\out)\cup \supp_{\player}(\outb)$, it follows $\outb(\act_\player)>0$. Therefore, we can use failure of Lemma~\ref{lem:equal beliefs at a cvx comb of outcomes w same beliefs}-\eqref{lem:equal beliefs at a cvx comb of outcomes w same beliefs: equal ratios} to deduce that $\out_{\actb_\player} \neq \outb_{\act_\player}$. Using these facts, we obtain that 
\[
\out(\act_\player)\outb(\actb_\player)\outb_{\actb_\player}
+ \outb(\act_\player)\out(\actb_\player)\out_{\actb_\player}
 - \outb(\actb_\player)\out(\act_\player)\out_{\act_\player}
- \out(\actb_\player)\outb(\act_\player)\outb_{\act_\player} =
\outb(\act_{\player})\out(\actb_{\player})(\out_{\actb_\player} - \outb_{\act_\player})\neq \mathbf{0}.
\]
It follows that \eqref{eq:lem-1118} fails.
\end{proof}

We are now ready to prove Proposition~\ref{pro:local equal-belief actions via extreme points}. The ``if'' portion is straightforward; the ``only if'' portion uses Lemma~\ref{lem:equal beliefs at a cvx comb of outcomes w same beliefs}.

\begin{proof}[Proof of Proposition~\ref{pro:local equal-belief actions via extreme points}]
We first prove the ``if'' portion. Let $\out\in\BCEset$ and $\act_\player,\actb_\player\in\supp_\player(\out)$. Let $t^1,\ldots,t^n>0$ and $\out^1,\ldots,\out^n\in\ext(\BCEset)$ such that 
\[
\out=\sum_{m=1}^n t^m \out^m. 
\]
Simple algebra shows that for all $\actc_\player\in \supp_\player(\out)$
\[
\out_{\actc_\player}=\sum_{m=1}^n \frac{t^m\out^m(\actc_\player)}{\sum_{l=1}^n t^l\out^l(\actc_\player)} \out^m_{\actc_\player}.
\]
If Proposition~\ref{pro:local equal-belief actions via extreme points}-\eqref{pro:local equal-belief actions via extreme points: equal beliefs} holds, then 
\begin{align*}
\out_{\act_\player}
 =\sum_{m=1}^n \frac{t^m\out^m(\act_\player)}{\sum_{l=1}^n t^l\out^l(\act_\player)} \out^m_{\act_\player}
 & =\sum_{m=1}^n \frac{t^m\out^m(\act_\player)}{\sum_{l=1}^n t^l\out^l(\act_\player)} \mu\\
& = \mu \\
& = \sum_{m=1}^n \frac{t^m\out^m(\actb_\player)}{\sum_{l=1}^n t^l\out^l(\actb_\player)} \mu
= \sum_{m=1}^n \frac{t^m\out^m(\actb_\player)}{\sum_{l=1}^n t^l\out^l(\actb_\player)} \out^m_{\actb_\player}
= \out_{\actb_\player}.
\end{align*}
Suppose now that Proposition~\ref{pro:local equal-belief actions via extreme points}-\eqref{pro:local equal-belief actions via extreme points: equal ratios} holds. For every $m$, $\out^m(\act_\player)\out^m_{\act_\player} = \const\out^m(\actb_\player)\out^m_{\actb_\player}$ implies $\out^m(\act_\player)=\const\out^m(\actb_\player)$ and $\out^m_{\act_\player}= \out^m_{\actb_\player}$. Thus,
\[
\out_{\act_\player}
 =\sum_{m=1}^n \frac{t^m\out^m(\act_\player)}{\sum_{l=1}^n t^l\out^l(\act_\player)} \out^m_{\act_\player}
 =\sum_{m=1}^n \frac{t^m\const \out^m(\actb_\player)}{\sum_{l=1}^n t^l\const\out^l(\actb_\player)} \out^m_{\actb_\player} = \sum_{m=1}^n \frac{t^m \out^m(\actb_\player)}{\sum_{l=1}^n t^l\out^l(\actb_\player)} \out^m_{\actb_\player}
= \out_{\actb_\player}.
\]
This concludes the proof of the proposition's ``if'' portion.

We now show the proposition's ``only if'' portion. We proceed by contradiction: we assume that Proposition~\ref{pro:local equal-belief actions via extreme points}-\eqref{pro:local equal-belief actions via extreme points: equal beliefs} and Proposition~\ref{pro:local equal-belief actions via extreme points}-\eqref{pro:local equal-belief actions via extreme points: equal ratios} both fail and show that there exists $\out\in \BCEset$ such that $\act_\player,\actb_\player \in \supp_\player(\out)$ and $\out_{\act_\player}\neq\out_{\actb_\player}$. As we are done if $\out_{\act_\player} \neq \out_{\actb_\player}$ for some $\out \in \ext(\BCEset)$ with $\act_\player,\actb_\player \in \supp_\player (\out)$, assume $\out_{\act_\player} = \out_{\actb_\player}$ holds for all such $\out$.

Since Proposition~\ref{pro:local equal-belief actions via extreme points}-\eqref{pro:local equal-belief actions via extreme points: equal beliefs} fails, and $\act_\player$ and $\actb_\player$ are $\BCEset$-coherent, there exist $\out,\outb\in\ext(\BCEset)$ such that $\out(\act_\player)>0$, $\outb(\actb_\player)>0$, and $\out_{\act_\player}\neq \outb_{\actb_\player}$. As we are done if $(0.5\out+0.5\outb)_{\act_\player}\neq(0.5\out+0.5\outb)_{\actb_\player}$, assume $(0.5\out+0.5\outb)_{\act_\player}=(0.5\out+0.5\outb)_{\actb_\player}$. Since $\out_{\act_\player}\neq \outb_{\actb_\player}$, Lemma~\ref{lem:equal beliefs at a cvx comb of outcomes w same beliefs}-\eqref{lem:equal beliefs at a cvx comb of outcomes w same beliefs: equal beliefs} fails. Thus, Lemma~\ref{lem:equal beliefs at a cvx comb of outcomes w same beliefs}-\eqref{lem:equal beliefs at a cvx comb of outcomes w same beliefs: equal ratios} must hold: there exist $\lambda>0$ such that $\out(\act_\player)= \const \out(\actb_\player)$ and  $\outb(\act_\player)= \const \outb(\actb_\player)$; in particular, $\out(\actb_\player)>0$ and $\outb(\act_\player)>0$.

Since Proposition~\ref{pro:local equal-belief actions via extreme points}-\eqref{pro:local equal-belief actions via extreme points: equal ratios} fails, there must exist $\outc\in \ext(\BCEset)$ such that $\outc(\act_\player) \neq \const \outc(\actb_\player)$; in particular, $\outc(\act_\player)>0$ or $\outc(\actb_\player)>0$. Let $\actc_\player\in \{\act_\player, \actb_\player\}$ such that $\outc(\actc_\player)>0$. Since $\out_{\act_\player}\neq \outb_{\actb_\player}$, either $\outc_{\actc_\player}\neq \out_{\act_\player}$, or $\outc_{\actc_\player}\neq \out_{\actb_\player}$, or both. Thus, by Lemma~\ref{lem:equal beliefs at a cvx comb of outcomes w same beliefs}, either $(0.5\out+0.5\outc)_{\act_\player}\neq(0.5\out+0.5\outc)_{\actc_\player}$, or $(0.5\outb+0.5\outc)_{\actb_\player}\neq(0.5\outb+0.5\outc)_{\actc_\player}$, or both. In any case, we have found $\out\in\BCEset$ such that $\act_\player,\actb_\player \in \supp_\player(\out)$ and $\out_{\act_\player}\neq\out_{\actb_\player}$. 
\end{proof}

\section{Vanishing Cost Equilibria}

\subsection{Proof of Theorem \ref{thm:complete_info_limit}}\label{sec:proof_vanishing_cost}
The ``only if'' side of the theorem follows from Theorem \ref{thm:mon_tech}. Next we prove the ``if'' side.

Let $\out$ be a complete-information Nash equilibrium. Following the notation in the main text, let $\alpha_{\paystate,\player}\in \Delta(\Act_\player)$ be the distribution of $\player$'s action given $\paystate$. 

Let $(\out^n)_{n\in \mathbb{N}}$ be a sequence of sBCE that converges to $\out$. For every player $\player$ and every $n\in \mathbb{N}$, define $\PartA_\player^{\out^n}$ as in Section \ref{subsec:embedding}. Thus, $\PartA_\player^{\out^n}$ is the partition of $\Act_\player$ such that $\act_\player$ and $\actb_\player$ are in the same cell if and only if either $\act_\player,\actb_\player\in\supp_\player(\out^n)$ and $\out^n_{\act_\player}=\out^n_{\actb_\player}$, or $\act_\player,\actb_\player\notin\supp_\player(\out^n)$. Without loss of generality, we assume that $\PartA_\player^{\out^n}=\PartA_\player^{\out^m}$ for all $\player\in\Player$ and $m,n\in\mathbb{N}$ (pass to a subsequence if necessary). To ease the exposition, we write $\PartA_\player$ instead of $\PartA_\player^{\out^n}$; we also write
\begin{align*}
\PartA_{-\player} = \left\{\prod_{\playerb\neq\player} B_\playerb : B_\playerb\in \PartA_\playerb\text{ for all } \playerb\neq\player\right\},\\
\PartA = \{B_\player\times B_{-\player}:B_\player\in\PartA_\player\text{ and }B_{-\player}\in\PartA_{-\player}\}.
\end{align*}

Next, we construct one canonical representation $(\mathcal{S}^n,\aplan^n)$ for each $\out^n$ (see Section \ref{subsec:embedding} for the general definition of canonical representation). The information structure $\mathcal{S}^n=(\Corstate,\corprior^n,(\Signal_\player,\exper_\player)_{\player\in\Player})$ is specified as follows:
\begin{itemize}
\item To construct $\Corstate\subseteq\prod_{\player\in\Player}\Delta(\PartA_\player)$, let $\delta_{B_\player}\in \Delta(\PartA_\player)$ be the Dirac measure concentrated on  $B_\player\in\PartA_\player$. Moreover, let $\bar{\mact}_{\theta,\player}$ be the measure on $\PartA_\player$ induced by $\alpha_{\paystate,\player}$: for $B_\player\in \Act_\player$,
\[
\bar{\alpha}_{\theta,\player}(B_\player)=\sum_{\act_\player\in B_\player}\mact_{\theta,\player}(\act_\player).
\]
We set $\Corstate=\prod_{\player}\Corstate_\player$ where for every player $\player$,
\[
\Corstate_\player=\left\{\delta_{B_\player}:B_\player\in\Act_\player\right\}\bigcup\left\{\bar{\alpha}_{\paystate,\player}:\paystate\in\Paystate\right\}.
\]
\item To construct $\corprior^n:\Paystate\rightarrow\Delta(\Corstate)$, first we take a sequence $(t^n)_{n\in \mathbb{N}}$ in $(0,1)$ such that (i) $t^n\rightarrow 1$, and (ii) for every $n$, $t^n \out < \out^n$ (such a sequence exists because $\out^n\rightarrow\out$). For every $n$, we define the outcome $r^n\in\Delta(\Act\times\Paystate)$ by
\[
r^n(\act,\paystate) = \frac{p^n(\act,\paystate)-t^n p(\act,\paystate)}{1-t^n},
\]
and the Markov kernel $\corprior^{r,n}:\Paystate\rightarrow\Delta(\Corstate)$ by
\[
\corprior^{r,n}(\corstate\vert\paystate)=\begin{cases}
\sum_{\act\in B}\frac{r^n(\act,\paystate)}{\payprior(\paystate)} &\text{if }B\in \PartA\text{ and }\corstate=\left(\delta_{B_\player}\right)_{\player\in\Player},\\
0 &\text{otherwise.}
\end{cases}
\]
We also denote by $\corprior^\out:\Paystate\rightarrow\Delta(\Corstate)$ the Markov kernel given by
\[
\corprior^{\out}(\corstate\vert\paystate)=\begin{cases}
1 &\text{if }\corstate=\left(\bar{\alpha}_{\paystate,\player}\right)_{\player\in\Player},\\
0 &\text{otherwise.}
\end{cases}
\]
Finally, we construct $\corprior^{n}:\Paystate\rightarrow\Delta(\Corstate)$ as follows:
\[
\corprior^{n}(\corstate\vert\paystate)=(1-t^n)\corprior^{r,n}(\corstate\vert\paystate)+t^n \corprior^{\out}(\corstate\vert\paystate).
\]

\item For every player $\player$, we take $\Signal_\player $ sufficiently rich so that $\PartA_\player\subseteq \Signal_\player$.

\item For every $\player\in \Player$, $\signal_\player\in \Signal_\player$, $\corstate\in \Corstate$, and $\theta\in\Theta$, we define the experiment $\exper_\player$ by
\begin{equation*}
\exper_\player(\signal_\player|\corstate,\paystate)=
\begin{cases}
\corstate_\player(\signal_\player) &\text{if } \signal_\player\in \PartA_\player,\\
0 &\text{otherwise}.
\end{cases}
\end{equation*}
\end{itemize}
The profile of action plans $\aplan^n=(\aplan_\player^n)_{\player\in\Player}$ is given by, for all $\player\in\Player$, $\act_\player\in\Act_\player$, and $\signal_\player\in\PartA_\player$,
\begin{equation*}
\aplan_\player^n(\act_\player|\signal_\player) = 
\begin{cases}
    \frac{\out^n(\act_\player)}{\sum_{\actb_\player \in \signal_\player}\out^n(\actb_\player)} & \text{if }\act_\player \in \signal_\player\text{ and }\signal_\player\subseteq \supp_\player(\out^n),\\
    \frac{1}{\vert\Act_\player\vert-\vert\supp_\player(\out^n)\vert} & \text{if }\act_\player \in \signal_\player\text{ and }\signal_\player=\Act_\player\setminus \supp_\player(\out^n),\\
    0 & \text{otherwise.}
\end{cases}
\end{equation*}

\begin{lemma}
The pair $(\mathcal{S}^n,\aplan^n)$ is a canonical representation of $\out^n$.
\end{lemma}

\begin{proof}
We need to verify that $\out^n$ is the measure induced by $(\mathcal{S}^n,\aplan^n)$ on $\Act\times\Paystate$. We will use the following claim:
\begin{claim}
For all $B\in\PartA$ and $\paystate\in\Paystate$,
\begin{equation}\label{eq:tart_cherry}
\sum_{\act\in B}\out^n(\act,\paystate)=\sum_{\corstate}\left[\prod_{\player}\corstate_\player(B_\player)\right]\corprior^n(\corstate\vert\paystate)\payprior(\paystate)
\end{equation}
\end{claim}
\begin{proof}[Proof of the claim]
By definition of $\corprior^n$, the right-hand side of (\ref{eq:tart_cherry}) is equal to
\begin{equation}\label{eq:kombu}
(1-t^n)\sum_{\corstate}\left[\prod_{\player}\corstate_\player(B_\player)\right]\corprior^{r,n}(\corstate\vert\paystate)\payprior(\paystate)+ t^n \sum_{\corstate}\left[\prod_{\player}\corstate_\player(B_\player)\right]\corprior^{\out}(\corstate\vert\paystate)\payprior(\paystate).
\end{equation}
We observe that, by definition of $\corprior^{r,n}$, 
\begin{equation}\label{eq:kombu1}
\sum_{\corstate}\left[\prod_{\player}\corstate_\player(B_\player)\right]\corprior^{r,n}(\corstate\vert\paystate)\payprior(\paystate)=\sum_{\act\in B}r^n(\act,\paystate). 
\end{equation}
Moreover, by definition of $\corprior^{\out}$,
\begin{equation}\label{eq:kombu2}
\sum_{\corstate}\left[\prod_{\player}\corstate_\player(B_\player)\right]\corprior^{\out}(\corstate\vert\paystate)\payprior(\paystate)=\left[\prod_{\player}\bar{\mact}_{\paystate,\player}(B_\player)\right]\payprior(\paystate)=\sum_{\act\in B}\out(\act,\paystate).
\end{equation}
Combining (\ref{eq:kombu})-(\ref{eq:kombu2}), we deduce that the right-hand side of (\ref{eq:tart_cherry}) is equal to 
\[
(1-t^n) \sum_{\act\in B}r^n(\act,\paystate)+t^n \sum_{\act \in B}\out(\act,\paystate)= \sum_{\act\in B}\left[(1-t^n)r^n(\act,\paystate)+t^n\out(\act,\paystate)\right].
\]
Since $\out^n=(1-t^n)r^n+t^n\out$, we conclude that (\ref{eq:tart_cherry}) holds.
\end{proof}

By Lemma \ref{lem:p-decomposition}, for all $\act\in \Act$ and $\theta\in\Theta$, we can decompose $\out^n(\act,\paystate)$ as
\begin{equation}\label{eq:mango}
\out^n(\act,\paystate)  =
\left[\prod_{\player}\aplan^n_\player(\act_\player|\PartA_\player(\act_\player))\right]\sum_{\actb\in \PartA(\act)}\out^n(\actb,\paystate),
\end{equation}
where $\PartA_\player(\act_\player)$ is the cell of the partition that contains $\act_\player$, and $\PartA(\act)=\prod_{\player}\PartA_i(\act_\player)$.
Putting together (\ref{eq:tart_cherry}) and (\ref{eq:mango}), we obtain that
\begin{align*}
\out^n(\act,\paystate) & = \left[\prod_{\player}\aplan^n_\player(\act_\player|\PartA_\player(\act_\player))\right]\sum_{\corstate}\left[\prod_{\player}\corstate_\player(\PartA_\player(\act_\player))\right]\corprior^n(\corstate\vert\paystate)\payprior(\paystate)\\
& = \sum_{\corstate}\left[\prod_{\player}\aplan^n_\player(\act_\player|\PartA_\player(\act_\player))\corstate_\player(\PartA_\player(\act_\player))\right]\corprior^n(\corstate\vert\paystate)\payprior(\paystate)\\
& = \sum_{\corstate,\signal}\left[\prod_{\player}\aplan^n_\player(\act_\player|\signal_\player)\corstate_\player(\signal_\player)\right]\corprior^n(\corstate\vert\paystate)\payprior(\paystate)\\
& = \sum_{\corstate,\signal}\left[\prod_{\player}\aplan^n_\player(\act_\player|\signal_\player)\exper_\player(\signal_\player\vert\corstate,\paystate)\right]\corprior^n(\corstate\vert\paystate)\payprior(\paystate)
\end{align*}
where the first equality follows from (\ref{eq:tart_cherry}) and (\ref{eq:mango}), the second equality is just algebra, the third equality holds because $\aplan^n(\act_\player\vert\signal_\player)>0$ if and only if $\act_\player\in\signal_\player$, and the last equality by definition of $\exper_\player$. We conclude that $\out^n$ is the measure induced by $(\mathcal{S}^n,\aplan^n)$ on $\Act\times\Paystate$.\end{proof}

Since $\out^n$ is a separated BCE and $(\mathcal{S}^n,\aplan^n)$ is a canonical representation of $\out^n$, it follows from Lemma \ref{lem:p-canonical} that for every player $\player$, there exists a monotone $\icost^n_\player:\Delta(\Signal_\player)^{\Corstate\times\Paystate}\rightarrow\mathbb{R}_+$ such that $(\exper,\aplan^n)$ is an equilibrium of $(\BGame,\IT^n)$ with $\IT^n=(\Corstate,\corprior^n,(\Signal_\player,\Delta(\Signal_\player)^{\Corstate\times\Paystate},\icost^n_\player)_{\player\in\Player})$. Moreover, we can choose $\icost_\player^n$ such that 
\[
\max_{\experb_\player\in \Delta(\Signal_\player)^{\Corstate\times\Paystate}} \icost_{\player}^n(\experb_\player) \leq \frac{1}{n} + \hat{v}_\player\left(\mathcal{S}^n,\aplan^n\right) - \left[\frac{n-1}{n}\grossval_\player\left(\out^n\right)+\frac{1}{n} \noinfoval_\player\left(\out^n\right)\right].
\]
As $n\rightarrow\infty$, the upper bound on $\player$'s costs converges to 
\begin{equation}\label{eq:bold}
\sum_{\paystate}\payprior(\paystate) \left[\max_{\act_\player\in\Act_\player}\sum_{B_{-\player}\in\PartA_{-i}}\sum_{\act_{-\player}\in B_{-\player}}\util_\player(\act_\player,\act_{-\player},\paystate)\prod_{\playerb\neq \player}\frac{\out(\act_\playerb)\bar{\alpha}_{\paystate,\playerb}(B_\playerb)}{\sum_{\actb_\playerb\in B_\playerb}\out(\actb_\playerb)}\right]-\grossval_\player\left(\out\right),
\end{equation}
where we adopt the convention that $\frac{0}{0}=0$. Hence, if we show that (\ref{eq:bold}) is equal to zero, we can conclude that $\out$ is a vanishing cost equilibrium, as desired.

To prove that (\ref{eq:bold}) is equal to zero, we need the following intermediate result:
\begin{lemma}\label{lem:flx_table}
For all $\paystate\in\paystate$, $\player\in\Player$, $B_\player\in \PartA_\player$, and $\act_\player\in B_\player$,
\[
\frac{\out(\act_\player)\bar{\alpha}_{\paystate,\player}(B_\player)}{\sum_{\actb_\player\in B_\player}\out(\actb_\player)}= \alpha_{\paystate,\player}(\act_\player).
\]
\end{lemma}
\begin{proof}
We divide the proof in three cases. Case (i): Assume $\sum_{\actb_\player\in B_\player}\out(\actb_\player)=0$. Then $\bar{\alpha}_{\paystate,\player}(B_\player)=0$ (because $\sum_{\actb_\player\in B_\player}\out(\actb_\player)=\sum_\paystate \payprior(\paystate)\bar{\alpha}_{\paystate,\player}(B_\player)$) and $\alpha_{\paystate,\player}(\act_\player)=0$ (because $\act_\player\in B_\player$). Thus,
\[
\frac{\out(\act_\player)\bar{\alpha}_{\paystate,\player}(B_\player)}{\sum_{\actb_\player\in B_\player}\out(\actb_\player)}= \frac{0}{0}=0=\alpha_{\paystate,\player}(\act_\player).
\]

Case (ii): Assume $\sum_{\actb_\player\in B_\player}\out(\actb_\player)>0$ and $\out(\act_\player)=0$. Then $\alpha_{\paystate,\player}(\act_\player)=0$ (because $\out(\act_\player)=\sum_\paystate\payprior(\paystate)\alpha_{\paystate,\player}(\act_\player)$). Thus,
\[
\frac{\out(\act_\player)\bar{\alpha}_{\paystate,\player}(B_\player)}{\sum_{\actb_\player\in B_\player}\out(\actb_\player)}=0=\alpha_{\paystate,\player}(\act_\player).
\]

Case (iii): Assume $\sum_{\actb_\player\in B_\player}\out(\actb_\player)>0$ and $\out(\act_\player)>0$. Take any $\actb_\player\in B_\player$ such that $\out(\actb_\player)>0$. Since $\out^n\rightarrow\out$, $\out^n(\act_\player)>0$ and $\out^n(\actb_\player)>0$ for all $n$ sufficiently large. Thus, by definition of $\PartA_\player$, $\out_{\act_\player}^n=\out_{\actb_\player}^n$ for all sufficiently large $n$ (in fact, for all $n$). We deduce that $\out_{\act_\player}=\out_{\actb_\player}$: for all $\act_{-\player}\in\Act_{-\player}$ and $\paystate\in\Paystate$,
\begin{equation}\label{eq:friday}
\frac{\out(\act_\player,\act_{-\player},\paystate)}{\out(\act_\player)}=\frac{\out(\actb_\player,\act_{-\player},\paystate)}{\out(\actb_\player)} 
\end{equation}
Since $\out$ is a complete-information Nash equilibrium, players' actions are conditionally independent given the payoff state. Thus, (\ref{eq:friday}) becomes
\[
\frac{\mact_{\paystate,\player}(\act_\player)}{\out(\act_\player)}=\frac{\mact_{\paystate,\player}(\actb_\player)}{\out(\actb_\player)},
\]
which is the same as 
\[
\frac{\mact_{\paystate,\player}(\act_\player)}{\out(\act_\player)}\out(\actb_\player)=\mact_{\paystate,\player}(\actb_\player).
\]
Clearly, the equality also holds for $\actb_\player\in B_\player$ such that $\out(\actb_\player)=0$. Thus, summing over all $\actb_\player\in B_\player$, we obtain that 
\[
\frac{\mact_{\paystate,\player}(\act_\player)}{\out(\act_\player)}\left(\sum_{\actb_\player\in B_\player}\out(\actb_\player)\right)=\overline{\mact}_{\paystate,\player}(B_\player).
\]
Rearranging the equality, we conclude that
\[
\frac{\out(\act_\player)\bar{\alpha}_{\paystate,\player}(B_\player)}{\sum_{\actb_\player\in B_\player}\out(\actb_\player)}= \alpha_{\paystate,\player}(\act_\player).
\]
\end{proof}

It follows from Lemma \ref{lem:flx_table} that (\ref{eq:bold}) is equal to 
\begin{equation*}
\sum_{\paystate}\payprior(\paystate) \left[\max_{\act_\player}\sum_{\act_{-\player}}\util_\player(\act_\player,\act_{-\player},\paystate)\prod_{\playerb\neq \player}\alpha_{\paystate,\playerb}(\act_\playerb)\right]-\grossval_\player\left(\out\right),
\end{equation*}
which, in turn, is equal to zero since $\out$ is complete-information Nash equilibrium:
\begin{align*}
\grossval_\player\left(\out\right)&=\sum_{\paystate}\payprior(\paystate)\left[\sum_{\act}\util_\player(\act,\paystate)\mact_{\paystate,\player}(\act_\player)\prod_{\playerb\neq\player}\mact_{\paystate,\playerb}(\act_\playerb)\right]\\
&=\sum_{\paystate}\payprior(\paystate) \left[\max_{\act_\player}\sum_{\act_{-\player}}\util_\player(\act_\player,\act_{-\player},\paystate)\prod_{\playerb\neq \player}\mact_{\paystate,\playerb}(\act_\playerb)\right],
\end{align*}
where the second equality holds because, given $\theta$, $\mact_{\paystate,\player}$ is a best response to $(\mact_{\paystate,\playerb})_{\playerb\neq\player}$. We conclude that $\out$ is a vanishing cost equilibrium.

\subsection{Full Characterization}\label{sec:proof_vanishing_cost_gen}

In this section we provide a characterization of \emph{all} vanishing cost equilibria. When information costs are negligible, both the payoff state and the correlation state become freely learnable. Thus, every vanishing cost equilibrium must be (i) a convex combination of complete-information Nash equilibria, and (ii) the limit of a sequence of separated BCEs (from Theorem \ref{thm:mon_tech}). It turns out that vanishing cost equilibria satisfy not only (i) and (ii), but also additional ``measurability'' conditions that we present next.

Fix a base game $\BGame$. For every player $\player$, let $\PartA_\player$ be a partition of $\Act_\player$. We denote by $B_\player$ a generic element of $\PartA_\player$. Define 
\begin{align*}
\PartA_{-\player} = \left\{\prod_{\playerb\neq\player} B_\playerb : B_\playerb\in \PartA_\playerb\text{ for all } \playerb\neq\player\right\},\\
\PartA = \{B_\player\times B_{-\player}:B_\player\in\PartA_\player\text{ and }B_{-\player}\in\PartA_{-\player}\}.
\end{align*}
We refer to $\PartA_{-\player}$ and $\PartA$ as \textbf{product partitions} of $\Act_{-\player}$ and $\Act$, respectively. 

Given a product partition $\PartA$ of $\Act$, we say an  outcome $\out$ is \textbf{$\PartA$-measurable} if for every $\player\in\Player$, $B_\player\in\PartA_\player$, and $\act_\player,\actb_\player\in B_\player\cap \supp_\player(\out)$, 
\[
\out_{\act_\player}=\out_{\actb_\player}.
\]
For an intuition, take the perspective of a mediator who wants to implement an outcome $\out$. If $\out$ is $\PartA$-measurable, then the mediator can implement it as follows: draw a payoff state $\paystate$  and an element $B$ of $\PartA$ with probability 
$
\sum_{\act\in B}\out(\theta,\act)
$; communicate to each player $\player$ the realized $B_\player$, and let them privately draw an action $\act_\player\in B_\player$ with probability $\out(\act_\player)/\sum_{\actb_\player\in B_i}\out(\actb_\player)$.

As in Section \ref{subsec:embedding}, let $\PartA^{\out}$ be the product partition of $\Act$ such that, for every player $\player$, actions $\act_\player$ and $\actb_\player$ are in the same cell if and only if either $\act_\player,\actb_\player\in\supp_\player(\out)$ and $\out_{\act_\player}=\out_{\actb_\player}$, or $\act_\player,\actb_\player\notin\supp_\player(\out)$. Note that $\out$ is measurable with respect to $\PartA_{\out}$; except for zero-probability actions, $\PartA_{\out}$ is the coarsest product partition for which $\out$ is measurable. 

Given an outcome $\out$ and a product partition $\PartA$ of $\Act$, we say an outcome $\outb$ is \textbf{$(\PartA,\out)$-decomposable}  if $\outb$ is $\PartA$-measurable, and for every $\player\in\Player$, $B_\player\in\PartA_\player$, and $\act_\player,\actb_\player\in B_\player$,
\begin{equation}\label{eq:decomposable}
\outb(\act_\player) \out(\actb_\player) = \out(\act_\player) \outb(\actb_\player).
\end{equation}
For an intuition, take the perspective of a mediator who wants to implement an outcome $\outb$. If $\outb$ is $(\PartA,\out)$-decomposable, then the mediator can implement it as follows: draw a payoff state $\paystate$  and an element $B$ of $\PartA$ with probability 
$
\sum_{\act\in B}\outb(\theta,\act)
$; communicate to each player $\player$ the realized $B_\player$, and let them privately draw an action $\act_\player\in B_\player$ with probability $\out(\act_\player)/\sum_{\actb_\player\in B_\player}\out(\actb_\player)=\outb(\act_\player)/\sum_{\actb_\player\in B_\player}\outb(\actb_\player)$.

Next we use $\PartA$-measurability and $(\PartA,\out)$-decomposability to characterize vanishing cost equilibrium:

\begin{theorem}\label{thm:vanishing_costs_all}
An outcome $\out$ is a vanishing cost equilibrium if and only if there exists a product partition $\PartA$ of $\Act$ such that
\begin{itemize}
\item[(i)] $\out$ is a convex combination of finitely many $(\PartA,\out)$-decomposable complete-information Nash equilibria, and 
\item[(ii)] $\out$ is the limit of a sequence $(\out^n)_{n=1}^\infty$ of separated BCEs, with $\PartA^{\out^n}=\PartA$ for all $n$.
\end{itemize}
\end{theorem}

Theorem \ref{thm:vanishing_costs_all} generalizes Theorem \ref{thm:complete_info_limit}. To see the relationship between the two results, let $\out$ be a complete-information Nash equilibrium that is the limit of a sequence  $(\out^n)_{n=1}^\infty$ of separated BCEs, as in Theorem \ref{thm:complete_info_limit}. Passing to a subsequence, we can assume that $\PartA^{\out^n}=\PartA^{\out^1}$ for all $n$. Since $\out^n\rightarrow \out$, $\out$ is $\PartA^1$-measurable. Because (\ref{eq:decomposable}) trivially holds for $\outb=\out$, $\out$ is $(\PartA^1,\out)$-decomposable. Thus, $\out$ satisfies the hypotheses of Theorem \ref{thm:vanishing_costs_all} with $\PartA:=\PartA^1$. In particular, one can obtain Theorem \ref{thm:complete_info_limit} as a corollary of Theorem \ref{thm:vanishing_costs_all}.

We divide the proof of the theorem \ref{thm:vanishing_costs_all} in two parts.

\subsection{Proof of the ``if'' side of Theorem \ref{thm:vanishing_costs_all}}

Let $\PartA$ be a product partition of $\Act$, and let $\{\outb^1,\ldots,\outb^L\}$ be a finite set of complete-information Nash equilibria. For every $l=1,\ldots,L$, we denote by $\alpha^l_{\paystate,\player}\in \Delta(\Act_\player)$ the conditional distribution of $\player$'s action given $\paystate$. Let $\out$ be an outcome in the convex hull of  $\{\outb^1,\ldots,\outb^L\}$: 
\[
\out=\sum_{l=1}^L s^l \outb^l.
\]
Without loss of generality, suppose that $s^l>0$ for all $l=1,\ldots,L$. 

Assume that for all $l=1,\ldots,L$, $\outb^l$ is $(\PartA,\out)$-decomposable. Furthermore, assume that  $\out$ is the limit of a sequence $(\out^n)_{n=1}^\infty$ of separated BCEs such that $\PartA^{\out^n}=\PartA$ for all $n$. We want to prove that $\out$ is a vanishing cost equilibrium. 

We begin with constructing one canonical representation $(\mathcal{S}^n,\aplan^n)$ for each $\out^n$ (see Section \ref{subsec:embedding} for the general definition of canonical representation). The information structure 
$\mathcal{S}^n=(\Corstate,\corprior^n,(\Signal_\player,\exper_\player)_{\player\in\Player})$
is specified as follows:
\begin{itemize}
\item To construct $\Corstate\subseteq\prod_{\player\in\Player}\Delta(\PartA_\player)$, let $\delta_{B_\player}\in \Delta(\PartA_\player)$ be the Dirac measure concentrated on  $B_\player\in\PartA_\player$. Moreover, let $\bar{\mact}_{\theta,\player}^l$ be the measure on $\PartA_\player$ induced by $\alpha_{\paystate,\player}^l$: for $B_\player\in \Act_\player$,
\[
\bar{\mact}_{\theta,\player}^l(B_\player)=\sum_{\act_\player\in B_\player}\mact_{\theta,\player}^l(\act_\player).
\]
We set $\Corstate=\prod_{\player}\Corstate_\player$ where for every player $\player$,
\[
\Corstate_\player=\left\{\delta_{B_\player}:B_\player\in\Act_\player\right\}\bigcup\left\{\bar{\mact}^l_{\paystate,\player}:\paystate\in\Paystate\right\}.
\]

\item To construct $\corprior^n:\Paystate\rightarrow\Delta(\Corstate)$, first we take a sequence $(t^n)_{n=1}^\infty$ in $(0,1)$ such that 
\begin{itemize}
\item $t^n\rightarrow 1$, and 
\item $t^n \out < \out^n$ for all $n$. 
\end{itemize}
Such a sequence exists because $\out^n\rightarrow\out$. For each $n$, we define the outcome $r^n\in\Delta(\Act\times\Paystate)$ by
\[
r^n(\act,\paystate) = \frac{p^n(\act,\paystate)-t^n p(\act,\paystate)}{1-t^n},
\]
and the Markov kernel $\corprior^{r,n}:\Paystate\rightarrow\Delta(\Corstate)$ by
\[
\corprior^{r,n}(\corstate\vert\paystate)=\begin{cases}
\sum_{\act\in B}\frac{r^n(\act,\paystate)}{\payprior(\paystate)} &\text{if }B\in \PartA\text{ and }\corstate=\left(\delta_{B_\player}\right)_{\player\in\Player},\\
0 &\text{otherwise.}
\end{cases}
\]
We also denote by $\corprior^\out:\Paystate\rightarrow\Delta(\Corstate)$ the Markov kernel given by
\[
\corprior^{\out}(\corstate\vert\paystate)=\begin{cases}
s^l &\text{if }l=1,\ldots,L\text{ and }\corstate=\left(\bar{\alpha}^l_{\paystate,\player}\right)_{\player\in\Player},\\
0 &\text{otherwise.}
\end{cases}
\]
Finally, we construct $\corprior^{n}:\Paystate\rightarrow\Delta(\Corstate)$ as follows:
\[
\corprior^{n}(\corstate\vert\paystate)=(1-t^n)\corprior^{r,n}(\corstate\vert\paystate)+t^n \corprior^{\out}(\corstate\vert\paystate).
\]

\item For every player $\player$, we take $\Signal_\player $ sufficiently rich so that $\PartA_\player\subseteq \Signal_\player$.

\item For every $\player\in \Player$, $\signal_\player\in \Signal_\player$, $\corstate\in \Corstate$, and $\theta\in\Theta$, we define the experiment $\exper_\player$ by
\begin{equation*}
\exper_\player(\signal_\player|\corstate,\paystate)=
\begin{cases}
\corstate_\player(\signal_\player) &\text{if } \signal_\player\in \PartA_\player,\\
0 &\text{otherwise}.
\end{cases}
\end{equation*}
\end{itemize}
The profile of action plans $\aplan^n=(\aplan_\player^n)_{\player\in\Player}$ is given by, for all $\player\in\Player$, $\act_\player\in\Act_\player$, and $\signal_\player\in\PartA_\player$,
\begin{equation*}
\aplan_\player^n(\act_\player|\signal_\player) = 
\begin{cases}
    \frac{\out^n(\act_\player)}{\sum_{\actb_\player \in \signal_\player}\out^n(\actb_\player)} & \text{if }\act_\player \in \signal_\player\text{ and }\signal_\player\subseteq \supp_\player(\out^n),\\
    \frac{1}{\vert\Act_\player\vert-\vert\supp_\player(\out^n)\vert} & \text{if }\act_\player \in \signal_\player\text{ and }\signal_\player=\Act_\player\setminus \supp_\player(\out^n),\\
    0 & \text{otherwise.}
\end{cases}
\end{equation*}

\begin{lemma}
The pair $(\mathcal{S}^n,\aplan^n)$ is a canonical representation of $\out^n$.
\end{lemma}
\begin{proof}
We need to verify that $\out^n$ is the measure induced by $(\mathcal{S}^n,\aplan^n)$ on $\Act\times\Paystate$. We will use the following claim:
\begin{claim}
For all $B\in\PartA$ and $\paystate\in\Paystate$,
\begin{equation}\label{eq:tart_cherry_oa}
\sum_{\act\in B}\out^n(\act,\paystate)=\sum_{\corstate}\left[\prod_{\player}\corstate_\player(B_\player)\right]\corprior^n(\corstate\vert\paystate)\payprior(\paystate)
\end{equation}
\end{claim}
\begin{proof}[Proof of the claim]
By definition of $\corprior^n$, the right-hand side of (\ref{eq:tart_cherry_oa}) is equal to
\begin{equation}\label{eq:kombu_oa}
(1-t^n)\sum_{\corstate}\left[\prod_{\player}\corstate_\player(B_\player)\right]\corprior^{r,n}(\corstate\vert\paystate)\payprior(\paystate)+ t^n \sum_{\corstate}\left[\prod_{\player}\corstate_\player(B_\player)\right]\corprior^{\out}(\corstate\vert\paystate)\payprior(\paystate).
\end{equation}
We observe that, by definition of $\corprior^{r,n}$, 
\begin{equation}\label{eq:kombu1_oa}
\sum_{\corstate}\left[\prod_{\player}\corstate_\player(B_\player)\right]\corprior^{r,n}(\corstate\vert\paystate)\payprior(\paystate)=\sum_{\act\in B}r^n(\act,\paystate). 
\end{equation}
Moreover, by definition of $\corprior^{\out}$,
\begin{align}\label{eq:kombu2_oa}
\sum_{\corstate}\left[\prod_{\player}\corstate_\player(B_\player)\right]\corprior^{\out}(\corstate\vert\paystate)\payprior(\paystate) & =\sum_{l}s^l\left[\prod_{\player}\bar{\mact}^l_{\paystate,\player}(B_\player)\right] \payprior(\paystate)\notag\\
& = \sum_{l}s^l\left[\sum_{\act\in B}\frac{\outb^l(\act,\paystate)}{\payprior(\paystate)}\right]\payprior(\paystate) =\sum_{\act\in B}\out(\act,\paystate).
\end{align}
Combining (\ref{eq:kombu_oa})-(\ref{eq:kombu2_oa}), we deduce that the right-hand side of (\ref{eq:tart_cherry_oa}) is equal to 
\[
(1-t^n) \sum_{\act\in B}r^n(\act,\paystate)+t^n \sum_{\act \in B}\out(\act,\paystate)= \sum_{\act\in B}\left[(1-t^n)r^n(\act,\paystate)+t^n\out(\act,\paystate)\right].
\]
Since $\out^n=(1-t^n)r^n+t^n\out$, we conclude that (\ref{eq:tart_cherry_oa}) holds.
\end{proof}
By Lemma \ref{lem:p-decomposition}, for all $\act\in \Act$ and $\theta\in\Theta$, we can decompose $\out^n(\act,\paystate)$ as
\begin{equation}\label{eq:mango_oa}
\out^n(\act,\paystate)  =
\left[\prod_{\player}\aplan^n_\player(\act_\player|\PartA_\player(\act_\player))\right]\sum_{\actb\in \PartA(\act)}\out^n(\actb,\paystate),
\end{equation}
where $\PartA_\player(\act_\player)$ is the cell of the partition that contains $\act_\player$, and $\PartA(\act)=\prod_{\player}\PartA_i(\act_\player)$.
Putting together (\ref{eq:tart_cherry_oa}) and (\ref{eq:mango_oa}), we obtain that
\begin{align*}
\out^n(\act,\paystate) & = \left[\prod_{\player}\aplan^n_\player(\act_\player|\PartA_\player(\act_\player))\right]\sum_{\corstate}\left[\prod_{\player}\corstate_\player(\PartA_\player(\act_\player))\right]\corprior^n(\corstate\vert\paystate)\payprior(\paystate)\\
& = \sum_{\corstate}\left[\prod_{\player}\aplan^n_\player(\act_\player|\PartA_\player(\act_\player))\corstate_\player(\PartA_\player(\act_\player))\right]\corprior^n(\corstate\vert\paystate)\payprior(\paystate)\\
& = \sum_{\corstate,\signal}\left[\prod_{\player}\aplan^n_\player(\act_\player|\signal_\player)\corstate_\player(\signal_\player)\right]\corprior^n(\corstate\vert\paystate)\payprior(\paystate)\\
& = \sum_{\corstate,\signal}\left[\prod_{\player}\aplan^n_\player(\act_\player|\signal_\player)\exper_\player(\signal_\player\vert\corstate,\paystate)\right]\corprior^n(\corstate\vert\paystate)\payprior(\paystate)
\end{align*}
where the first equality follows from (\ref{eq:tart_cherry_oa}) and (\ref{eq:mango_oa}), the second equality is just algebra, the third equality holds because $\aplan^n(\act_\player\vert\signal_\player)>0$ if and only if $\act_\player\in\signal_\player$, and the last equality by definition of $\exper_\player$. We conclude that $\out^n$ is the measure induced by $(\mathcal{S}^n,\aplan^n)$ on $\Act\times\Paystate$.\end{proof}

Since $\out^n$ is a separated BCE and $(\mathcal{S}^n,\aplan^n)$ is a canonical representation of $\out^n$, it follows from Lemma \ref{lem:p-canonical} that for every player $\player$, there exists a monotone $\icost^n_\player:\Delta(\Signal_\player)^{\Corstate\times\Paystate}\rightarrow\mathbb{R}_+$ such that $(\exper,\aplan^n)$ is an equilibrium of $(\BGame,\IT^n)$ with $\IT^n=(\Corstate,\corprior^n,(\Signal_\player,\Delta(\Signal_\player)^{\Corstate\times\Paystate},\icost^n_\player)_{\player\in\Player})$. Moreover, we can choose $\icost_\player^n$ such that 
\[
\max_{\experb_\player\in \Delta(\Signal_\player)^{\Corstate\times\Paystate}} \icost_{\player}^n(\experb_\player) \leq \frac{1}{n} + \hat{v}_\player\left(\mathcal{S}^n,\aplan^n\right) - \left[\frac{n-1}{n}\grossval_\player\left(\out^n\right)+\frac{1}{n} \noinfoval_\player\left(\out^n\right)\right].
\]
As $n\rightarrow\infty$, the upper bound on $\player$'s costs converges to 
\begin{equation}\label{eq:bold_oa}
\sum_{\paystate}\payprior(\paystate)\sum_l s^l\left[\max_{\act_\player\in\Act_\player}\sum_{B_{-\player}\in\PartA_{-i}}\sum_{\act_{-\player}\in B_{-\player}}\util_\player(\act_\player,\act_{-\player},\paystate)\prod_{\playerb\neq \player}\frac{\out(\act_\playerb)\bar{\alpha}^l_{\paystate,\playerb}(B_\playerb)}{\sum_{\actb_\playerb\in B_\playerb}\out(\actb_\playerb)}\right]-\grossval_\player\left(\out\right),
\end{equation}
where we adopt the convention that $\frac{0}{0}=0$. Hence, if we show that (\ref{eq:bold_oa}) is equal to zero, we can conclude that $\out$ is a vanishing cost equilibrium, as desired.

To prove that (\ref{eq:bold_oa}) is equal to zero, we need the following intermediate results:
\begin{lemma}\label{lem:flx_table_oa}
For all $\paystate\in\paystate$,  $l\in\{1,\ldots,L\}$, $\player\in\Player$, $B_\player\in \PartA_\player$, and $\act_\player\in B_\player$,
\[
\frac{\out(\act_\player)\bar{\alpha}^l_{\paystate,\player}(B_\player)}{\sum_{\actb_\player\in B_\player}\out(\actb_\player)}= \frac{\outb^l(\act_\player)\bar{\alpha}^l_{\paystate,\player}(B_\player)}{\sum_{\actb_\player\in B_\player}\outb^l(\actb_\player)}
\]
where on both sides of the equation we adopt the convention that $\frac{0}{0}=0$.
\end{lemma}
\begin{proof}
If $\bar{\alpha}^l_{\paystate,\player}(B_\player)=0$, then (trivially) 
\[
\frac{\out(\act_\player)\bar{\alpha}^l_{\paystate,\player}(B_\player)}{\sum_{\actb_\player\in B_\player}\out(\actb_\player)}= 0=\frac{\outb^l(\act_\player)\bar{\alpha}^l_{\paystate,\player}(B_\player)}{\sum_{\actb_\player\in B_\player}\outb^l(\actb_\player)}
\]
Suppose now that $\bar{\alpha}^l_{\paystate,\player}(B_\player)>0$. Then, we have
\[
\sum_{\actb_\player\in B_\player}\outb^l(\actb_\player)=\sum_{\paystate}\payprior(\paystate)\bar{\alpha}^l_{\paystate,\player}(B_\player)>0,
\]
which in turn implies that $\sum_{\actb_\player\in B_\player}\out(\actb_\player)>0$ (since $\outb^l$ is absolutely continuous with respect to $\out$). Since $\outb^l$ is $(\PartA,\out)$-decomposable, It follows from (\ref{eq:decomposable}) that 
\[
\frac{\out(\act_\player)}{\sum_{\actb_\player\in B_\player}\out(\actb_\player)}=\frac{\outb^l(\act_\player)}{\sum_{\actb_\player\in B_\player}\outb^l(\actb_\player)}.
\]
Multiplying both sides of the equation by $\bar{\alpha}^l_{\paystate,\player}(B_\player)$, we obtain the desired result.
\end{proof}
\begin{lemma}\label{lem:bx_table_oa}
For all $\paystate\in\paystate$,  $l\in\{1,\ldots,L\}$, $\player\in\Player$, $B_\player\in \PartA_\player$, and $\act_\player\in B_\player$,
\[
\frac{\outb^l(\act_\player)\bar{\alpha}^l_{\paystate,\player}(B_\player)}{\sum_{\actb_\player\in B_\player}\outb^l(\actb_\player)}= \alpha^l_{\paystate,\player}(\act_\player).
\]
\end{lemma}
\begin{proof}

We divide the proof in three cases. Case (i): Assume $\sum_{\actb_\player\in B_\player}\outb^l(\actb_\player)=0$. Then $\bar{\alpha}^l_{\paystate,\player}(B_\player)=0$ and $\alpha^l_{\paystate,\player}(\act_\player)=0$. Thus,
\[
\frac{\outb^l(\act_\player)\bar{\alpha}^l_{\paystate,\player}(B_\player)}{\sum_{\actb_\player\in B_\player}\outb^l(\actb_\player)}= \frac{0}{0}=0=\alpha^l_{\paystate,\player}(\act_\player).
\]

Case (ii): Assume $\sum_{\actb_\player\in B_\player}\outb^l(\actb_\player)>0$ and $\outb^l(\act_\player)=0$. Then $\alpha^l_{\paystate,\player}(\act_\player)=0$. Thus,
\[
\frac{\outb^l(\act_\player)\bar{\alpha}^l_{\paystate,\player}(B_\player)}{\sum_{\actb_\player\in B_\player}\outb^l(\actb_\player)}=0=\alpha^l_{\paystate,\player}(\act_\player).
\]

Case (iii): Assume $\sum_{\actb_\player\in B_\player}\outb^l(\actb_\player)>0$ and $\outb^l(\act_\player)>0$. Take any $\actb_\player\in B_\player$ such that $\outb^l(\actb_\player)>0$. Since $\outb^l$ is $\PartA$-measurable, $\outb^l_{\act_\player}=\outb^l_{\actb_\player}$: for all $\act_{-\player}\in\Act_{-\player}$ and $\paystate\in\Paystate$,
\begin{equation}\label{eq:friday_oa}
\frac{\outb^l(\act_\player,\act_{-\player},\paystate)}{\outb^l(\act_\player)}=\frac{\outb^l(\actb_\player,\act_{-\player},\paystate)}{\outb^l(\actb_\player)} 
\end{equation}
Since $\outb^l$ is a complete-information Nash equilibrium, players' actions are conditionally independent given the payoff state. Thus, (\ref{eq:friday_oa}) becomes
\[
\frac{\mact^l_{\paystate,\player}(\act_\player)}{\outb^l(\act_\player)}=\frac{\mact^l_{\paystate,\player}(\actb_\player)}{\outb^l(\actb_\player)},
\]
which is the same as 
\[
\frac{\mact^l_{\paystate,\player}(\act_\player)}{\outb^l(\act_\player)}\outb^l(\actb_\player)=\mact^l_{\paystate,\player}(\actb_\player).
\]
Clearly, the equality also holds for $\actb_\player\in B_\player$ such that $\outb^l(\actb_\player)=0$. Thus, summing over all $\actb_\player\in B_\player$, we obtain that 
\[
\frac{\mact^l_{\paystate,\player}(\act_\player)}{\outb^l(\act_\player)}\left(\sum_{\actb_\player\in B_\player}\outb^l(\actb_\player)\right)=\overline{\mact}^l_{\paystate,\player}(B_\player).
\]
Rearranging the equality, we conclude that
\[
\frac{\outb^l(\act_\player)\bar{\alpha}^l_{\paystate,\player}(B_\player)}{\sum_{\actb_\player\in B_\player}\outb^l(\actb_\player)}= \alpha^l_{\paystate,\player}(\act_\player).
\]

\end{proof}

It follows from Lemmas \ref{lem:flx_table_oa} and \ref{lem:bx_table_oa} that (\ref{eq:bold_oa}) is equal to 
\begin{equation*}
\sum_{\paystate}\payprior(\paystate) \sum_l s^l\left[\max_{\act_\player}\sum_{\act_{-\player}}\util_\player(\act_\player,\act_{-\player},\paystate)\prod_{\playerb\neq \player}\mact^l_{\paystate,\playerb}(\act_\playerb)\right]-\grossval_\player\left(\out\right),
\end{equation*}
which, in turn, is equal to zero since each $\outb^l$ is complete-information Nash equilibrium:
\begin{align*}
\grossval_\player\left(\out\right)&=\sum_{\paystate}\payprior(\paystate)\sum_l s^l\left[\sum_{\act}\util_\player(\act,\paystate)\mact^l_{\paystate,\player}(\act_\player)\prod_{\playerb\neq\player}\mact^l_{\paystate,\playerb}(\act_\playerb)\right]\\
&=\sum_{\paystate}\payprior(\paystate) \sum_l s^l\left[\max_{\act_\player}\sum_{\act_{-\player}}\util_\player(\act_\player,\act_{-\player},\paystate)\prod_{\playerb\neq \player}\mact^l_{\paystate,\playerb}(\act_\playerb)\right],
\end{align*}
where the second equality holds because, given $\theta$, $\mact^l_{\paystate,\player}$ is a best response to $(\mact^l_{\paystate,\playerb})_{\playerb\neq\player}$. We conclude that $\out$ is a vanishing cost equilibrium.

\subsection{Proof of the ``only if'' side of Theorem \ref{thm:vanishing_costs_all}}

Let $\out$ be a vanishing cost equilibrium: we want to show that there exists a product partition $\PartA$ of $\Act$ such that Theorem \ref{thm:vanishing_costs_all}-(i) and Theorem \ref{thm:vanishing_costs_all}-(ii) hold. 

By the definition of vanishing cost equilibrium, for every $n\in\{1,2,\ldots\}$, we can find an unconstrained rational-inattention technology $\itech^n=(\Corstate^n,\corprior^n,(\Signal_\player^n,\Exper^n_\player,\icost_\player^n)_{\player\in\Player})$ and an equilibrium $(\exper^n,\aplan^n)$ of $(\BGame,\itech^n)$ such that, denoting by $\out^n$ the outcome of $(\exper^n,\aplan^n)$, 
\begin{equation*}
\max_{\exper_\player^\prime\in \Delta(\Signal_\player^n)^{\Corstate^n\times\Paystate}} \icost_\player^n(\exper_\player^\prime) \leq \frac{1}{n}
\quad \text{and} \quad 
\left\vert \out(\act,\paystate) - \out^n(\act,\paystate) \right\vert  
\leq \frac{1}{n}
\end{equation*}
for all $\player \in \Player$, $\act \in \Act$, and $\paystate \in \Paystate$. Possibly passing to a subsequence, we can assume that
\[
\PartA^{\out^{n}}=\PartA^{\out^{n+1}}
\]
for all $n$. Let $\PartA:=\PartA^{\out^{n}}$ be this fixed product partition. Since each $\out^n$ is a sBCE by Theorem~\ref{thm:mon_tech}, we obtain that Theorem \ref{thm:vanishing_costs_all}-(ii) holds.

In the rest of the proof, we show that Theorem \ref{thm:vanishing_costs_all}-(i) is satisfied. We begin with introducing some notation. Let $\Mact_\player$ be the set of player $\player$'s mixed actions: $\Mact_\player=\Delta(\Act_\player)$; we define  the Cartesian products $\Mact_{-\player}=\prod_{\playerb\neq \player}\Mact_\playerb$ and $\Mact=\Mact_{-\player}\times\Mact_\player$. For  $n\in\mathbb{N}$, $\corstate\in \Corstate^n$, and $\paystate \in \Paystate$, we denote by $\mact^{n,\corstate,\paystate}\in \Mact$ the induced mixed-action profile: for all $\player$ and $\act_\player$,
\[
\mact^{n,\corstate,\paystate}_{\player}(\act_\player) = \sum_{\signal_\player} \aplan^{n}_{\player} (\act_\player|\signal_\player) \exper^{n}_{\player}(\signal_\player|\corstate,\paystate).
\]
In addition, we define the transition kernel $\Paystate\ni\paystate\mapsto \cord^{n}_\paystate\in\Delta\left(\Mact\right)$ by
\[
\cord^{n}_\paystate\left(M\right) = 
\corprior^{n}\left( 
\left\{ 
\corstate \in \Corstate^{n} : \mact^{n,\corstate,\paystate} \in M \right\}
|\paystate\right)
\]
for all Borel sets $M \subseteq \Mact$.
Direct computation shows that 
\begin{equation}\label{eq:may_9_23}
\out^{n}(\act,\paystate) = \payprior(\paystate) \int_{\Mact} \prod_\player \mact_\player(\act_\player) \,\dd\cord^{n}_\paystate(\mact)
\end{equation}
for all $\act$ and $\paystate$. By weak* compactness of $\Delta\left(\Mact\right)$ (\citealp{Aliprantis2006}, Theorem 15.11), it is without loss (potentially by passing to a subsequence) to assume that, for every $\paystate$, $\cord^{n}_\paystate$ converges in the weak* topology to some limit $\cord_\paystate$. Since $\out^n\rightarrow\out$, we obtain that
\begin{equation}\label{eq:barycenter}
\out (\act,\paystate) = \payprior(\paystate) \int_{\Mact} \prod_\player \mact_\player(\act_\player) \,\dd\cord_\paystate(\mact)
\end{equation}
for all $\act$ and $\paystate$. 

The next lemma relates the support of $\cord^n_\paystate$ (which is finite) to $\out^n$:

\begin{lemma}\label{lem:May_10_30}
For all $n\in\mathbb{N}$, $\paystate\in\Paystate$, $\mact \in \supp (\cord^n_\paystate)$, $\player\in\Player$, $B_\player\in\PartA_\player$, and $\act_\player,\actb_\player\in B_\player$,
\begin{equation}\label{eq:may_9_23_59}
\mact_\player(\act_\player)\out^n(\actb_\player)=\mact_\player(\actb_\player)\out^n(\act_\player).
\end{equation}
\end{lemma}

\begin{proof}
Since $\mact \in \supp (\cord^n_\paystate)$, $\mact$ is absolutely continuous with respect to $\out^n$---see (\ref{eq:may_9_23}). Thus, (\ref{eq:may_9_23_59}) trivially holds if $\out^n(\act_\player)=0$ or $\out^n(\actb_\player)=0$. 

Suppose now that $\out^n(\act_\player)>0$ and $\out^n(\actb_\player)>0$. Paralleling the notation in Section \ref{subsec:embedding},  we denote by $\nu^n\in\Delta(\Act\times\Signal^n\times\Corstate^n\times\Paystate)$ the probability measure over actions, signals, and states induced by the information structure $(\Corstate^n,\corprior^n,(\Signal^n_\player,\exper^n_\player)_{\player\in\Player})$ and the profile of action plans $\aplan^n$. Let $\Signal^n_{\act_\player}$ and $\Signal^n_{\actb_\player}$ be the set of positive-probability signals that make player $\player$ take actions $\act_\player$ and $\actb_\player$:
\begin{align*}
\Signal^n_{\act_\player} & =\{\signal_\player:\nu^n(\signal_\player)>0\text{ and }\aplan^n_\player(\act_\player|\signal_\player)>0\},\\
\Signal^n_{\actb_\player} & =\{\signal_\player:\nu^n(\signal_\player)>0\text{ and }\aplan^n_\player(\actb_\player|\signal_\player)>0\}.
\end{align*}
By Lemma \ref{lem:BR in info-game contains BR from outcome}, for all $\signal_\player\in\Signal^n_{\act_\player}$ and $\signal_\player^\prime\in\Signal^n_{\actb_\player}$,
\[
\BR(\out_{\act_\player})\subseteq \BR(\nu^n_{\signal_\player})\quad\text{and}\quad\BR(\out_{\actb_\player})\subseteq \BR(\nu^n_{\signalb_\player}),
\]
where $\nu^n_{\signal_\player}\in \Delta(\Act_{-\player}\times\Signal^n_{-\player}\times\Corstate^n\times\Paystate)$ and $\nu^n_{\signalb_\player}\in \Delta(\Act_{-\player}\times\Signal^n_{-\player}\times\Corstate^n\times\Paystate)$ are the posterior beliefs generated by $\signal_\player$ and $\signalb_\player$. Since $\act_\player,\actb_\player \in B_\player$, we have $\out_{\act_\player}=\out_{\actb_\player}$, which in turn implies $\BR(\out_{\act_\player})=\BR(\out_{\actb_\player})$. Thus, 
\[
\BR(\nu^n_{\signal_\player})\cap \BR(\nu^n_{\signalb_\player})\neq\varnothing.
\]
It follows from Lemma \ref{lem:multi_agent}-(ii) that $\nu^n_{\signal_\player}=\nu^n_{\signalb_\player}$.

Since $\mact_\player\in \supp (\cord^n_\paystate)$, there must be $\corstate\in\Corstate^n$ and $\paystate\in\Paystate$ such that $\corprior^n(\corstate\vert \paystate)>0$ and $\mact_\player=\mact^{n,\corstate,\paystate}_{\player}$. Then, given a fixed $\signal_\player^{\prime}\in\Signal_{\actb_\player}^n$,
\begin{align*}
\frac{\mact_\player(\act_\player)}{\out^n(\act_\player)}& = \sum_{\signal_\player\in \Signal_\player^n} \frac{\aplan^{n}_{\player} (\act_\player|\signal_\player) \exper^{n}_{\player}(\signal_\player|\corstate,\paystate)}{\out^n(\act_\player)}
= \sum_{\signal_\player\in \Signal_{\act_\player}^n} \frac{\aplan^{n}_{\player} (\act_\player|\signal_\player) \exper^{n}_{\player}(\signal_\player|\corstate,\paystate)}{\out^n(\act_\player)}\\
& = \sum_{\signal_\player\in \Signal_{\act_\player}^n} \frac{\aplan^{n}_{\player} (\act_\player|\signal_\player) \exper^{n}_{\player}(\signal_\player|\corstate,\paystate)\nu^n(\signal_\player)}{\out^n(\act_\player)\nu^n(\signal_\player)}=\sum_{\signal_\player\in \Signal_{\act_\player}^n} \frac{\aplan^{n}_{\player} (\act_\player|\signal_\player) \exper^{n}_{\player}(\signal^\prime_\player|\corstate,\paystate)\nu^n(\signal_\player)}{\out^n(\act_\player)\nu^n(\signal_\player^\prime)}= \frac{\exper^{n}_{\player}(\signal^\prime_\player|\corstate,\paystate)}{\nu^n(\signal_\player^\prime)},
\end{align*}
where the fourth inequality follows from $\nu^n_{\signal_\player}=\nu^n_{\signalb_\player}$ for all $\signal_\player\in \Signal_{\act_\player}^n$. A similar argument delivers that, given a fixed $\signal_\player\in\Signal_{\act_\player}^n$, 
\[
\frac{\mact_\player(\actb_\player)}{\out^n(\actb_\player)}=\frac{\exper^{n}_{\player}(\signal_\player|\corstate,\paystate)}{\nu^n(\signal_\player)}.
\]
We conclude that 
\[
\frac{\mact_\player(\act_\player)}{\out^n(\act_\player)}=\frac{\exper^{n}_{\player}(\signal^\prime_\player|\corstate,\paystate)}{\nu^n(\signal_\player^\prime)}=\frac{\exper^{n}_{\player}(\signal_\player|\corstate,\paystate)}{\nu^n(\signal_\player)}=\frac{\mact_\player(\actb_\player)}{\out^n(\actb_\player)},
\]
where the third equality again follows from $\nu^n_{\signalb_\player}=\nu^n_{\signal_\player}$.\end{proof}

The next lemma relates the support of $\cord_\paystate$ (which may not be finite) to $\out$:

\begin{lemma}\label{lem:slope_day}
Let $\paystate\in\Paystate$ and $\mact \in \supp (\cord_\paystate)$. Then, for all $\player\in\Player$, $B_\player\in\PartA_\player$, and $\act_\player,\actb_\player\in B_\player$,
\begin{equation}\label{eq:may_10_23_10}
\mact_\player(\act_\player)\out(\actb_\player)=\mact_\player(\actb_\player)\out(\act_\player).
\end{equation}
\end{lemma}

\begin{proof}
Since the support correspondence is lower hemicontinuous (\citealp{Aliprantis2006}, Theorem 17.14), there exists a sequence $(\mact_\player^n)_{n=1}^\infty$ such that $\mact_\player^n\rightarrow \mact_\player$, and $\mact_\player^n\in \supp (\cord^n_\paystate)$ for all $n$. By Lemma \ref{lem:May_10_30},
\[
\mact^n_\player(\act_\player)\out^n(\actb_\player)=\mact^n_\player(\actb_\player)\out^n(\act_\player)
\]
for all $n$. Taking the limit as $n\rightarrow\infty$, we obtain (\ref{eq:may_10_23_10}).
\end{proof}

For every state $\paystate$, let $\NE_\paystate$ be the set of Nash equilibria of the complete-information game corresponding to $\paystate$:
\[
\NE_\paystate = \bigcap_{\player\in\Player}\left\{  
\mact \in \Mact: 
\mact_\player \in\arg\max_{\mactb_\player\in\Mact_\player}\sum_{\act}\util_\player(\act,\paystate)\mactb_\player(\act_\player)\prod_{\playerb\neq\player}\mact_{\playerb}(\act_\playerb)
\right\}.
\]
Note that the set $\NE_\paystate$ is closed.

\begin{lemma}\label{lem:derby_day}
For all $\paystate\in\Paystate$, $\supp (\cord_\paystate)\subseteq \NE_\paystate$.
\end{lemma}
\begin{proof}
For $\mact\in\Mact$, let $\util_\player(\mact,\paystate)$ be $\player$'s expected utility in state $\paystate$:
\[
\util_\player(\mact,\paystate)=\sum_{\act\in\Act}\util_\player(\act,\paystate)\prod_{\playerb\in \Player}\mact_\playerb(\act_\playerb).
\]
For $\mact_{-\player}\in\Mact_{-\player}$, let $\util^*_\player(\mact_{-\player},\paystate)$ be $\player$'s expected utility by best responding to $\mact_{-\player}$ in state $\theta$:
\[
\util^*_\player(\mact_{-\player},\paystate)=\max_{\act_\player\in\Act_\player}\sum_{\act_{-\player}\in\Act_{-\player}}\util_\player(\act_{\player},\act_{-\player},\paystate)\prod_{\playerb\neq\player}\mact_\playerb(\act_\playerb).
\]
Note that $\mact\in\NE_\paystate$ if and only if $\util^*_\player(\mact_{-\player},\paystate)=\util_\player(\mact,\paystate) $ for all $\player\in\Player$.

Consider now the information acquisition game $(\BGame,\itech^n)$, and take player $\player$'s perspective. Let $\bar{\exper}^{n}_{\player}$ be a fully-revealing experiment that tells the exact value of $(\corstate^{n},\paystate)$. Since information acquisition is unconstrained, the experiment $\bar{\exper}^{n}_{\player}$ is feasible. Since $(\exper^n,\aplan^n)$ is an equilibrium, player $\player$'s payoff from playing her equilibrium strategy should be larger than the payoff she obtains from deviating to $\bar{\exper}^{n}_{\player}$ and taking the optimal action given each state. Tedious but straightforward computation shows this inequality is equivalent to
\[
\sum_{\paystate} \payprior(\paystate) \int_{\Mact} \left[ 
\util^{*}_{\player}(\mact_{-\player},\paystate) - \util_\player(\mact,\paystate)
\right]
\dd\cord^{n}_{\paystate}(\mact) 
\leq \icost^{n}_{\player}(\bar{\exper}^{n}_{\player}) - \icost^{n}_{\player}(\exper^{n}_{\player}).
\]
Since information costs are bounded by $1/n$, we have $\icost^{n}_{\player}(\bar{\exper}^{n}_{\player}) - \icost^{n}_{\player}(\exper^{n}_{\player})\leq 1/n$, and so we obtain that 
\begin{align*}
& \sum_{\paystate} \payprior(\paystate) \int_{\Mact} \left[ 
\util^{*}_{\player}(\mact_{-\player},\paystate) - \util_\player(\mact,\paystate)
\right]
\dd\cord_\paystate(\mact) \\
= &\lim_{n \rightarrow \infty} \sum_{\paystate} \payprior(\paystate) \int_{\Mact} \left[ 
\util^{*}_{\player}(\mact_{-\player},\paystate) - \util_\player(\mact,\paystate)
\right]
\dd\cord^{n}_{\paystate}(\mact)\leq 0.
\end{align*}
Because $\util^{*}_\player(\mact_{-\player},\paystate) \geq \util_\player(\mact,\paystate)$ for all $\mact\in\Mact$, we obtain that $\util_\player(\mact,\paystate) = \util^{*}_{\player}(\mact,\paystate)$ for $\cord_\paystate$-almost all $\mact\in\Mact$. Since this is true for every player $\player$, we obtain that $\cord_\paystate(\NE_\paystate)=1$. Since $\NE_\paystate$ is closed, we conclude that $\supp (\cord_\paystate)\subseteq \NE_\paystate$.
\end{proof}

We are ready to show that Theorem \ref{thm:vanishing_costs_all}-(ii) holds. Fix a state $\paystate$, and denote by $\out_\paystate\in \Delta(\Act)$ the conditional distribution of $\act$:
\[
\out_\paystate(\act)=\frac{\out(\act,\paystate)}{\payprior(\paystate)}.
\]
Let $M_\paystate$ be the set of all $\mact\in\NE_\paystate$ that satisfy (\ref{eq:may_10_23_10}). Note that the set $M_\paystate$ is compact. Denote by $\co(M_\paystate)\subseteq\Delta(\Act)$ the convex hull of $M_\paystate$, that is, the set of all convex combinations of (finitely many) elements of $M_\paystate$ (here we identify $\alpha$ with the product measure induced on $\Act$). Since $M_\paystate$ is compact, $\co(M_\paystate)$ is compact. By Lemmas \ref{lem:slope_day} and \ref{lem:derby_day}, the probability measure $\cord_\paystate$ puts probability one on $M_\paystate$. Since $\out_\paystate$ is the barycenter of $\cord_\paystate$---see (\ref{eq:barycenter})---we obtain that  $\out_\paystate\in \co(M_\paystate)$ (\citealp{phelps2001lectures}, Proposition 1.2). 

Overall, for every state $\paystate$, there are $\mact_\paystate^l\in M_\paystate$, with $l=1,\ldots,L_\paystate$, such that 
\[
\out_\paystate(\act) = \sum_{l=1}^{L_\paystate} s^l_\paystate \left[\prod_{\player\in\Player}\mact_{\paystate,\player}^l (\act_\player)\right]
\]
with $s^l_\paystate \geq 0$ for all $l\in L_\paystate$, and $\sum_{l=1}^{L_\paystate} s^l_\paystate=1$. It follows from a standard cake-cutting argument that, without loss of generality, we can assume that $s^l_\paystate$ and $L_\paystate$ are independent of $\paystate$. Given $s^l:=s^l_\paystate$ and $L:=L_\paystate$, we define $\outb_l\in \Delta(\Act)$ by
\[
\outb_l (\act,\paystate) = \payprior(\paystate)\prod_{\player\in\Player}\mact_{\paystate,\player}^l (\act_\player),
\]
and we notice that 
\[
\out(\act,\paystate)\sum_{l=1}^L s^l \outb_l (\act,\paystate).
\]
Each $\outb_l$ is a $(\PartA,\out)$-decomposable complete-information Nash equilibrium. We conclude that Theorem \ref{thm:vanishing_costs_all}-(ii) holds.

\end{spacing}

\end{document}